\documentclass[11pt]{book}
\usepackage{geometry}                    
\usepackage{graphicx}
\usepackage{amssymb}
\usepackage{amsmath}
\usepackage{epstopdf}
\usepackage{amsthm}

\usepackage{etoolbox}
\makeatletter
\patchcmd{\chapter}{\if@openright\cleardoublepage\else\clearpage\fi}{}{}{}
\makeatother

\usepackage{algorithmic}
\usepackage{color}

\newtheorem{lemma}{ Lemma}[section]
\newtheorem{proposition}{ Proposition}[section]
\newtheorem{corollary}{ Corollary}
\newtheorem{remark}{ Remark}

\newtheorem{algorithm}{Algorithm}
\newtheorem{theorem}{ Theorem}[section]

\newtheorem{definition}{Definition}[section]
\newtheorem{example}{Example}[section]
\newcommand{\HRule}{\rule{\linewidth}{0.5mm}}
\begin{document}

\begin{titlepage}
\begin{center}


\textsc{\LARGE School of Mathematics and Statistics, University of Sydney\\[1.5cm]}


\HRule \\[0.4cm]
{ \huge \bfseries Stochastic Analysis Seminar on Filtering Theory \\[0.4cm] }

\HRule \\[1.5cm]

\textsc{\Large \emph{Author:}\\
Andrew \textsc{Papanicolaou}}\\alpapani@maths.usyd.edu.au\\[0.5cm]

\vfill


\end{center}
{\large These notes were originally written for the Stochastic Analysis Seminar in the Department of Operations Research and Financial Engineering at Princeton University, in February of 2011. The seminar was attended and supported by members of the Research Training Group, with the author being partially supported by NSF grant DMS-0739195. }

\end{titlepage}
\tableofcontents
\chapter{Hidden Markov Models}

\noindent We begin by introducing the concept of a Hidden Markov Model (HMM). Let $t\in[0,\infty)$ denote time, and consider a Markov process $X_t$ which takes value in the state-space $\mathcal S$. We assume that the distribution function of $X_t$ has either a mass or a density, and we denote mass/density with $p_t(x)$ such that
\[p_t(x) =\frac{d}{dx}\mathbb P(X_t\leq x)\qquad\hbox{for densities},\]
\[p_t(x) =\mathbb P(X_t=x)\qquad\hbox{for masses}\]
for any $t\geq 0$ and $\forall x\in\mathcal S$. The generator of $X_t$ is the operator $Q$ with domain $\mathcal B(Q)$, such that for any bounded function $g(x)\in\mathcal B(Q)$ we have a backward equation,
\[\frac{\mathbb E[g(X_{t+\Delta t})|X_t=x] -g(x)}{\Delta t}\rightarrow Qg(x)\qquad\hbox{as }\Delta t\searrow 0\]
for any $x\in\mathcal S$. Provided that regularity conditions are met, the adjoint leads to the forward equation,
\[\frac{d}{dt}p_t(x) = Q^*p_t(x).\]
\begin{example}
If $\mathcal S=\{x_1,\dots,x_m\}$ (a finite space) and the operator $Q$ is a jump-intensity matrix such that $Q_{ji}\geq 0$ for all $i\neq j$ and $\sum_{i\neq j}Q_{ji} = -Q_{jj}$ for all $j\in\{1,\dots,m\}$. The forward equation is then
\[\frac{d}{dt}p_t(x_i) = \sum_{j=1}^mp_t(x_j)Q_{ji}.\]
\end{example}
\begin{example} If $\mathcal S=\mathbb R$ and $X_t$ is an It\^o process such as
\[dX_t=a(X_t)dt+\sigma dB_t,\]
then $Q =\mathcal L= \frac{\sigma^2}{2}\frac{\partial^2}{\partial x^2}\cdot+a(x)\frac{\partial}{\partial x}\cdot$ is the generator. The backward equation is then
\[\frac{d}{dt}\mathbb Eg(X_t)=\mathbb E\mathcal Lg(X_t) = \frac{\sigma^2}{2}\mathbb Eg''(X_t)+\mathbb Ea(X_t)g'(X_t)\]
for any bounded function $g\in C^2(\mathbb R)$. Provided that $a(x)$ and the initial distribution satisfy some conditions for regularity, there is also a forward equation given by the adjoint
\[\frac{\partial}{\partial t}p_t(x)=\mathcal L^*p_t(x) = \frac{\sigma^2}{2}\frac{\partial^2}{\partial x^2}p_t(x)-a(x)\frac{\partial}{\partial x}p_t(x)-a'(x)p_t(x)\]
for any $x\in\mathbb R$.
\end{example}
In addition to $X_t$, there is another process $Y_t$ that is a noisy function of $X_t$. The process $Y_t$ can given by the SDE
\[dY_t = h(t,X_t)dt+\gamma(t,X_t)dW_t\]
where $W_t$ is an independent Wiener process, or can be given discretely,
\[Y_{t_k} =  h(t_k,X_{t_k})+\gamma (W_{t_k}-W_{t_{k-1}})\]
where $(t_k)_k$ is a set of discrete times at which data is collected. 

As a pair, $(X_t,Y_t)$ are a Markov chain. The process $X_t$ is of primary interest to us and is referred to as the `signal' process, however it is not observable. Instead, the process $Y_t$ is in some way observable and so we call it the `measurement.' Hence, $(X_t,Y_t)$ is an HMM and the goal is to calculate estimates of $X_t$ that are optimal in a posterior sense given observations on $Y_t$. 

\section{Basic Nonlinear Filtering}
Let $\mathcal F_t^Y$ denote the filtration generated by the observations on $Y$ up to time $t$. The optimal posterior estimate of $X_t$ in terms of mean-square error (MSE) is
\[\widehat X_t = \mathbb E[X_t|\mathcal F_t^Y]=\arg\min_{f\in\mathcal F_t^Y}\mathbb E(f-X_t)^2 .\]
\begin{proposition}$\widehat X_t$ is the unique $\mathcal F_t^Y$-measurable minimizer of MSE.
\end{proposition}
\begin{proof} Let $f$ be another $\mathcal F_t^Y$-measurable estimate of $X_t$. Then
\[MSE(f) = \mathbb E(f-X_t)^2 = \mathbb E(f-\widehat X_t+\widehat X_t-X_t)^2\]

\[= \mathbb E(f-\widehat X_t)^2+2\mathbb E(f-\widehat X_t)(\widehat X_t-X_t)+\mathbb E(\widehat X_t-X_t)^2\]

\[= \mathbb E(f-\widehat X_t)^2+2\mathbb E\left[(f-\widehat X_t)\mathbb E[(\widehat X_t-X_t)|\mathcal F_t^Y]\right]+\mathbb E(\widehat X_t-X_t)^2\]

\[= \mathbb E(f-\widehat X_t)^2+\mathbb E(\widehat X_t-X_t)^2\]

\[\geq \mathbb E(\widehat X_t-X_t)^2=MSE(\widehat X_t)\]
with equality holding iff $f=\widehat X_t$ almost everywhere.

\end{proof}
The filtering measure is defined as
\[\pi_t(\mathcal A) = \mathbb P(X_t\in \mathcal A|\mathcal F_t^Y)\]
for any Borel set $\mathcal A$, and for any measurable function $g$ 
\[\hat g_t = \mathbb E[g(X_t)|\mathcal F_t^Y].\]
\begin{remark}
There are also smoothing and prediction distributions. When posteriors have density, we write
\[\pi_{t|T}(dx)=\mathbb P(X_t\in dx|\mathcal F_T^Y).\]
We say that $\pi_{t|T}$ is the \textbf{smoothing} density if $T>t$ and the \textbf{prediction} density if $T<t$. Smoothing requires significantly more calculation to compute, but the prediction simply requires us to solve the forward equation for $X_t$ in the interval $[t,T]$ with initial condition $\pi_t$.

\end{remark}
\begin{example}\textbf{Filtering With Discrete Observations; The Bayesian Case.}  Let $\mathcal S$ be a countable state-space, let $h(x)$ be a  known nonlinear function, let $\gamma(x)=\gamma>0$, and for $k=0,1,2,3,\dots$ let there be specific times $t_k$ at which observations are collected on $Y_t$. At each time $t_k$, let $Y_k=Y_{t_k}$ and denote the history of measurements up to time $t_k$ as $Y_{0:k}=\{Y_0,Y_1,\dots,Y_k\}$. Denote $X_k=X_{t_k^-}$. 

Consider the following discrete differential:
\[Y_{k+1}=Y_k+h(X_{k+1})\Delta t_k+\gamma\Delta W_k\]
where $\Delta t_k = t_{k+1}-t_k$ and $\Delta W_k=W_{t_{k+1}}-W_{t_k}$.

Using Bayes rule, we find the the filtering distribution has a mass function
\[\pi_k(x) = \mathbb P(X_k=x|Y_{0:k})\]
for any $x\in\mathcal S$. It can be written recursively as follows,
\begin{eqnarray}
\label{eq:bayesFilter}
\pi_{k+1}(x) &=& \frac{1}{c_{k+1}}\psi_{k+1}(x)e^{Q^*\Delta t_k}[\pi_k](x)
\end{eqnarray}
where $e^{Q^*\Delta t_k}[~\cdot~]$ is the kernel of $X_t$'s forward transition probabilities, $\psi_{k+1}(x)$ is the likelihood ratio of $\{X_{k+1}=x,Y_k\}$ given $Y_{k+1}$,
\[\psi_{k+1}(x) = \exp\left\{-\frac{ h^2(x)\Delta t_k-2(Y_{k+1}-Y_k)h(x)\Delta t_k}{2\gamma^2}\right\}\]
and $c_{k+1}=\int\psi_{k+1}(x)e^{Q^*\Delta t_k}\pi_k(dx)$ is a normalizing constant. Equation (\ref{eq:bayesFilter}) can be shown to hold true through a use of Bayes formula and by the independence properties of the HMM.\\
\begin{proof}\textbf{(of equation (\ref{eq:bayesFilter}))} Regarding the likelihood function of $\{X_{k+1}=x,Y_k\}$ given $Y_{k+1}$, it is
 \[L(Y_{k+1}|Y_k,X_{k+1}=x)\propto\exp\left\{-\frac{1}{2}\left(\frac{Y_{k+1}-Y_k-h(x)\Delta t_k}{\gamma\Delta t_k}\right)^2\right\}\]
 
 \[= \exp\left\{-\frac{(Y_{k+1}-Y_k)^2+h^2(x)\Delta t_k-2h(x)(Y_{k+1}-Y_k)}{2\gamma^2}\right\}\]
 and since we are only interested in how this likelihood varies with $x$, we can remove the terms that do not have $x$ in them,
 
 \[\propto \exp\left\{-\frac{h^2(x)\Delta t_k-2h(x)(Y_{k+1}-Y_k)}{2\gamma^2}\right\}=\psi_{k+1}(x),\]
 and so it is in fact the likelihood-ratio of $\{X_{k+1}=x\}$.
 
Now, by Bayes formula we have,
\[\pi_{k+1}(x) = \frac{\mathbb P(X_{k+1}=x;Y_{0:k+1})}{\mathbb P(Y_{0:k+1})}\]

\[=\frac{\mathbb P(Y_{k+1}|X_{k+1}=x;Y_{0:k})\mathbb P(X_{k+1}=x;Y_{0:k})}{\mathbb P(Y_{0:k+1})}=\frac{\psi_{k+1}(x)\mathbb P(X_{k+1}=x;Y_{0:k})}{\mathbb P(Y_{0:k+1})}\]

\[=\frac{\psi_{k+1}(x)\sum_{v\in\mathcal S}\mathbb P(X_{k+1}=x,X_k=v;Y_{0:k})}{\mathbb P(Y_{0:k+1})}\]

\[=\frac{\psi_{k+1}(x)\sum_{v\in\mathcal S} \mathbb P(X_{k+1}=x|X_k =v)\mathbb P(X_k\in dv;Y_{0:k})}{\mathbb P(Y_{0:k+1})}\]

\[=\frac{\psi_{k+1}(x)\sum_{v\in\mathcal S} \mathbb P(X_{k+1}=x|X_k =v)\mathbb P(X_k\in dv|Y_{0:k})}{\mathbb P(Y_{k+1}|Y_{0:k})}\]

\[=\frac{\psi_{k+1}(x)\sum_{v\in\mathcal S} \mathbb P(X_{k+1}=x|X_k =v)\pi_k(dv)}{\mathbb P(Y_{k+1}|Y_{0:k})}\]

\[=\frac{\psi_{k+1}(x)e^{Q^*\Delta t_k}\pi_k(dv)}{\mathbb P(Y_{k+1}|Y_{0:k})}\]
and clearly, $\mathbb P(Y_{k+1}|Y_{0:k})$ is the integral of numerator of the last line over $x$. 
\end{proof}
\end{example}

The Bayesian filter is an essential tool for numerical computations of nonlinear filtering. The Kalman filter (see Jazwinski \cite{jazwinski}) is also an important tool, but it only applies to linear Gaussian models or models that are well-approximated as such. The contemporary way to compute nonlinear filters is via Monte Carlo with a particle filter (which we'll talk about in a later section). In continuous time, approximating filters based on discretization of the differential $dY_t$ have been shown to converge as $\Delta t\searrow 0$ for a certain class of filtering problems, but we must be able to approximate the law of $X_t$, and $h$ must be bounded (see Kushner \cite{kushner08}).

In summary, an HMM consists of a pair of process $(X_t,Y_t)$ where $X_t$ is an unobserved signal which is a Markov process, while $Y_t$ is an observable process that depends on $X_t$ through a system of known functions and known parameters. We use filtering to compute the posterior distribution of $X_t$ given $\mathcal F_t^Y$.

\chapter{Filtering and the VIX}

\section{Stochastic Volatility} 
\label{sec:intoSVM}
Consider an equity model with stochastic volatility,
\[dS_t = \mu S_tdt+f(X_t)S_tdW_t\]
where $S_t$ is the price of a stock, the function $f(x)$ is known and $X_t$ is a hidden Markov process. For instance, the Heston model, where $f(x) = \sqrt{X}$ and 
\[dX_t = \kappa(m-X_t)dt+\gamma\sqrt{X_t}dB_t\]
where we model the volatility leverage effect by saying that $\frac{1}{t}\mathbb EB_tW_t=\rho $ with $\rho\in[-1,0)$.

In general, if we observe a continuum of prices, then $f(X_t)$ is measurable with respect to the filtration generated by $\{S_\tau:\tau\leq t\}$. Let $Y_t=\log S_t$, and notice that 
\[dY_t = \left(\mu-\frac{1}{2}f^2(X_t)\right)dt+f(X_t)dW_t.\]
For a fixed $t>0$, let $(t_k)_k$ be a partition of $[0,t]$, then the quadratic variation of $Y$ is the cumulative variance
\[[Y]_t = \lim_{\|P\|\searrow0}\sum_k(\Delta Y_{t_k})^2=\int_0^tf^2(X_\tau)d\tau\qquad\hbox{in probability}\]
where $\|P\|=\sup_k(t_{k+1}-t_k)$. Clearly, then $\int_0^tf^2(X_\tau)d\tau$ is $\mathcal F_t^Y$-measurable, and if $f(X_t)$ is a continuous process we have
\[\frac{d}{dt}[Y]_t = f^2(X_t)\]
is also $\mathcal F_t^Y$-measurable. So at the very least (e.g. for $f(X_t)$ a continuous process) volatility is observable for almost everywhere $t$, and $X_t$ is observable if $f^{-1}$ exists. 

Nonetheless, it is still beneficial to have a Markov structure for $X_t$ so that we can price derivatives on $S_t$. For instance, in the example by Elliot \cite{elliot2007}, the Black-Scholes price of a European call option in the presence of Markovian volatility, when $\rho = 0$ is
\[C(t,S_t,X_t;T,K)=\mathbb E^*[C_{BS}(t, S_t;T,K,Z_{[t,T]})|X_t,S_t] =\mathbb E^*[C_{BS}(t, S_t;T,K,Z_{[t,T]})|\mathcal F_t^Y]\]
where $Z_{[t,T]} = \frac{1}{T-t}\int_t^Tf^2(X_s)ds$, and $\mathbb E^*[~\cdot~]$ is the market's pricing measure. Given $X_t$ and the parameters of $X$'s dynamics under the market measure, we can compute the expected return of the call option either explicitly or through Monte Carlo. 

If $X_t$ is not known (i.e. observations are discrete) then we could take a filtering expectation,
\[C(t,S_t;T,K)\stackrel{?}{=}\int C(t,S_t;T,K,x)\pi_t(dx)\]
but such a price would be significantly biased if the market placed any premium on volatility. This option pricing formula exemplifies the challenge of interpreting the filter in financial math.

\section{The VIX}
Filtering can be used to extract the risk-premium placed on volatility by the market. Let $S_t$ be the S\&P500 index. For any time $t$ and some time-window $T>0$, the VIX index is the square root of the market's prediction of average variance during $[t,t+T]$,
\[V_t= \mathbb E^*\left[\frac{1}{T}\int_t^{t+T}f^2(X_s)ds\Big|\mathcal F_t^m\right]\]
where $\mathbb E^*[~\cdot~]$ is the market's pricing measure, and $\mathcal F_t^m$ is the filtration generated by \textit{all} the information in the market (i.e. $\mathcal F_t^Y\subset\mathcal F_t^m$). The process $V_t$ is the `fair' price of a variance swap whose floating leg is the realized variance,
\[RV_{[t,t+T]} =\lim_{\|P\|\searrow 0}\frac{1}{T}\sum_{k}(\Delta Y_{t_k})^2=_p\frac{1}{T}\int_t^{t+T}f^2(X_s)ds\]
where $(t_k)_k$ is a partition of $[t,t+T]$ (i.e. $t=t_0<t_1<\dots t_N=t+T$, with $\|P\|=\sup_k(t_{k+1}-t_k)$ going to zero as $N$ gets large). For a pre-specified notional amount, the payoff of a variance swap is
\[\hbox{notional}\times\left(RV_{[t,t+T]} - V_t\right).\]

\subsection{The VIX Formula}
For diffusion models without jumps, it was shown by Demeterfi et al \cite{derman} that a portfolio of out-of-the-money call and put contracts along with a short position in a futures contract replicates the VIX. It was shown by Carr et al \cite{carr2002,carr2005} that models with jumps can be approximated by the same setup. The following lemma derives the strategy for diffusions:
\begin{lemma} Let $T$ denote the life of the contract, let $F_{t,T}$ denote the future price on $S_{t+T}$ at time $t\geq 0$, and let $r$ be the rate so that $F_{t,T}=S_te^{rT}$. If $S_t$ is purely a diffusion process (i.e. has no jump terms in its differential), then the market's expectation of future realized variance is

\begin{equation}
\label{eq:VIX}
V_t=\frac{2e^{rT}}{T}\left(\int_{K\leq F_{t,T}}P_t(K,T)\frac{dK}{K^2}+\int_{K\geq F_{t,T}}C_t(K,T)\frac{dK}{K^2}     \right)
\end{equation}
where $P_t(K,T)$ and $C_t(K,T)$ denote the price of a put and a call option at time $t$ with strike $K$ and time to maturity $T$, respectively.
\end{lemma}
\begin{proof} Under the market measure, there is a Wiener process $dW_t^*$ such that the returns on the stock satisfy
\[\frac{dS_t}{S_t} = rdt+f(X_t)dW_t^*\]
and the log-price satisfies
\[d\log(S_t) = dY_t = \left(r-\frac{1}{2}f^2(X_t)\right)dt+f(X_t)dW_t^*.\]
Integrating the returns and the log-price separately and then subtracting, we can eliminate all randomness to get the cumulative variance,
\[\int_t^{t+T}\frac{dS_\tau}{S_\tau} -\int_t^{t+T}dY_\tau = \frac{1}{2}\int_t^{t+T}f^2(X_\tau)d\tau\]
which shows us that the realized variance satisfies the following:
\[RV_{[t,t+T]}=\frac{2}{T}\left(\int_t^{t+T}\frac{dS_\tau}{S_\tau}-\log(S_{t+T}/S_t)\right)\]

\[=\frac{2}{T}\left(\int_t^{t+T}\frac{dS_\tau}{S_\tau}-\log(S_{t+T}/F_{t,T})-\log(F_{t,T}/S_t))\right).\qquad\qquad(*)\]
Then, through some simple calculus, we see that
\[-\log(S_{t+T}/F_{t,T}) \]

\[= -\frac{S_{t+T}-F_{t,T}}{F_{t,T}}+\int_{K\leq F_{t,T}}(K-S_{t+T})^+\frac{dK}{K^2}+\int_{K\geq F_{t,T}}(S_{t+T}-K)^+\frac{dK}{K^2}.\]
Plugging this expression for $-\log(S_{t+T}/F_{t,T})$ into $(*)$ we see that the realized variance can be written as the payoffs of several contracts and a continuum of puts and calls that were out-of-the money at time $t$:
\[RV_{[t,t+T]}=\frac{2}{T}\Bigg(\int_t^{t+T}\frac{dS_\tau}{S_\tau} -\frac{S_{t+T}-F_{t,T}}{F_{t,T}}-\log(F_{t,T}/S_t)\]

\[+\int_{K\leq F_{t,T}}(K-S_{t+T})^+\frac{dK}{K^2}+\int_{K\geq F_{t,T}}(S_{t+T}-K)^+\frac{dK}{K^2}\Bigg).\]
Taking expectation of both with respect to the market measure, the noise in $\int\frac{dS}{S}$ vanishes, and since both $\log(F_{t,T}/S_t)=rT$ and $\mathbb E^*S_{t+T}=F_{t,T}$, we have

\[V_t = \mathbb E^*[RV_{t,T}|\mathcal F_t^m]\]

\[=\frac{2}{T}\Big(\int_{K\leq F_{t,T}}\mathbb E^*[(K-S_{t+T})^+|\mathcal F_t^m]\frac{dK}{K^2}+\int_{K\geq F_{t,T}}\mathbb E^*[(S_{t+T}-K)^+|\mathcal F_t^m]\frac{dK}{K^2}\Big)\]
and if we multiply and divide the RHS by $e^{-rT}$ we get the result.
\end{proof}

\subsubsection{Risk Premium}
Under the market measure, the expected returns on a variance swap are zero. However, statistically speaking, variance swaps exhibit a slight bias against the holder of the contract. In other words, the person who receives $\hbox{\textit{notional}}\times(RV_{[t,t+T]}-V_t)$ at time $t+T$ will have an average return that is slightly negative. But because of the volatility leverage effect, the variance swap has the potential to provide relief in the form of a positive cash flow when volatility is high and equities are losing. When quoted in the market, the VIX is quoted as the square-root of $V_t$ in percentage points, and is a mean-reverting process that drifts between 15\% and 50\%. The VIX has the nickname `the investor fear gauge', as it should because it is composed primarily of out-of-the-money options which means that there is an increase in crash-a-phobia whenever the VIX increases. 

To get a more precise understanding of investors' fears, it would be nice to remove any actual increases in volatility and merely examine the bias in the market measure's prediction of variance. In other-words, we would like to predict variance in the physical measure, and then compare it with the VIX's prediction to gain a sense of how much of a premium is being placed on risk, or how much fear there is out there.

If we can identify an HMM whose observable component generates $\mathcal F_t^m$, then the market's price of volatility risk (aka the risk-premium) is
\[\mathcal {RP}_t = V_t -  \mathbb E\left[\frac{1}{T}\int_t^{t+T}f^2(X_s)ds\Big|\mathcal F_t^m\right].\]

\begin{remark}
Carr and Wu \cite{carr2006} point out that if we define the martingale change of measure with the Radon-Nykodym derivative $M_t$, then
\[V_t = \mathbb E_t^*[RV_{t,t+T}] = \frac{ \mathbb E_t[M_{t+T}RV_{t,t+T}] }{ \mathbb E_t[M_{t+T}]}=\mathbb E_t[RV_{t,t+T}]+cov_t\left(\frac{M_{t+T}}{\mathbb E_t[M_{t+T}]},RV_{t,t+T}\right)\]
where $\mathbb E_t[~\cdot~]= \mathbb E[~\cdot~|\mathcal F_t^m]$ and $\mathbb E_t^*[~\cdot~]= \mathbb E^*[~\cdot~|\mathcal F_t^m]$, which leads to an expression for the risk-premium, and is $\mathcal F_t^m$-measurable provided that $M_{t+T}$ is known.

\end{remark}
\subsubsection{More on the VIX}
Historically, indices like the S\&P500 exhibit contrary motion with volatility. In particular, periods of high volatility often coincide with bearish markets. Whaley \cite{whaley} describes how tradable volatility assets, such as the VIX, provided market makers with new ways to hedge the options they had written. For instance, a market maker who was short a portfolio of options has always been able to go long in some other types of contracts to reduce the portfolio's Delta to almost zero, meaning that the portfolio would not be hugely affected by small changes in the value of the underlying. When VIX was introduced, it allowed the same market maker the opportunity to also hedge the Vega of his/her short position in options. Prior to volatility hedging instruments, a market maker might be exposed to the rising options prices that occur as volatility increases. However, instruments such VIX futures and calls on VIX futures changed all that, as a short position in options could then have its Vega reduced to almost nothing with the appropriate number of VIX contracts.

Traders use VIX futures and options to hedge in times of uncertain volatility in the SPX or the SPY.\footnote{SPX is the S\&P500; SPY is the tracking stock for the S\&P 500.} When volatility traders notice a spread between the VIX and VIX futures, they realize that a correction is probable. Therefore, if VIX is trading higher than the VIX futures price, then buying calls in both SPY and VIX will make money because either a) VIX goes down and the SPY goes up which places the SPY call in-the-money, or b) the VIX stays up and the VIX futures close the gap which places the VIX call in-the-money. A similar strategy with puts on SPY and VIX can be devised when the VIX is significantly lower than the VIX futures price. The rule of thumb is: provided that the spread between VIX and VIX futures is wide enough, the correction in VIX will create enough change in the market that one of these straddles can cover its initial cost and provide some profit to the investor. The operative word in the strategy mentioned in this paragraph is \textit{spread}, which is precisely what we are looking for as we filter for the risk-premium.

\chapter{Stochastic Volatility Filter for Heston Model}

\noindent With discrete observations, we derive a Bayesian stochastic volatility filter for a Heston model. Let $Y_t$ denote the log-price of the equity, and let $\sqrt{X_t}$ denote volatility. The dynamics of the processes are,
\begin{eqnarray}
\label{eq:heston}
dX_t &=& \kappa(\bar X-X_t)dt+\gamma\sqrt{X_t}dB_t\\
\label{eq:logReturns}
dY_t& =& \left(\mu-\frac{1}{2}X_t\right)dt+\sqrt{X_t}\left(\rho dB_t+\sqrt{1-\rho^2}dW_t\right)
\end{eqnarray}
and the interpretation of the model parameters is as follows:
\begin{eqnarray}
\nonumber
\bar X&=&\hbox{the long-time average of $X_t$}\\
\nonumber
\kappa&=&\hbox{the rate of mean-reversion (on $X_t$)}\\
\nonumber
\gamma&=&\hbox{volatility of volatility}\\
\nonumber
\rho&=&\hbox{models volatility leverage effect when $\rho\in(-1,0)$}\\
\nonumber
\mu&=&\hbox{the mean-rate of returns on the stock}
\end{eqnarray}
Certain restrictions on the parameters need to be put in place, such as the Feller condition: $\gamma\leq \sqrt{2\kappa \bar X}$ to insure the $X_t$ is well-defined (see chapter volatility time scales in \cite{FPS00}). There are times $(t_n)_n$ for which the process is actually observed, and if our model is correct\footnote{`correct' not only means that the processes follow these parametric SDEs, but it also means that $\{W_t\}$ and $\{B_t\}$ are idiosyncratic noises that are endemic to this system and not correlated with other data in the market.} then $\mathcal F_t^m=\mathcal F_t^Y = \sigma\{Y_{t_n}:t_n\leq t\}$. 
\section{The Filter}
\label{sec:volFilter}
We showed in previous lectures that $X_t$ is measurable when $Y$ is observed continuously. It is also straight forward to show that $\mathbb E[g(X_t)|\mathcal F_t^Y]\rightarrow g(X_t)$ as the partition of $[0,t]$ shrinks to zero. 
\begin{lemma} 
\label{lemma:discreteToCont}
For any $N\in\mathbb Z^+$ let there be a partition of $[0,t]$ into $N$-many points, and let $\mathcal F_t^N$ denote the filtration generated by the observations of $Y$ at the partitioned points $(i.e.~\mathcal F_t^N = \sigma\{\{Y_{t_n}\}_{n=0}^N\})$. If the filtrations are increasing with $N$, then  in the context of the model given by (\ref{eq:heston}) and (\ref{eq:logReturns}) we have
\[\mathbb E[g(X_t)|\mathcal F_t^N]\rightarrow g(X_t)\qquad\hbox{a.s.}\]
as $N\nearrow\infty$ for any function $g(x)$.
\end{lemma}
\begin{proof}
Let $\mathcal F_t^Y=\sigma\{Y_s:s\leq t\}$. We assume that the filtrations are increasing with $N$, and they are certainly bounded by $\mathcal F_t^Y$,
\[\mathcal F_t^N\subset\mathcal F_t^{N+1}\subset\dots\dots\subseteq\mathcal F_t^Y.\]
Therefore there exists $\mathcal F_t^*$ such that $\bigvee_{N=1}^\infty\mathcal F_t^N=\mathcal F_t^*$. Now because the path of $Y_t$ is continuous on all sets of non-zero probability, the information contained in $\mathcal F_t^*$ is enough to measure the event $\{Y_t=y\}$ for any $y\in\mathbb R$ even if $t$ is not a partition point for any finite $N$. Therefore, $\mathcal F_t^* = \mathcal F_t^Y$ a.s., and by the L\'evy 0-1 law we have
\[\lim_N\mathbb E[g(X_t)|\mathcal F_t^N]=\mathbb E[g(X_t)|\mathcal F_t^*] = \mathbb E[g(X_t)|\mathcal F_t^Y]= g(X_t)\]
which proves the lemma.

\end{proof}
From lemma \ref{lemma:discreteToCont} we know that our posterior estimates are consistent as we refine the partition. Now we need to determine the filter for a specific partition. From here forward consider a specific finite partition of the time domain, and we only consider the filter at times when data has arrived. For ease in notation we let $X_n=X_{t_n}$, $Y_n=Y_{t_n}$ and $\mathcal F_n^Y=\mathcal F_{t_n}^Y$.
\begin{proposition} 
\label{prop:SVMfilter}
Let $(Y_t,X_t)$ be the price and volatility processes in the Heston model from (\ref{eq:heston}) and (\ref{eq:logReturns}), and assume the Feller condition $\gamma^2\leq 2\bar X\kappa$. Then there is a kernel that gives $X$'s transition density, $e^{Q^*\Delta t}(x|v)=\frac{d}{dx}\mathbb P(X_{t+\Delta t}\leq x|X_t=v)$ for any $x,v\in\mathbb R^+$, and the filtering distribution for $X_n$ at observation time $t_n=n\Delta t$ has a density. This density is given recursively as
\begin{align}
\nonumber
&\pi_n(x)\\
\label{eq:piDensity}
&=\frac{1}{c_n}\int \mathbb E\left[\mathbb L(y|(X_u)_{\{t_{n-1}\leq u\leq t_n\}},Y_{n-1})\Big|X_n=x,X_{n-1}=v,Y_{n-1}\right]e^{Q^*\Delta t}(x|v)\pi_{n-1}(dv)\Bigg|_{y=Y_n}\ ,
\end{align}
for almost-everywhere $x\in\mathbb R^+$, where $c_n$ is a normalizing constant, and $\mathbb L$ is the likelihood of the any path $(x_u)_{\{t_{n-1}\leq u\leq t_n\}}$ given observations $Y_n$ and $Y_{n-1}$, and is given by 
\[\mathbb L(y|(x_u)_{\{t_{n-1}\leq u\leq t_n\}},Y_{n-1}) =\frac{\exp\left\{ -\frac{1}{2}\left(\frac{ (y-Y_{n-1})-\left(\mu\Delta t-.5\int_{t_{n-1}}^{t_n}x_udu\right)-\rho\xi_n(x)}{\sqrt{(1-\rho^2)\int_{t_{n-1}}^{t_n}x_udu}}\right)^2 \right\}}{ \sqrt{(1-\rho^2)\int_{t_{n-1}}^{t_n}x_udu}}\]
with 
\[\xi_n(x) = \frac{1}{\gamma}\left\{ \Delta x_{n-1}-\kappa\left(\bar X\Delta t-\int_{t_{n-1}}^{t_n}x_udu\right)\right\}\ .\]
\end{proposition}

\begin{proof} Given the Feller condition, the CIR process $dX_t=\kappa(\bar X-X_t)dt+\gamma\sqrt{X_t}dB_t$ is well-known to have a transition density that can be written in terms of a modified Bessel function (see \cite{aitSahalia1999}), and so $\Gamma^{\Delta t}(\cdot|v)$ is a smooth density function for all $v\geq 0$. Furthermore, it was shown in \cite{DY2002} that $(Y_t,X_t)$ has a smooth transition density function, that is, 
\[\mbox P^{\Delta t}(y,x|s,v)\doteq \frac{\partial^2}{\partial y\partial x}\mathbb P(Y_n\leq y,X_n\leq x|Y_{n-1}=s,X_{n-1}=v)\]
is smooth for $x> 0$, $y> 0$, and $\Delta t>0$, and does not collect mass at $x=0$ or $y=0$. Hence, the filter has a density that can be written using Bayes rule:
\[\pi_n(x) = \frac{\int\mbox P^{\Delta t}(Y_n,x|Y_{n-1},v)\pi_{n-1}(v)dv}{\int\hbox{[numerator]}dx}\ ,\]
where we don't need to assume smoothness of $\pi_{n-1}$ because it is smoothed by its convolution with $\mbox P^{\Delta t}$ in the $dv$-integral.

Now, from equation (\ref{eq:logReturns}) we notice the following:
\begin{align*}
Y_n-Y_{n-1} &=  \mu\Delta t-\frac 12\int_{t_{n-1}}^{t_n}X_udu+\rho\int_{t_{n-1}}^{t_n}\sqrt{X_u}dB_u+\sqrt{1-\rho^2}\int_{t_{n-1}}^{t_n}\sqrt{X_u}d W_u\\
&=_d  \mu\Delta t-\frac 12\int_{t_{n-1}}^{t_n}X_udu+\rho\int_{t_{n-1}}^{t_n}\sqrt{X_u}dB_u+\sqrt{(1-\rho^2)\int_{t_{n-1}}^{t_n}X_udu}~~\mathcal Z
\end{align*}
where ``$=_d$" signifies equivalence in distribution, and $\mathcal Z$ is another independent standard normal random variable. This means that conditional on the path $(X_u)_{t_{n-1}\leq u\leq t_n}$ and $S_{n-1}$, 

\[\frac{Y_n-Y_{n-1} -\left(\mu\Delta t-\frac 12\int_{t_{n-1}}^{t_n}X_udu+\rho\int_{t_{n-1}}^{t_n}\sqrt{X_u}dB_u\right)}{\sqrt{(1-\rho^2)\int_{t_{n-1}}^{t_n}X_udu}}=_d\mathcal Z\ .\]
Then noticing $\xi_n$ evaluated at $(X_u)_{t_{n-1}\leq u\leq t_n}$ is the the same as $\xi_n\left(X\right)=\int_{t_{n-1}}^{t_n}\sqrt{X_u}dB_u$, it follows that
\[\frac{Y_n-Y_{n-1} -\left(\mu\Delta t-\frac 12\int_{t_{n-1}}^{t_n}X_udu+\rho\xi_n\left(X\right)\right)}{\sqrt{(1-\rho^2)\int_{t_{n-1}}^{t_n}X_udu}}=_d\mathcal Z\ .\]
This shows the likelihood of the path $(X_u)_{t_{n-1}\leq u\leq t_n}$ given $S_{n-1}$ and $S_n=y$ is in fact the function $\mathbb L$.

Finally, given Bayes rule for the density $\pi_n$, the expression in equation (\ref{eq:piDensity}) displays the filter using a probabilistic representation of the transition density:

\begin{align*}
&\mbox P^{\Delta t}(y,x|Y_{n-1},v)\\
&=\frac{\partial}{\partial y}\int_0^y\mbox P^{\Delta t}(z,x|Y_{n-1},v)dz \\
& =\frac{\partial}{\partial y}\mathbb P(Y_n\leq y|X_n=x,X_{n-1}=v,Y_{n-1})\Gamma^{\Delta t}(x|v)\\
& =\frac{\partial}{\partial y}\mathbb E\left\{\mathbf 1_{Y_n\leq y}\Big |X_n=x,X_{n-1}=v,Y_{n-1}\right\}\Gamma^{\Delta t}(x|v)\\
&=\frac{\partial}{\partial y}\mathbb E\left[\mathbb E\left\{\mathbf 1_{Y_n\leq y}\Big|( X_u)_{\{t_{n-1}\leq u\leq t_n\}},Y_{n-1}\right\}\Big| X_n=x, X_{n-1}=v,Y_n,Y_{n-1}\right]\Gamma^{\Delta t}(x|v)\\
&=\mathbb E\left[\frac{\partial}{\partial y}\mathbb E\left\{\mathbf 1_{Y_n\leq y}\Big|( X_u)_{\{t_{n-1}\leq u\leq t_n\}},Y_{n-1}\right\}\Big| X_n=x, X_{n-1}=v,Y_n,Y_{n-1}\right]\Gamma^{\Delta t}(x|v)\\
&\propto\mathbb E\left[\mathbb L(y|(X_u)_{\{t_{n-1}\leq u\leq t_n\}},Y_{n-1})\Big|X_n=x, X_{n-1}=v,Y_{n-1}\right]e^{Q^*\Delta t}(x|v)\ .
\end{align*}
Lastly, when computing the likelihood based on the time-$n$ observation, the last line is evaluated at $y=Y_n$. This completes the proof of the proposition.

\end{proof}
 At this point it seems that the filter is rather complicated, and would be difficult to implement in real-time. Often times, what one might do is consider a discrete scheme that approximates the SDEs:
\begin{eqnarray}
\label{eq:discreteSDEy}
\Delta Y_{n-1}&=& (\mu-.5X_{n-1})\Delta t+\sqrt{X_{n-1}}(\rho \Delta B_{n-1}+\sqrt{1-\rho^2}\Delta W_{k-1})\\
\label{eq:discreteSDEx}
\Delta X_{n-1} &=& \kappa(\bar X-X_{n-1})\Delta t+\gamma\sqrt{X_{n-1}}\Delta B_{n-1}
\end{eqnarray}
where $\Delta Y_{n-1}=Y_n-Y_{n-1}$ and $\Delta X_{n-1}=X_n-X_{n-1}$. Using (\ref{eq:discreteSDEy}) and (\ref{eq:discreteSDEx}), we can impute an approximate filtering density to the density given in Proposition \ref{prop:SVMfilter}:

\begin{equation}
\label{eq:volFilter}
\tilde\pi_n(x) = \frac{1}{c_n}\int \psi_n(x,v)e^{Q^*\Delta t}(x|v)\tilde \pi_{n-1}(dv)
\end{equation}
where $\psi_n(x,v)$ is the likelihood of $\{X_n=x,X_{n-1}=v\}$ given $\{Y_n,Y_{n-1}\}$,
\[\psi_n(x,v) = \frac{1}{\sqrt{v(1-\rho^2)\Delta t}}\exp\left\{-\frac{(\Delta Y_{n-1}-(\mu-.5v)\Delta t-\sqrt{v}\rho\Delta B_{n-1}(x,v))^2}{2v(1-\rho^2)\Delta t}\right\}\]

\[\hbox{with}\qquad \Delta B_{n-1}(x,v) = \frac{1}{\gamma\sqrt v}\left(x-v-\kappa(\bar X-v)\Delta t\right).\]
If the model simplification given by (\ref{eq:discreteSDEy}) and (\ref{eq:discreteSDEx}) can be considered `correct', then there is no need to dispute the validity of the filter given by (\ref{eq:volFilter}). But in general, if one knows apriori that the continuous-time SDEs are the correct model, then there needs to be some analysis to verify that the approximate filter (such as that in(\ref{eq:volFilter})) converges as $\Delta t\searrow 0$, that is
\[\left|\int g(x)\tilde\pi_n(dx)-\int g(x)\pi_n(dx)\right|\rightarrow 0\]
in a strong sense as $\Delta t\rightarrow 0$. As was mentioned earlier, it is well-known (see Kushner \cite{kushner08}) that approximate filters that are sometimes consistent, but the results in \cite{kushner08} do not apply to the Heston model.
\section{Extracting the Risk-Premium}
Under the physical measure, we have,

\[\mathbb EX_t = \mathbb EX_0e^{-\kappa t}+\bar X(1-e^{-\kappa t})\]
and so the expected value of realized variance ($RV_{[0,T]} \doteq \frac 1T\int_0^TX_sds$) is
\[\mathbb E_0RV_{[0,T]} = \bar X-\frac{\bar X-\mathbb E_0X_0}{\kappa T}\left(1-e^{-\kappa T}\right)\]
where we have (without loss of generality) considered the case at time 0, and we have denoted the posterior expectation as $\mathbb E_0^*[~\cdot~]=\mathbb E^*[~\cdot~|\mathcal F_0^Y]$ and $\mathbb E_0[~\cdot~]=\mathbb E[~\cdot~|\mathcal F_0^Y]$. Therefore, there is the following close-formula for the risk-premium 

\[\mathcal{RP}_0\doteq \mathbb E_0^*RV_{[0,T]}-\mathbb E_0RV_{[0,T]}\]


\[=\mathbb E_0^*RV_{[0,T]}-\bar X+\frac{\bar X-\mathbb E_0X_0}{\kappa T}\left(1-e^{-\kappa T}\right).\]
This expression for the risk-premium holds whenever volatility-squared is modeled with a mean-reverting SDE with drift term $\kappa(\bar X-X_t)$, not just the Heston model.

Under the risk-neutral measure, the market adds a risk-premium term to $dX_t$:
\[dX_t = \kappa(\bar X-X_t)dt-\Lambda_tX_tdt+\gamma\sqrt{X_t}dB_t^*\]
where $B_t^*$ Brownian motion under the market measure, and $\Lambda_t$ is the market price of volatility risk. We leave the modeling of $\Lambda_t$ open here because we will not delve deeply into its correlation structure. However, $\Lambda_t$ is most likely thought of as a mean-reverting process and could be modeled as such. Under the market's measure there is the following expectation of variance 
\[\mathbb E_0^*X_t =\mathbb E_0X_t-\int_0^t\mathbb E_0^*[X_s\Lambda_s]e^{-\kappa(t-s)}ds.\]
From this, we see that $\mathbb E_0^*RV_{[0,T]}$ can be written as follows,

\[\mathbb E_0^*RV_{[0,T]} = \frac{1}{T}\int_0^T\mathbb E_0^*X_tdt = \frac{1}{T}\int_0^T\mathbb E_0X_tds-\frac{1}{T}\int_0^T\int_0^t\mathbb E_0^*[X_s\Lambda_s]e^{-\kappa(t-s)}dsdt\]

\[=\mathbb E_0[RV_{0,T}] - \underbrace{\frac{1}{\kappa T}\int_0^T\left(1-e^{-\kappa(T-s)}\right)\mathbb E_0^*[X_s\Lambda_s]ds}_{\hbox{risk-premium}}\qquad\qquad(**)\ .\]
 In $(**)$, notice that if $X$ and $\Lambda$ are independent, then for $\kappa\gg 1$ the risk-premium simplifies to 
 \[\frac{1}{\kappa T}\int_0^T\left(1-e^{-\kappa(T-s)}\right)\mathbb E_0^*[X_s\Lambda_s]ds\sim -\frac{\bar X}{T}\mathbb E_0^*\Lambda_T\ .\]
%
\section{Filtering Average Volatility in Fast Time-Scales}
The purpose of fast time-scales in volatility modeling is to capture mean reverting effects that occur on the order of 2 to 3 days. In the Heston model, suppose we are in a fast time-scale where $\gamma\sim\sqrt\kappa$ for $\kappa$ large. Then, it can be shown that the distribution of $X_t$ settles into a $\Gamma$ distribution almost instantaneously,
\[X_t\Rightarrow\Gamma\left(2\bar X,\frac{1}{2}\right)\qquad\hbox{as }\kappa\nearrow\infty\]
for all $t>0$. Therefore, there is a fast-averaging of the realized variance,
\[\frac{1}{T}\int_0^TX_sds\rightarrow \bar X\qquad\hbox{in probability as }\kappa\nearrow\infty\]
and so the expected payoff of any contract that is a function of realized variance will be deterministic unless $\bar X$ is random and/or unknown. Thus, building a risk-premium into the dynamics of $X_t$ in the manner that we did in the previous section will not be meaningful in fast time-scales. An alternative idea would be to take $\bar X$ as a hidden regime-process that is governed by another Markov chain, and thus adds another dimension to the HMM. Then, we can take realized variance as our observations and write a filter to estimate the regime $\bar X_t$. We do this as follows:

Take $\bar X_t$ to be a Markov chain with generator $Q$, for which we assume the standard structure for changes; changes in $\bar X_t$ are governed by a Poisson jump process so that over a time interval of length $\Delta t$, the probability of $\bar X_t$ changing states more than once is $o(\Delta t)$. The new dynamics of $X_t$ are then
\[dX_t = \kappa(\bar X_t-X_t)dt+\gamma\sqrt{X_t}dB_t\]
and for such a model the realized variance is a random variable in fast time-scales,
\[\frac{1}{T}\int_0^TX_sds \sim \frac{1}{T} \int_0^T\bar X_sds\qquad\hbox{for $\kappa$ large.}\]

Let $(t_{n,\ell})_{n,\ell}$ be a partition of some finite time interval, say $10$ years, where $n$ denotes the $nth$ week, and $\ell$ denotes the $\ell th$ trade. Then
\[t_{n+1,\ell}-t_{n,\ell}=\Delta t=\hbox{1 week,}\qquad\qquad\hbox{for any $\ell$}\]
 \[t_{n,\ell+1}-t_{n,\ell}=\hbox{ time from $\ell th$ quote until the next trade during week $n$.}\]
We have observations on $Y_t$ at each $t_{n,\ell}$ (i.e. $Y_{n,\ell}=Y_{t_{n,\ell}}$ is the $\ell th$ observations on the $nth$ week). For any $n$ and $\ell$ let $\Delta Y_{n,\ell} = Y_{n,\ell+1} -Y_{n,\ell} $, and for simplicity let $t_n=t_{n,0}$. We then have the following model for weekly observations on realized variance,
\[Z_{n+1}\doteq\frac{1}{\Delta t}\sum_{\ell}(\Delta Y_{n,\ell})^2=\frac{1}{\Delta t}\int_{t_{n}}^{t_{n+1}}\bar X_sds+\epsilon_{n+1}\]
where $\epsilon_{n+1}$ is a noise process with $\mathbb E\epsilon_n\epsilon_{n'} = 0$ if $n\neq n'$. Clearly, $Z_n$ is observable  and $(Z_n,\bar X_n)$ is a Markov process. From here the goal is to estimate the the state-space and transition rates of $\bar X_t$, and then apply the nonlinear filtering results in estimating the variance risk-premium. Letting $\bar\pi_n(x_i) = \mathbb P(\bar X_{t_n}=x_i|Y_{0:n})$, we have  an estimate of the physical measure's expectation of realized variance:

\[\mathbb E_{t_n}\left[\frac{1}{T}\int_{t_n}^{t_n+T} X_sds\right]\approx\mathbb E_{t_n}\left[\frac{1}{T}\int_{t_n}^{t_n+T}\bar X_sds\right]=\frac{1}{T}\int_{t_n}^{t_n+T}\sum_ix_ie^{Q^*(s-t_n)}\bar\pi_n(x_i)ds\ ,\]
for $\kappa\sim \gamma^2\gg 1$.



\chapter{The Zakai Equation}

\noindent Let $X_t\in\mathcal S$ be a Markov process with generator $Q$. Let the domain of $Q$ be denoted by $\mathcal B(Q)$, and let $\mathcal B_b(Q)$ denote the subset of bounded functions in $\mathcal B(Q)$. For any function $g(x)\in\mathcal B_b(Q)$ we have the following limit: 
\[\frac{\mathbb E[g(X_{t+\Delta t})|X_t=x]-g(x)}{\Delta t} \rightarrow Qg(x)\]
as $\Delta t\searrow 0$, for any $x\in\mathcal S$. If $Q$ is densely-defined and its resolvent set includes all positive real numbers, then the Hille-Yosida theorem applies, allowing us to write the the distribution of $X_t$ with a contraction semi-group. In these notes we assume that such conditions hold and that the transition density/mass is generated by an operator semigroup denoted by $e^{Q^*t}$.

A standard nonlinear filtering problem in SDE theory assumes that $X_t$ is unobserved and that a process $Y_t$ is given by an SDE
\begin{eqnarray}
\label{eq:dY}
dY_t&=&h(t,X_t)dt+\gamma dW_t\qquad\hbox{observed}
\end{eqnarray}
where $W_t$ is an independent Wiener process, $\gamma>0$ and we assume that $h(t,\cdot)$ is bounded for all $t<\infty$.

The pair $(X_t,Y_t)$ is an HMM for which filtering can be used to find the posterior distribution. Let $\mathcal F_t^Y = \sigma\{Y_s:s\leq t\}$, and for any measurable function $g(x)$ let 
\[\hat g_t = \mathbb E[g(X_t)|\mathcal F_t^Y].\]
The posterior expectation of $\hat g_t$ is ultimately what is desired from filtering, but the methods for obtaining the posterior distribution are quite involved. The discrete Bayesian tools that we've used in earlier lectures cannot be used here because we are in a continuum that does not allow us to break apart the layers of the HMM. Instead, we will exploit well-known ideas from SDE theory to obtain the differentials for the filtering distribution. In particular, we will use the Girsanov theorem to obtain the Zakai equation.
\section{Discrete Motivation from Bayesian Perspective}
\label{eq:discreteZakai}
In their book, Karatzas and Shreve \cite{K-S} give a discrete motivation for how the Girsanov theorem works. In a similar fashion, we consider a discrete problem and then construct a change of measure from the ratio of the appropriate densities. We then show how it is analogous to its continuous-time counterpart, and leads to a discrete approximation of the Zakai equation.

To do so, we start by considering a discrete-time analogue of (\ref{eq:dY}),
\[Y_n=Y_{n-1}+h(t_{n-1},X_{n-1})\Delta t+\gamma\Delta W_{n-1},\]
and assume that the unconditional distribution of $X_t$ is a density for all $t\geq 0$ (the same idea will be applicable when $X_t$'s distribution has a mass function). We can easily apply Bayes theorem to obtain the filtering distribution on a Borel set $\mathcal A$:
\[\pi_n(\mathcal A) = \frac{1}{c_n}\int\psi_n(v)e^{Q^*\Delta t}(\mathcal A|v)\pi_{n-1}(dv)\]
where $e^{Q^*\Delta t}$ represents kernel of $X_t$'s transition densities, the likelihood function is
\[\psi_n(v) = \exp\left\{-.5\left(\frac{\Delta Y_{n-1}-h(t_{n-1},v)\Delta t}{\gamma\sqrt{\Delta t}}\right)^2\right\},\qquad\hbox{with }\Delta Y_{n-1}=Y_n-Y_{n-1},\]
and $c_n$ is a normalizing constant. The Lebesgue differentiation theorem can be applied to obtain the density of the posterior,
\[\frac{1}{|\mathcal A|}\pi_n(\mathcal A)\rightarrow \pi_n(dx) = \frac{1}{c_n}\int\psi_n(v)e^{Q^*\Delta t}(dx|v)\pi_{n-1}(dv)\]
when $\mathcal A$ shrinks nicely to $\{x\}$.

Keeping this discrete model and filter in mind, let's shift our attention to a joint density function of all observations and a possible path $(x_0,x_1,\dots,x_n)$ taken by $(X_0,X_1,\dots,X_n)$,
\[dp_n\doteq\mathbb P(Y_0,Y_1,\dots,Y_n;dx_0,dx_1,\dots,dx_n)\]

\[ = \mathbb P(Y_0,Y_1,\dots,Y_n|x_0,x_1,\dots,x_n)\mathbb P(dx_0,dx_1,\dots,dx_n)\]

\[=\left(\prod_{\ell=0}^{n-1}\psi_\ell(x_\ell)\right)\times \mathbb P(dx_0,dx_1,\dots,dx_n)=\left(\prod_{\ell=0}^{n-1}e^{-.5\left(\frac{\Delta Y_\ell-h(t_\ell,x_\ell)\Delta t}{\gamma\sqrt{\Delta t}}\right)^2}\right)\times \mathbb P(dx_0,dx_1,\dots,dx_n)\]

\[=e^{-.5\sum_{\ell=0}^{n-1}\left(\frac{\Delta Y_\ell-h(t_\ell,x_\ell)\Delta t}{\gamma\sqrt{\Delta t}}\right)^2}\times \mathbb P(dx_0,dx_1,\dots,dx_n)\]
where $\mathbb P(dx_0,dx_1,\dots,dx_n)$ can be obtained using the exponential of $Q^*$.

Next, consider an equivalent measure in $\tilde{\mathbb P}$ where $\Delta Y_n/\gamma$ is Brownian motion independent of $X_n$, and the law of $X_n$ remains the same. Under this new measure the joint density function is
\[d\tilde p_n\doteq\tilde{\mathbb P}(Y_0,Y_1,\dots,Y_n;dx_0,dx_1,\dots,dx_n) = \tilde{\mathbb P}(Y_0,Y_1,\dots,Y_n)\tilde{\mathbb P}(dx_0,dx_1,\dots,dx_n)\]

\[=\left(\prod_{\ell=0}^{n-1}e^{-.5\left(\frac{\Delta Y_\ell}{\gamma\sqrt{\Delta t}}\right)^2}\right)\times \mathbb P(dx_0,dx_1,\dots,dx_n)\]

\[=e^{-.5\sum_{\ell=0}^{n-1}\left(\frac{\Delta Y_\ell}{\gamma\sqrt{\Delta t}}\right)^2}\times \mathbb P(dx_0,dx_1,\dots,dx_n).\]
The ratio of these densities is written follows:
\[\frac{dp_n}{d\tilde p_n} \doteq M_n\Big|_{\mathcal F_n^Y} = \exp\left\{\sum_{\ell=0}^{n-1}\frac{h(t_\ell,x_\ell)\Delta Y_\ell}{\gamma^2}-\frac{1}{2}\sum_{\ell=0}^{n-1}\frac{h^2(t_\ell,x_\ell)\Delta t}{\gamma^2}\right\}\]
which is the likelihood ratio of any path for the discrete observation model. It is also the discrete analog of the exponential martingale that we use in the Girsanov theorem. Furthermore, we can use $M_n$ to rewrite the filtering expectation in terms of the alternative measure,
\[\mathbb E[g(X_n)|\mathcal F_n^Y] = \frac{\int g(x_n)dp_n}{\int dp_n}=\frac{\int g(x_n)M_nd\tilde p_n}{\int M_nd\tilde p_n}=\frac{\tilde{\mathbb E}[g(X_n)M_n|\mathcal F_n^Y]}{\tilde{\mathbb E}[M_n|\mathcal F_n^Y]},\]
and if we define $\phi_n[g] = \tilde{\mathbb E}[g(X_n)M_n|\mathcal F_n^Y]$ we can write the filtering expectation as
\[\hat g_n=\mathbb E[g(X_n)|\mathcal F_n^Y] = \frac{\phi_n[g]}{\phi_n[1]}\]
for any function $g(x)\in\mathcal B_b(Q)$.

It turns out to be advantageous to analyze under $\tilde{\mathbb P}$-measure because the dynamics $\phi_n$ are linear when we move to a continuum of observations. To get a sense of the linearity, consider the following discrete expansion for small $\Delta t$,
\[\phi_{n+1}[g] = \tilde{\mathbb E}[g(X_{n+1})M_{n+1}|\mathcal F_{n+1}^Y]\]

\[=\tilde{\mathbb E}\left[g(X_{n+1})\left(1+\frac{h(t_n,X_n)}{\gamma^2}\Delta Y_n\right)M_n\Big|\mathcal F_{n+1}^Y\right]+o(\Delta t)\]

\[=\tilde{\mathbb E}\left[g(X_{n+1})M_n\Big|\mathcal F_{n+1}^Y\right]+\tilde{\mathbb E}\left[g(X_{n+1})\left(\frac{h(t_n,X_n)}{\gamma^2}\right)M_n\Big|\mathcal F_{n+1}^Y\right]\Delta Y_n+o(\Delta t)\]
because $M_{n+1}=\left(1+\frac{h(t_n,X_n)}{\gamma^2}\Delta Y_n\right)M_n+o(\Delta t)$ by a discrete interpretation of It\^o's lemma, and $\Delta Y_n$ is $\mathcal F_{n+1}^Y$-measurable. Now, because $X$ is independent of $Y$ under $\tilde{\mathbb P}$, the conditioning up to time $n+1$ is superfluous and we can reduce it down to time $n$,
\[=\tilde{\mathbb E}\left[g(X_{n+1})M_n\Big|\mathcal F_n^Y\right]+\tilde{\mathbb E}\left[g(X_{n+1})\left(\frac{h(t_n,X_n)}{\gamma^2}\right)M_n\Big|\mathcal F_n^Y\right]\Delta Y_n+o(\Delta t).\]
Then we can again exploit the independence of $X$ from $Y$ under $\tilde{\mathbb P}$ and use an approximation of $X$'s backwards operator.

\[=\tilde{\mathbb E}\left\{\tilde{\mathbb E}\left[g(X_{n+1})M_n\Big|\mathcal F_n^Y\vee X_n\right]\Bigg|\mathcal F_n^Y\right\}\qquad\qquad\qquad\qquad\qquad\qquad\qquad\]

\[\qquad\qquad\qquad\qquad\qquad+\tilde{\mathbb E}\left\{\tilde{\mathbb E}\left[g(X_{n+1})\left(\frac{h(t_n,X_n)}{\gamma^2}\right)M_n\Big|\mathcal F_n^Y\vee X_n\right]\Bigg|\mathcal F_n^Y\right\}\Delta Y_n+o(\Delta t)\]

\[=\tilde{\mathbb E}\left\{(I+Q\Delta t)\tilde{\mathbb E}\left[g(X_n)M_n\Big|\mathcal F_n^Y\vee X_n\right]\Bigg|\mathcal F_n^Y\right\}\qquad\qquad\qquad\qquad\qquad\qquad\qquad\]

\[\qquad\qquad\qquad+\tilde{\mathbb E}\left\{(I+Q\Delta t)\tilde{\mathbb E}\left[g(X_n)\left(\frac{h(t_n,X_n)}{\gamma^2}\right)M_n\Big|\mathcal F_n^Y\vee X_n\right]\Bigg|\mathcal F_n^Y\right\}\Delta Y_n+o(\Delta t)\]

\[=\tilde{\mathbb E}\left[(I+Q\Delta t)g(X_n)M_n\Big|\mathcal F_n^Y\right]+\tilde{\mathbb E}\left[(I+Q\Delta t)g(X_n)\left(\frac{h(t_n,X_n)}{\gamma^2}\right)M_n\Big|\mathcal F_n^Y\right]\Delta Y_n+o(\Delta t)\]
and then using the fact that $\Delta Y_n\cdot \Delta t=o(\Delta t)$, we have

\[=\tilde{\mathbb E}[g(X_n)M_n|\mathcal F_n^Y]+\tilde{\mathbb E}[Qg(X_n)M_n|\mathcal F_n^Y]\Delta t+\tilde{\mathbb E}\left[g(X_n)\frac{h(t_n,X_n)}{\gamma^2}M_n\Big|\mathcal F_n^Y\right]\Delta Y_n+o(\Delta t)\]

\[=\phi_n[g]+\phi[Qg]\Delta t+\phi\left[g\frac{h}{\gamma^2}\right]\Delta Y_n+o(\Delta t)\]
which foreshadows the Zakai equation in a discrete setting,
\[\Delta \phi_n[g] = \phi[Qg]\Delta t+\phi\left[g\frac{h}{\gamma^2}\right]\Delta Y_n+o(\Delta t).\]

\section{Derivation of the Zakai Equation}
In this section we use the Girsanov theorem as the main tool in a formal derivation of the nonlinear filtering equations in continuous time. For ease in notation we let $h(t,x) = h(x)$, but this does change the results because the would merely need to be rewritten to include the time dependence in $h(~)$.

We start by considering the finite interval $[0,T]$ and defining the following exponential,
\[M_t \doteq \exp\left\{\frac{1}{2\gamma^2}\int_0^th^2(X_s)ds+\frac{1}{\gamma}\int_0^th(X_s)dW_s   \right\}\]

\[= \exp\left\{-\frac{1}{2\gamma^2}\int_0^th^2(X_s)ds+\frac{1}{\gamma^2}\int_0^th(X_s)dY_s   \right\}\]
for all $t\in[0,T]$. Since it was initially assumed that $W_t$ was independent, we can define an equivalent measure $\tilde{\mathbb P}$ by
\[d\tilde{\mathbb P} = M_T^{-1}d\mathbb P.\]
By the Girsanov theorem we know that 
\begin{enumerate}
\item $\tilde{\mathbb P}$ is a probability measure, and
\item $Y_t/\gamma$ is $\tilde{\mathbb P}$-Brownian motion for $t\in[0,T]$, conditioned on $X$.
\end{enumerate}
It is easy to show with moment generating functions that $\tilde{\mathbb P}(X_t\leq x) = \mathbb P(X_t\leq x)$. We can also easily show that $M_t$ is a true $\tilde{\mathbb P}$-martingale,
\[
\tilde{\mathbb E}M_t = \int M_t(\omega)d\tilde{\mathbb P}(\omega) =  \int M_t(\omega)M_T^{-1}(\omega)d\mathbb P(\omega)  = \mathbb E[M_t/M_T] 
\]
\[
= \mathbb E\left[\exp\left\{-\frac{1}{2\gamma^2}\int_t^Th^2(X_s)ds-\frac{1}{\gamma}\int_t^Th(X_s)dW_s   \right\} \right]
\]
\[=1\ .\]
Lastly, a simple lemma shows that $X$ and $Y$ are path-wise independent under $\tilde{\mathbb P}$:

\begin{lemma} $X$ and $Y$ are path-wise independent under $\tilde{\mathbb P}$.\end{lemma}
\begin{proof}
For an arbitrary path-wise function $f_1$ we have
\[\tilde {\mathbb E}f_1(Y/\gamma) = \tilde {\mathbb E}[\tilde{\mathbb E}[f_1(Y/\gamma)|X]] = \tilde {\mathbb E}[\mathbb Ef_1(W)] = \mathbb Ef_1(W)\]
showing that $Y/\gamma$ is $\tilde{\mathbb P}$-Brownian motion unconditional on $X$. Then for another arbitrary path-wise function $f_2$ we have
\[\tilde {\mathbb E}f_1(Y/\gamma)f_2(X) = \tilde {\mathbb E}[f_2(X)\tilde {\mathbb E}[f_1(Y/\gamma)|X]]=\tilde {\mathbb E}[f_2(X)\mathbb Ef_1(W)]\]

\[=\mathbb Ef_1(W)\tilde{\mathbb  E}f_2(X)=\tilde{\mathbb E}f_1(Y/\gamma)\tilde{\mathbb E}f_2(X)\]
and so $X$ and $Y$ are $\tilde{\mathbb P}$-independent. 
\end{proof}

From here forward, define the measure $\phi_t$ on $\mathcal B_b(Q)$ as
\[\phi_t[g]\doteq\tilde{\mathbb E}[g(X_t)M_t|\mathcal F_t^Y].\]
With this new measure we can express another important result regarding the Girsanov change of measure, namely the Kallianpur-Streibel formula:
\begin{lemma}\textbf{Kallianpur-Streibel Formula:}
\[\mathbb E[g(X_t)|\mathcal F_t^Y] = \frac{\phi_t[g]}{\phi_t[1]}\]
for any $g(x)\in\mathcal B(Q)$.
\end{lemma}
\begin{proof} For any $ A\in\mathcal F_t^Y$ we have
\[\mathbb E\left\{\mathbf 1_A\mathbb E[g(X_t)|\mathcal F_t^Y]\right\}=\mathbb E[\mathbf 1_Ag(X_t)]=\tilde{\mathbb E}[\mathbf 1_Ag(X_t)M_t]\]

\[=\tilde{\mathbb E}\left\{\mathbf 1_A\tilde{\mathbb E}[g(X_t)M_t|\mathcal F_t^Y]\right\}=\tilde{\mathbb E}\left\{\mathbf 1_A\frac{\tilde{\mathbb E}[g(X_t)M_t|\mathcal F_t^Y]}{\tilde{\mathbb E}[M_t|\mathcal F_t^Y]}\tilde{\mathbb E}[M_t|\mathcal F_t^Y]\right\}=\tilde{\mathbb E}\left\{\mathbf 1_A\frac{\phi_t[g]}{\phi_t[1]}\tilde{\mathbb E}[M_t|\mathcal F_t^Y]\right\}\]

\[=\tilde{\mathbb E}\left\{\tilde{\mathbb E}\left[\mathbf 1_A\frac{\phi_t[g]}{\phi_t[1]}M_t\Big|\mathcal F_t^Y\right]\right\}=\tilde{\mathbb E}\left[\mathbf 1_A\frac{\phi_t[g]}{\phi_t[1]}M_t\right]=\mathbb E\left[\mathbf 1_A\frac{\phi_t[g]}{\phi_t[1]}\right]\]
and since $A$ was an arbitrary set, this shows that the result holds wp1.
\end{proof}

Now, for any $t\in[0,T]$, we use Fubini's theorem to bring the differential inside the expectation, and from there we apply It\^o's lemma, which gives us the following differential,
\[d\tilde{\mathbb E}[g(X_t)M_t|\mathcal F_T^Y] = \tilde{\mathbb E}[Qg(X_t)M_t|\mathcal F_T^Y]dt+ \tilde{\mathbb E}\left[g(X_t)\frac{h(X_t)}{\gamma^2}M_t\Big|\mathcal F_T^Y\right]dY_t\]
for any $g(x)\in\mathcal B_b(Q)$. From this we can construct the integrated form of the differential,
\[\tilde{\mathbb E}[g(X_t)M_t|\mathcal F_T^Y] \]

\[=\tilde{\mathbb E}[g(X_0)|\mathcal F_T^Y]+\int_0^t\tilde{\mathbb E}[Qg(X_s)M_s|\mathcal F_T^Y]ds+ \int_0^t\tilde{\mathbb E}\left[g(X_s)\frac{h(X_s)}{\gamma^2}M_s\Big|\mathcal F_T^Y\right]dY_s\]
and since $X$ is independent of $Y$ under the $\tilde{\mathbb P}$-measure, we can reduce the filtrations from $\mathcal F_T^Y$ to $\mathcal F_s^Y$ for all $s\leq T$, giving us,
\begin{equation}
\label{eq:intZakai}
\tilde{\mathbb E}[g(X_t)M_t|\mathcal F_t^Y] =\mathbb E[g(X_0)]+\int_0^t\tilde{\mathbb E}[Qg(X_s)M_s|\mathcal F_s^Y]ds+ \int_0^t\tilde{\mathbb E}\left[g(X_s)\frac{h(X_s)}{\gamma^2}M_s\Big|\mathcal F_s^Y\right]dY_s
\end{equation}
for all $t\leq T$. Inserting $\phi_s[\cdot]$ in (\ref{eq:intZakai}) wherever possible and then taking the differential with respect to $t$, we have \textbf{the Zakai equation:}
\begin{equation}
\label{eq:zakai}
d\phi_t[g] = \phi_t[Qg]dt+\phi_t\left[g\frac{h}{\gamma^2}\right]dY_t
\end{equation}
for all $t\leq T$. The Zakai equation can also be considered for general unbounded functions $g(x)\in\mathcal B(Q)$, but we have restricted ourselves to the bounded case in order to insure that $\phi_t[g]$ is finite almost surely. Existence of solutions to (\ref{eq:zakai}) is straight-forward because we have derived it by differentiating $\tilde{\mathbb E}[g(X_t)M_t|\mathcal F_t^Y]$. Uniqueness of measure-valued solutions to (\ref{eq:zakai}) has been shown by Kurtz and Ocone \cite{kurtzOcone1988} using a filtered martingale problem, and by Rozovsky \cite{rozovsky1992} using a Radon measure representation of $\phi_t$.
\subsection{The Adjoint Zakai Equation}
Depending on the nature of the filtering problem, the unnormalized probability measure $\phi_t[~\cdot~]$  may have a density/mass function. For example, suppose that $X_t\in\mathbb R$ is a diffusion process satisfying the SDE
\[dX_t= a(X_t)dt+\sigma dB_t\]
where $B_t\perp W_t$. Then the generator is $Q=\mathcal L = \frac{\sigma^2}{2}\frac{\partial^2}{\partial x^2}\cdot~+a(x)\frac{\partial}{\partial x}\cdot~$ and the Zakai equation is
\[d\phi_t[g] = \phi_t[\mathcal Lg]dt+\frac{1}{\gamma^2}\phi_t[gh]dY_t\]

\[=\frac{\sigma^2}{2}\phi_t\left[\frac{\partial^2}{\partial x^2}g\right]dt+\phi_t\left[a\frac{\partial}{\partial x}g\right]dt+\frac{1}{\gamma^2}\phi_t[gh]dY_t\]
for any bounded function $g(x)$ with a 2nd derivative. Depending on $a(x),\sigma$ and the initial conditions, $\tilde\pi_t$ may be a density so that 
\[\phi_t[g] = \int g(x)\tilde\pi_t(x)dx,\] 
and provided that certain regularity conditions are met, the adjoint of the Zakai equation gives us an SPDE for $\tilde\pi_t$,
\begin{equation}
\label{eq:adjointZakai}
d\tilde\pi_t(x) = \mathcal L^*\tilde\pi_t(x)dt+\frac{h(x)}{\gamma^2}\tilde\pi_t(x)dY_t
\end{equation}
with the initial condition $\tilde\pi_0(x)= \mathbb P(X_0\in dx)$, and with the adjoint operator given by $\mathcal L^* =\frac{\sigma^2}{2}\frac{\partial^2}{\partial x^2}\cdot~-\frac{\partial}{\partial x}\left(a(x)~\cdot~\right)$. Existence and uniqueness of such densities is beyond the scope of these notes. Readers who are interested in regularity of solutions should read the book by Pardoux \cite{pardoux}.

In the case of filtering distributions that are composed of mass functions, the adjoint equation is similar. For instance, if $X_t\in\{x_1,\dots,x_m\}$ is a finite-state Markov chain with generator $Q$, the adjoint equation holds without any regularity conditions, 
\[d\tilde\pi_t(x_i) = Q^*\tilde\pi_t(x_i)dt+\frac{h(x_i)}{\gamma^2}\tilde\pi_t(x_i)dY_t\]

\[= \sum_{j=1}^mQ_{ji}\tilde\pi_t(x_j)dt+\frac{h(x_i)}{\gamma^2}\tilde\pi_t(x_i)dY_t.\]
General existence and uniqueness for $\tilde{\pi}_t$ in this discrete case was shown by Rozovsky \cite{rozovsky1972}.

\subsection{Kushner-Stratonovich Equation}
Using the Zakai equation of (\ref{eq:zakai}) and observing that our assumption that $h$ is bounded implies $\mathbb P(\phi_t[1]<\infty)=1$ for all $t\in [0,T]$, we can apply It\^o's lemma to obtain
\begin{equation}
\label{eq:kushner}
d\hat g_t = d\left(\frac{\phi_t[g]}{\phi_t[1]}\right) = \widehat{Qg}_tdt+ \frac{\widehat{gh}_t-\hat g_t\hat h_t}{\gamma^2}\left(dY_t-\hat h_tdt\right).
\end{equation}
However, it should be mentioned that (\ref{eq:kushner}) was originally obtain a few years before the Zakai equation using other methods. Under the appropriate regularity conditions, the Kushner-Stratonovich equation is the adjoint (\ref{eq:kushner}) and is the nonlinear equation for the filtering distribution,
\begin{equation}
\nonumber
d\pi_t(x) = Q^*\pi_t(x)dt+\frac{h(x)-\hat h_t}{\gamma^2}\pi_t(x)\left(dY_t-\hat h_tdt\right)
\end{equation}
where $\pi_t(x)$ is either a density or a mass function (depending on the type of problem).
\subsection{Smoothing}
The smoothing filter has been derived in \cite{bobrovskyZeitouni}, but in the case of regularized processes where the adjoint Zakai equation holds. In this section we derive a similar results but for the general case of functions in $\mathcal B_b(\mathcal S)$, and we'll also derive the backward SPDE for smoothing in the regularized case.

Consider the times $\tau$ and $t$ such that $0\leq \tau\leq t$. The filtering expectation of $g(X_\tau)$ is
\[\mathbb E[g(X_\tau)|\mathcal F_t^Y] = \frac{\tilde{\mathbb E}[M_tg(X_\tau)|\mathcal F_t^Y] }{\tilde{\mathbb E}[M_t|\mathcal F_t^Y] }= \frac{\tilde{\mathbb E}[M_tg(X_\tau)|\mathcal F_t^Y] }{\phi_t[1]}.\]
For $\tau$ fixed and for $t$ increasing, the Zakai equation is
\begin{eqnarray*}
d\tilde{\mathbb E}[M_tg(X_\tau)|\mathcal F_t^Y] &=&\frac{1}{\gamma^2}\tilde{\mathbb E}[M_th(X_t)g(X_\tau)|\mathcal F_t^Y] dY_t\qquad\hbox{for }t>\tau,\\
\tilde{\mathbb E}[M_\tau g(X_\tau)|\mathcal F_\tau^Y] &=&\phi_\tau[g].
\end{eqnarray*}
This equation can be solved to obtain the smoothing distribution, but lacks a differential formula for changes in $\tau$. However, in the regular case there is a backward SPDE that will provide a differential for changes in $\tau$. This backward SDE will provide improved efficiency for coding and analysis.

Suppose there is sufficient regularity so that $\tilde\pi_t(x)\doteq\phi_t[\delta_x] $ satisfies the adjoint Zakai equation (see equation (\ref{eq:adjointZakai})). Given $\mathcal F_t^Y$, the smoothing filter for any time $\tau\in[0,t]$ is
\[\frac{d}{dx}\mathbb P(X_\tau\leq x|\mathcal F_t^Y)= \frac{\tilde{\mathbb E}[M_t\delta_x(X_\tau)|\mathcal F_t^Y]}{\tilde{\mathbb E}[M_t|\mathcal F_t^Y]}\]

\[= \frac{\tilde{\mathbb E}[M_\tau\delta_x(X_\tau)\tilde{\mathbb E}[M_t/M_\tau|\mathcal F_t^Y\vee\mathcal F_\tau^X]|\mathcal F_t^Y]}{\tilde{\mathbb E}[M_t|\mathcal F_t^Y]}\]

\[= \frac{\tilde{\mathbb E}[M_\tau\delta_x(X_\tau)\tilde{\mathbb E}[M_t/M_\tau|\mathcal F_t^Y\vee \{X_\tau=x\}]|\mathcal F_t^Y]}{\tilde{\mathbb E}[M_t|\mathcal F_t^Y]}\]

\[=\frac{\tilde{\mathbb E}[M_\tau\delta_x(X_\tau)|\mathcal F_t^Y] \cdot\tilde{\mathbb E}[M_t/M_\tau|\mathcal F_t^Y\vee \{X_\tau=x\}]}{\tilde{\mathbb E}[M_t|\mathcal F_t^Y]}\]

\[=\frac{\tilde{\mathbb E}[M_\tau\delta_x(X_\tau)|\mathcal F_\tau^Y]\cdot \tilde{\mathbb E}[M_t/M_\tau|\mathcal F_t^Y\vee \{X_\tau=x\}]}{\tilde{\mathbb E}[M_t|\mathcal F_t^Y]}\]

\[=\frac{\tilde\pi_\tau(x)\alpha_{\tau,t}(x)}{\phi_t[1]},\]
where $\alpha_{\tau,t}$ is define as 

\[\alpha_{\tau,t}(x) \doteq \tilde{\mathbb E}\left[M_t/M_\tau\Big|\mathcal F_t^Y\vee\{X_\tau=x\}\right]\] 

\[=\tilde{\mathbb E}\left[\exp\left\{-\frac{1}{2\gamma^2}\int_\tau^th^2(X_s)ds+\frac{1}{\gamma^2}\int_\tau^th(X_s)dY_s   \right\}\Big|\mathcal F_t^Y\vee\{X_\tau=x\}\right].\]
The function $\alpha_{\tau,t}(x)$ is the smoothing component and satisfies the following backward SPDE
\begin{eqnarray*}
d\alpha_{\tau,t}(x) &=& Q\alpha_{\tau,t}(x)d\tau-\frac{1}{\gamma^2} h(x)\alpha_{\tau,t}(x)dY_\tau\qquad\hbox{for }\tau\leq t,\\
\alpha_{t,t}&\equiv&1,
\end{eqnarray*}
or in integrated form
\[\alpha_{\tau,t}(x) = 1-\int_\tau^tQ\alpha_{s,t}(x)ds + \frac{1}{\gamma^2}\int_\tau^t h(x)\alpha_{s,t}(x)dY_s.\]

\chapter{The Innovations Approach}

\noindent Let $X_t\in\mathcal S$ be a Markov process with generator $Q$. Let the domain of $Q$ be denoted by $\mathcal B(Q)$, and let $\mathcal B_b(Q)$ denote the subset of bounded functions in $\mathcal B(Q)$. For any function $g(x)\in\mathcal B_b(Q)$ we have the following limit: 
\[\frac{\mathbb E[g(X_{t+\Delta t})|X_t=x]-g(x)}{\Delta t} \rightarrow Qg(x)\]
as $\Delta t\searrow 0$, for any $x\in\mathcal S$. If $Q$ is densely-defined and its resolvent set includes all positive real numbers, then the Hille-Yosida theorem applies, allowing us to write the the distribution of $X_t$ with a contraction semi-group. In these notes we assume that such conditions hold and that the transition density/mass is generated by an operator semigroup denoted by $e^{Q^*t}$.

A standard nonlinear filtering problem in SDE theory assumes that $X_t$ is unobserved and that a process $Y_t$ is given by an SDE
\begin{eqnarray}
\label{eq:dY}
dY_t&=&h(X_t)dt+\gamma dW_t\qquad\hbox{observed}
\end{eqnarray}
where $W_t$ is an independent Wiener process, $\gamma>0$ and we assume that the $h$ is bounded. Let the filtration $\mathcal F_t^Y=\sigma\{Y_s:s\leq t\}$ so that for an integrable function $g(x)$ we have
\[\hat g_t\doteq \mathbb E[g(X_t)|\mathcal F_t^Y]\]
for all $t\in[0,T]$.

\section{Innovations Brownian Motion}
Let $\nu_t$ denote the \textit{innovations process} whose differential is given as follows
\[d\nu_t = dY_t - \hat h_tdt.\]
with $\nu_0=0$, and where $\hat h_t = \mathbb E[h(X_t)|\mathcal F_t^Y]$. 

\begin{proposition} The process $\nu_t/\gamma$ is an $\mathcal F_t^Y$ Brownian motion. \end{proposition}

\begin{proof} Is is clear that $\nu_t$ is (i) $\mathcal F_t^Y$-measurable, continuous and square integrable on $[0,T]$. To show that it is a local martingale, we take expectations for any $s\leq t$ as follows
\[\mathbb E[\nu_t|\mathcal F_s^Y]-\nu_s=\mathbb E\left[Y_t-\int_0^t\hat h_\tau d\tau\Big|\mathcal F_s^Y\right]-\left(Y_s-\int_0^s\hat h_\tau d\nu_\tau\right)\]

\[=\mathbb E\left[\gamma(W_t-W_s)+\int_s^th(X_\tau)d\tau\Big|\mathcal F_s^Y\right]-\mathbb E\left[\int_0^t\hat h_\tau d\tau\Big|\mathcal F_s^Y\right]+\int_0^s\hat h_\tau d\nu_\tau\]

\[=\mathbb E\left[\gamma(W_t-W_s)+\int_s^th(X_\tau)d\tau\Big|\mathcal F_s^Y\right]-\mathbb E\left[\int_s^t\hat h_\tau d\nu_\tau\Big|\mathcal F_s^Y\right]\]

\[=\gamma\mathbb E\left[W_t-W_s\Big|\mathcal F_s^Y\right]+\mathbb E\left[\int_s^th(X_\tau)d\tau-\int_s^t\hat h_\tau d\nu_\tau\Big|\mathcal F_s^Y\right]=0.\]
Furthermore, the cross-variation of $\nu_t/\gamma$ is the same as the cross-variation of $W_t$, and so by the L\'evy characterisation of Brownian motion (see Karatzas and Shreve \cite{K-S}), $\nu_t/\gamma$ is also $\mathcal F_t^Y$ Brownian motion.

\end{proof}

Given that $\nu_t$ is $\mathcal F_t^Y$-Brownian motion, it may seem obvious that any $L^2$-integrable and $\mathcal F_T^Y$ measurable random variable has an integrated representation in terms of $\nu$, but it is not easily seen that $\mathcal F_t^Y = \sigma\{\nu_s:s\leq t\}$. The following proposition provides a proof to verify that it is indeed true. 
\begin{proposition}
\label{prop:martRep}
Every square integrable random variable $N$ that is $\mathcal F_T^Y$-measurable, has a representation of the form
\[N=\mathbb EN+\frac{1}{\gamma}\int_0^T f_sd\nu_s\]
where $f = \{f_s:s\leq T\}$ is progressively measurable and $\mathcal F_t^Y$-adapted and $\mathbb E\left[\int_0^Tf_s^2ds\right]<\infty$.
\end{proposition}
\begin{proof} For all $t\in[0,T]$, define $Z_t\doteq \exp\left\{ -\frac{1}{\gamma^2}\int_0^t\hat h_sd\nu_s-\frac{1}{2\gamma^2}\int_0^t\hat h_s^2ds\right\}$. Clearly, $Z_t$ is an $\mathcal F_t^Y$-martingale, and so we can define an equivalent measure $\tilde{\mathbb P}$ with the following Radon-Nikodym derivative,
\[\frac{d\tilde{\mathbb P}}{d\mathbb P} = Z_T.\]
As a consequence of the Girsanov theorem, $Y_t/\gamma$ is a $\tilde{\mathbb P}$-Brownian motion. Then apply the martingale representation theorem,
\[Z_T^{-1}N = \tilde{\mathbb E}[Z_T^{-1}N]+\frac{1}{\gamma}\int_0^Tq_sdY_s\]

\[=\tilde{\mathbb E}[Z_T^{-1}N]+\frac{1}{\gamma}\int_0^Tq_sd\nu_s+\frac{1}{\gamma}\int_0^Tq_s\hat h_sds\]
where $q = \{q_s:s\leq T\}$ is adapted and $\tilde{\mathbb P}\left(\int_0^tq_s^2<\infty\right)=\mathbb P\left(\int_0^tq_s^2<\infty\right)=1$. From here we can construct a $\tilde{\mathbb P}$-martingale from $Z_T^{-1}N$, 
\[\tilde N_t \doteq \tilde{\mathbb E}[Z_T^{-1}N|\mathcal F_t^Y]\]
and applying It\^o's lemma to $\tilde N_tZ_t$ we have
\[d\left( \tilde N_tZ_t \right)=-\frac{1}{\gamma^2}Z_t\tilde N_t\hat h_td\nu_t+\frac{1}{\gamma}Z_tq_td\nu_t+\frac{1}{\gamma}Z_t\hat h_tq_tdt-\frac{1}{\gamma}Z_t\hat h_tq_tdt\]

\[=-\frac{1}{\gamma^2}Z_t\tilde N_t\hat h_td\nu_t+\frac{1}{\gamma}Z_tq_td\nu_t.\]
Integrating from $0$ to $t$, we have
\[\tilde N_tZ_t = \mathbb EN+\frac{1}{\gamma}\int_0^t\left(q_s-\frac{1}{\gamma}\tilde N_s\hat h_s\right)Z_sd\nu_s,\]
and therefore, setting $f_t =\left ( q_t-\frac{1}{\gamma}\tilde N_t\hat h_t\right)Z_t$ for all $t\in [0,T]$, we have a unique $\mathcal F_t^Y$-adapted representation in terms of $\nu$, and the proposition is proved.

\end{proof}
\begin{remark} For a more general proof of proposition \ref{prop:martRep}, see proposition 2.31 on page 34 of Bain and Crisan \cite{bainCrisan}. They present the proof for generalized systems on unbounded time intervals, with merely the conditions that $\mathbb E\left[\int_0^th^2(X_s)ds\right]<\infty$ and $\mathbb P\left(\int_0^t\hat h_s^2ds<\infty\right)=1$ for all $t<\infty$.
\end{remark}

\section{The Nonlinear Filter}
In deriving the nonlinear filter with the innovations Brownian motion, it will be important to use the following martingale for any given $g\in \mathcal B_b(Q)$:
\[N_t\doteq \hat g_t - \int_0^t\widehat{Q g}_sds\]
for all $t\in [0,T]$.

\begin{lemma} $N_t$ is an $\mathcal F_t^Y$-adapted martingale.
\end{lemma}

\begin{proof}
It suffices to show that $\mathbb EN_t = N_0$ for any $t\in [0,T]$, which we do as follows:
\[\mathbb EN_t = \mathbb E\left[\hat g_t- \int_0^t\widehat{Q g}_sds\right]=\mathbb Eg_t-\int_0^t\mathbb E\left[\mathbb E[Qg(X_s)|\mathcal F_s]\right]ds\]

\[=\mathbb Eg_t-\int_0^t\mathbb EQg(X_s)ds=\mathbb E\left[g_t-\int_0^tQg(X_s)ds\right]=\mathbb Eg(X_0)=N_0.\]
\end{proof}
Knowing that $\hat g_t-\int_0^t\widehat{Qg}_sds$ is an $\mathcal F_t^Y$-martingale, we will apply proposition \ref{prop:martRep} as follows
\[\hat g_t-\int_0^t\widehat{Qg}_sds =\hat g_0+\frac{1}{\gamma} \int_0^tf_sd\nu_s\]
where $f_t$ is an $\mathcal F_t^Y$-predicable process for any $t\in [0,T]$. From here, the main point in the derivation of the nonlinear filter is in finding the function $f_t$ in terms of quantities that are more readily computable.

\begin{theorem} \textbf{The Nonlinear Filter.} \label{thm:nlf}
For any function $g\in \mathcal B_b(Q)$, the nonlinear filter is given by the following SDE
\[d\hat g_t = \widehat{Qg}_tdt+ \frac{\widehat{gh}_t-\hat g_t\hat h_t}{\gamma^2}d\nu_t\]
for all $t\in [0,T]$.
\end{theorem}

\begin{proof} We can apply proposition \ref{prop:martRep} to $N_t$ and we get
\[N_t = \mathbb EN_0+\frac{1}{\gamma}\int_0^tf_sd\nu_s=\hat g_0+\frac{1}{\gamma}\int_0^tf_sd\nu_s,\]
thus defining the conditional expectation at time $t$ as
\[\hat g_t = \hat g_0+\int_0^t\widehat{Qg}_sds+\frac{1}{\gamma}\int_0^tf_sd\nu_s.\]
From here, to complete the proof only requires us to identify $f_t$ explicitly. For some process $\psi\in L^\infty[0,T]$, define $\xi_t$ such that 
\[d\xi_t = \frac{i}{\gamma}\xi_t\psi_tdY_t\]
with $\xi_0=1$. We then apply It\^o's lemma to the following
\begin{eqnarray}
\label{eq:d1}
d\left(\hat g_t\xi_t \right)&=& \widehat{Qg}_t\xi_tdt+\frac{1}{\gamma}f_t\xi_td\nu_t+\frac{i}{\gamma}\hat g_t\xi_t\psi_t\left(d\nu_t+\hat h_tdt\right)+i\xi_t\psi_tf_tdt\\
\label{eq:d2}
d\mathbb E[g(X_t)\xi_t] &=& \mathbb E\left[Qg(X_t)\xi_t\right]dt+\frac{i}{\gamma}\mathbb E\left[g(X_t)\xi_t\psi_th(X_t)\right]dt.
\end{eqnarray}
If we integrate the integrands in (\ref{eq:d1}) and (\ref{eq:d2}) from time $0$ to time $t$, take expectations, multiply both sides by $\gamma$, and then subtract one from the other, and we are left with 
\[\int_0^ti\psi_s\mathbb E\left[\xi_s\left(\gamma f_s-g(X_s)h(X_s)+\hat g_s\hat h_s\right)\right]ds=0.\]
Hence, for almost every $t\in[0,T]$, we have 
\[\mathbb E\left[\xi_t\left(\gamma f_t-g(X_t)h(X_t)+\hat g_t\hat h_t\right)\right]=0=\mathbb E\left[\xi_t\left(\gamma f_t-\mathbb E[g(X_t)h(X_t)|\mathcal F_t^Y]+\hat g_t\hat h_t\right)\right]\]
and since $\xi_t$ belongs to a complete set, must have
\[f_t =\frac{ \mathbb E[g(X_t)h(X_t)|\mathcal F_t^Y]-\hat g_t\hat h_t}{\gamma}=\frac{\widehat {gh}_t-\hat g_t\hat h_t}{\gamma}\]
which proves the theorem.
\end{proof}
Existence of the solutions to the filtering SDE in theorem \ref{thm:nlf} is consequence of the fact that $\hat g_t$ is on such solution. The uniqueness of solutions to the filtering SDE can be grouped in with proofs for uniqueness of Zakai equation (see Kurtz and Ocone \cite{kurtzOcone1988} or Rozovsky \cite{rozovsky1992}) because there is a one-to-one relationship between measure-valued solutions of the two SDEs.
\subsection{Correlated Noise Filtering}
Suppose that $W_t$ and $X_t$ are correlated so that for any function $g\in \mathcal B_b(Q)$ we have

\[\frac{1}{\gamma}\sum_n\Delta Y_{t_n}\Delta g( X_{t_n})\stackrel{p}{\longrightarrow}\left<W,g(X)\right>_t\doteq\int_0^t \rho_s^gds\]
where $0=t_0<t_1<\dots<t_n=t$ and the limit is taken as $\sup_n(t_{n+1}-t_n)\rightarrow 0$. Then there is an added term in equation (\ref{eq:d2}), 

\[d\mathbb E[g(X_t)\xi_t] = \mathbb E\left[Qg(X_t)\xi_t\right]dt+\frac{i}{\gamma}\mathbb E\left[g(X_t)\xi_t\psi_th(X_t)\right]dt+i\mathbb E\left[\xi_t\psi_t\rho_t^g \right]dt\]
and so for almost every $t\in[0,T]$ we have
\[\mathbb E\left[\xi_t\left(\gamma f_t-\mathbb E[g(X_t)h(X_t)|\mathcal F_t^Y]+\hat g_t\hat h_t-\gamma\hat \rho_t^g\right)\right]=0\]
and so $f_t = \frac{1}{\gamma}(\widehat{gh}_t-\hat g_t\hat h_t+\gamma\hat \rho_t^g)$, and the nonlinear filter is

\[d\hat g_t = \widehat{Qg}_tdt+\frac{\widehat{gh}_t-\hat g_t\hat h_t+\gamma\hat \rho_t^g}{\gamma^2}d\nu_t.\]

\subsection{The Kalman-Bucy Filter}
Another big advantage to the innovations approach is in linear filtering. In particular, the case when the filtering problem consists of a system of linear SDEs. In this case, one needs to take some steps to verify that the filtering distribution is normal, but after doing so it is straight-forward to derive equations for the posterior's first and second moments.

Consider a non-degenerate linear observations model such that $h(x)=h\cdot x$ and $\gamma>0$, with $X_0$ being Gaussian distributed, and with state-space generator
\[Q = \frac{\sigma^2}{2}\frac{\partial^2}{\partial x^2}~\cdot~+ax\frac{\partial}{\partial x}~\cdot~\]
for all functions $g\in C^2(\mathbb R)$, where $\sigma $ and $a$ are constant coefficients. The process
\[N_t \doteq \widehat X_t - \int_0^t\widehat{QX}_sds = \widehat X_t-a\int_0^t\widehat X_sds \]
is square-integrable and a martingale, and if we extend proposition \ref{prop:martRep} for $h(x) = h\cdot x$ (see proposition 2.31 on page 34 of Bain and Crisan \cite{bainCrisan}), we can then write $N_t$ using innovations Brownian motion,
\[N_t =\widehat X_0+\frac{1}{\gamma}\int_0^tf_sd\nu_s\]
where $f_t\in L^2[0,T]$ and is $\mathcal F_t^Y$-adapted. Now, by simple stochastic calculus we can verify that $X_t$ is given by
\[X_t = e^{at}X_0+\sigma\int_0^te^{a(t-s)}dB_s\]
where $B_t\perp W_t$ (i.e. the state-space noise is independent of the observation noise). Defining the estimation error $\epsilon_t$ as
\[\epsilon_t\doteq X_t-\widehat X_t.\]
and applying It\^o's lemma, we have
\[d\epsilon_t = a\epsilon_tdt+\sigma dB_t-\frac{1}{\gamma}f_td\nu_t,\]
which has a solution given by the integrating factor,
\[\epsilon_t =e^{at}\epsilon_0+\sigma\int_0^te^{a(t-s)}dB_s-\frac{1}{\gamma}\int_0^te^{a(t-s)}f_sd\nu_s,\]
which is Gaussian distributed and uncorrelated with $Y$
\[\mathbb E[Y_s\epsilon_t] =\mathbb E[Y_s(X_t-\widehat X_t)] =\mathbb E[Y_s\mathbb E[X_t-\widehat X_t|\mathcal F_t^Y]]=0\]
for all $s\leq t$. Now observe the following:
\begin{itemize}
\item $(X_t,Y_s)$ are jointly Gaussian for any $s\leq t$,
\item $\widehat X_t$ is Gaussian and a linear function of $\{Y_s\}_{s\leq t}$,
\item therefore, $(X_t,\widehat X_t,Y_s)$ are jointly Gaussian for any $s\leq t$
\item in particular $(\epsilon_t,Y_s)$ are jointly Gaussian and uncorrelated for any $s\leq t$.
\end{itemize}
Therefore, $\epsilon_t$ is independent of $\mathcal F_t^Y$ (for a more detailed discussion see \cite{bainCrisan,oxendale}), and so the filter is Gaussian . 

If we apply the filter in theorem \ref{thm:nlf} with $g(x)=x$, it yields $f_t = \frac{1}{\gamma}\mathbb E[\epsilon_t^2|\mathcal F_t^Y]=\frac{1}{\gamma} \mathbb E\epsilon_t^2$, and an SDE for the evolution of the first filtering moment
\begin{eqnarray}
\label{eq:KF}
d\widehat X_t&=&a\widehat X_tdt+\frac{h\cdot\mathbb E\epsilon_t^2}{\gamma^2}d\nu_t.
\end{eqnarray}
We can also apply It\^o's lemma and then take expectations, which will result in a Riccati equation for the evolution of the filter's covariance,
\begin{eqnarray}
\nonumber
d\mathbb E[\epsilon_t^2]&=&2a\mathbb E[\epsilon_t^2]dt+\sigma^2dt+\mathbb E[f_t^2]dt\\
\nonumber
&&\\
\label{eq:Riccati}
&=&2a\mathbb E[\epsilon_t^2]dt+\sigma^2dt-\frac{h^2\cdot\mathbb E[\epsilon_t^2]^2}{\gamma^2}dt
\end{eqnarray}
where we have used the fact that $d\nu_t\cdot dB_t=0$, and $\mathbb E[f_t^2]=\frac{1}{\gamma^2}\mathbb E\left[\mathbb E[\epsilon_t^2]\mathbb E[ \epsilon_t^2]\right]=\frac{1}{\gamma^2}\mathbb E[\epsilon_t^2]^2$. Equations (\ref{eq:KF}) and (\ref{eq:Riccati}) are the Kalman filter.

\chapter{Numerical Methods for Approximating Nonlinear Filters}

\noindent In practice, the exact  filter can only be computed for models that are completely discrete, or for discrete-time linear Gaussian models in which the Kalman filter applies. In the other cases, the consistency of approximating schemes can be relatively trivial, while for others there is a fair amount of analysis required. In this lecture we present the general theory of Kushner \cite{kushner08} regarding the consistency of approximating filters, and also the Markov chain approximation methods of Dupuis and Kushner \cite{dupuisKushner}. But first we present the following simple result regarding filter approximations:\\

\begin{example} Let $X_t$ be an unobserved Markov process, and let the observation process $Y_t$ be given by an SDE
\begin{eqnarray}
\nonumber
dY_t&=&h(t,X_t)dt+\gamma dW_t
\end{eqnarray}
where $W_t$ is an independent Wiener process, and $\gamma>0$. For any partition of the interval $[0,t]$ into $N$-many points, let $\mathcal F_t^N = \sigma\{Y_{t_n}:n\leq N\}$, and let the filtration generated by the continuum of observations be denoted by $\mathcal F_t^Y = \sigma\{Y_s:s\leq t\}$. Assuming that
\[\mathcal F_t^N\subset\mathcal F_t^{N+1}\subset\dots\dots\subseteq\mathcal F_t^Y,\]
then from the continuity of $Y_t$ we know that $\bigvee_{N=1}^\infty\mathcal F_t^N = \mathcal F_t^Y$. Then, by L\'evy's 0-1 law we know that the conditional expectations converge,

\[\lim_N\mathbb E[g(X_t)|\mathcal F_t^N]=\mathbb E\left[g(X_t)\Bigg|\bigvee_{N=1}^\infty\mathcal F_t^N\right]=\mathbb E[g(X_t)|\mathcal F_t^Y].\]
for any integrable function $g(x)$. This clearly shows that the filter with discrete observations can be a consistent estimator of the filter with a continuum of observations. However, computing the filter with discrete observations may still require some approximations.
\end{example}

\section{Approximation Theorem for Nonlinear Filters}
In this section we present a proof of a theorem that essentially says: filtering expectations of bounded functions can be approximated by filters derived from models who's hidden state converges weakly to the true state. The theorem is presented in the context of continuous-time process with a continuum of observations, but can be reapplied in other cases with relatively minor changes.

For some $T<\infty$ and any $t\in[0,T]$, let $X_t$ be an unobserved Markov process with generator $Q$ with domain $\mathcal B(Q)$, and let $\mathcal B_b(Q)$ denote the subset of bounded functions in $\mathcal B(Q)$. For any function $g\in\mathcal B_b(Q)$, we have the following limit
\[\frac{\mathbb E[g(X_{t+\Delta t})|X_t=x]-g(x)}{\Delta t} \rightarrow Qg(x)\]
as $\Delta t\searrow 0$. The observed process is $Y_t$ is given by an SDE
\begin{eqnarray}
\label{eq:dY}
dY_t&=&h(X_t)dt+\gamma dW_t
\end{eqnarray}
where $W_t$ is an independent Wiener process, $\gamma>0$ and $h$ is a bounded function. For any time $t\in [0,T]$, let $\mathcal F_t^Y = \sigma\{Y_s:s\leq t\}$. From our study of Zakai equation we know that we can write the filtering expectation as
\begin{eqnarray}
\label{eq:contFilter}
\mathbb E[g(X_t)|\mathcal F_t^Y] &=& \frac{\mathbb E[g(\tilde X_t)\tilde M_t|\mathcal F_t^Y]}{\mathbb E[\tilde M_t|\mathcal F_t^Y]}
\end{eqnarray}
where the paths of $\tilde X$ have the same law as those of $X$ but are independent of $(X,Y)$, and with $\tilde M_t$ being the likelihood ratio
\[\tilde M_t \doteq \exp\left\{\frac{1}{\gamma^2}\int_0^th(\tilde X_s)dY_s-\frac{1}{2\gamma^2}\int_0^th^2(\tilde X_s)ds\right\}.\]
The Zakai equation can provide us with an SDE for filtering expectations, but direct numerical quadrature methods to compute the Zakai equation may be difficult to justify.

\subsection{Weak Convergence}

In order to approximate the nonlinear filter, we will look to approximate $X$ with a family of process $\{X^n\}$ which converge weakly to $X$. Let $\{\mathbb P^n\}$ be a family of measures on a metric space $\mathcal D$.
\begin{definition} \textbf{Weak Convergence.} $\mathbb P^n$ is said to converge weakly to $\mathbb P$ if 
\[\int f(x)d\mathbb P^n(x)\rightarrow \int f(x)d\mathbb P(x)\qquad\hbox{as }n\rightarrow \infty\]
for any bounded continuous function $f:\mathcal D\rightarrow \mathbb R$. For the induced processes $\{X^n\}_n$, we denote weak convergence by writing $X^n\Rightarrow X$.
\end{definition}

\begin{definition}\textbf{Tightness.} We say that the family of measures $\{\mathbb P^n\}_n$ is tight if for any $\epsilon>0$ there exists a compact set $K_\epsilon\subset \mathcal D$ such that
\[\inf_n\mathbb P^n(K_\epsilon)\geq 1-\epsilon.\]
If so we also say that the induced processes $\{X^n\}_n$ are tight. 
\end{definition}
\noindent An important result regarding tightness is Prokhorov's theorem:
\begin{theorem}\textbf{Prokhorov.} If $\mathcal D$ is a complete and separable metric space, then the family $\{\mathbb P^n\}_n$ contained in the space of all probability measure on $\mathcal D$ is relatively compact in the topology of weak convergence iff it is tight.
\end{theorem}

For the purposes of our study, we will consider processes $X$ which are right continuous with left-hand limits (c\'adl\'ag). Let $D[0,T]$ denote the space of c\'adl\'ag functions from $[0,T]$ to $\mathbb R$, equipped with Skorohod topology. If a family of probability measures $\{\mathbb P^n\}_n$ is tight, then weak convergence of $\mathbb P^n$ to the measure on $X$ can be shown by verifying that the laws of the induced processes $\{X^n\}_n$ converge to the law of solutions to the associated martingale problem
\[\mathbb E^n\left[g(X_t^n)-g(X_s^n)-\int_s^tQg(X_\tau^n)d\tau\bigg|\mathcal F_s\right]\rightarrow 0\]
in probability as $n\rightarrow \infty$ for any $g(x)\in\mathcal B_b(Q)$, and for any $s\leq t$. It can be shown that the law of the solution to this martingale problem is unique.

Given a family of measure $\{\mathbb P^n\}_n$, an extremely useful tool is the Skorohod Representation Theorem, which says the following:
\begin{theorem}\textbf{Skorohod Representation.}
Let $\{\mathbb P^n\}_n$ be a family of measures on a complete and separable metric space. If $\mathbb P^n$ converges weakly to $\mathbb P$, then there is a probability space $(\tilde\Omega,\tilde{\mathcal F},\tilde {\mathbb P})$ with random variables $\tilde X^n$ and $\tilde X$ for which
\begin{itemize}
\item $\tilde {\mathbb P}(X^n\in A) = \mathbb P^n(A)$ for all $n$ and any set $A\subset D[0,T]$,
\item $\tilde {\mathbb P}(X\in A) = \mathbb P(A)$ for any set $A\subset D[0,T]$,
\item and $\tilde X^n\rightarrow \tilde X$ $\qquad\tilde{\mathbb P}$-a.s. as $n\rightarrow \infty$.
\end{itemize}
\end{theorem}
\noindent In particular, if $\tilde X^n\rightarrow \tilde X\in C[0,T]$ a.s. in the Skorohod topology, then the convergence holds uniformly in $t$. For a detailed treatment of weak convergence, the Skorohod topology, martingale problems, and Skorohod representations, see the book by Ethier and Kurtz \cite{ethierKurtz} and the book by Yin and Zhang \cite{yinZhangCont}.

%

\subsection{Consistency Theorem}
Let $\{X^n\}_n$ be a family of random variables on the same probability space as $(X,Y)$, which also converge weakly to $X$. Let $\tilde X^n$ denote a copy of $X^n$ that is independent of $(X,Y)$, let the approximated likelihood be denoted by $\tilde M_t^n$,
\[\tilde M_t^n \doteq \exp\left\{\frac{1}{\gamma^2}\int_0^th( \tilde X_s^n)dY_s-\frac{1}{2\gamma^2}\int_0^th^2(\tilde X_s^n)ds\right\}\]
and define the approximated filtering expectation 
\begin{equation}
\label{eq:approxFilter}
\mathcal E_t^n\left[g(X_t)\right] \doteq \frac{\mathbb E[g(\tilde X_t^n)\tilde M_t^n|\mathcal F_t^Y]}{\mathbb E[\tilde M_t^n|\mathcal F_t^Y]}\qquad\forall t\in[0,T]
\end{equation}
for any function $g(x)\in\mathcal B_b(Q)\cap C(\mathcal S)$. We then have the following theorem:
\begin{theorem} 
\label{thm:mainThm}
Given a family of process $\{X^n\}$ taking values in $D[0,T]$ which converge weakly in the Skorohod topology to $X\in C[0,T]$, the approximated filter in (\ref{eq:approxFilter}) will converge uniformly
\[\sup_{t\leq T}\left|\mathcal E_t^n\left[g(X_t)\right] -\mathbb E[g(X_t)|\mathcal F_t^Y]   \right|\rightarrow 0,\qquad\hbox{for }X_t\in C[0,T]\]
in probability and in mean as $n\rightarrow\infty$, for any function $g(x)\in\mathcal B_b(Q)\cap C(\mathcal S)$.
\end{theorem}

\begin{proof} (taken from \cite{kushner08}) For the purposes of this proof, we can neglect the $h^2$ terms in $M_t$ and $M_t^n$. Let $\zeta^1$ and $\zeta^2$ be bounded processes independent of $W$, and consider the following estimate,

\begin{equation}
\label{eq:V}
V = \mathbb E\sup_{t\leq T}\left| \exp\left\{\int_0^t\zeta_s^1dW_s \right\}-\exp\left\{\int_0^t\zeta_s^2dW_s \right\}    \right|.
\end{equation}
We will use the following inequality that holds for real numbers $a$ and $b$,
\begin{equation}
\label{eq:expIneq}
|e^a-e^b|\leq |a-b|(e^a+e^b).
\end{equation}
For any real-valued sub-martingale $N_t$, we have the following inequality (see \cite{K-S}),
\begin{equation}
\label{eq:subMartIneq}
\mathbb E\sup_{t\leq T}N_t^2\leq4\mathbb EN_T^2.
\end{equation}
Inequality (\ref{eq:expIneq}) and a Schwarz inequality applied to (\ref{eq:V}) yields,
\begin{equation}
\label{eq:V2}
V^2\leq \mathbb E\sup_{t\leq T}\left| \int_0^t(\zeta_s^1-\zeta_s^2)dW_s  \right|^2\times \mathbb E\sup_{t\leq T}\left| \exp\left\{\int_0^t\zeta_s^1dW_s \right\}+\exp\left\{\int_0^t\zeta_s^2dW_s \right\}    \right|^2
\end{equation}
By (\ref{eq:subMartIneq}) the first term in (\ref{eq:V2}) is bounded by $4\mathbb E \int_0^T|\zeta_s^1-\zeta_s^2|^2ds$. To bound the second term we use the fact that
\[\mathbb E\exp\left\{\int_0^t\zeta_s^idW_s \right\} = \mathbb E \exp\left\{\frac{1}{2}\int_0^t|\zeta_s^i|^2ds \right\}\qquad\hbox{for }i=1,2\]
along with (\ref{eq:subMartIneq}) and the fact that $\exp\left\{\frac{1}{2}\int_0^t\zeta_s^idW_s \right\}$ are bounded sub-martingales. Using these facts we can find a constant $C$ that depends on $T$, $\zeta^1$ and $\zeta^2$, such that
\begin{equation}
\label{eq:Vbound}
V^2\leq C\cdot\mathbb E\int_0^T|\zeta_s^1-\zeta_s^2|^2ds.
\end{equation}

To prove the theorem is suffices to show that 
\[\sup_{t\leq T}\left|\mathbb E[g(\tilde X_t^n)\tilde M_t^n|\mathcal F_t^Y]-\mathbb E[g(\tilde X_t)\tilde M_t|\mathcal F_t^Y]\right|\rightarrow 0\]
in probability as $n\nearrow \infty$. From the boundedness of $h$ and $T$, we know that 
\begin{equation}
\label{eq:Mbound}
\mathbb E\sup_{t\leq T}\left(M_t^n\right)^2+\mathbb E\sup_{t\leq T}\left(M_t\right)^2<\infty.
\end{equation}
Let $\tilde X$ be a copy of $X$ that is independent of $(X,Y)$. By the Skorokhod representation theorem we can assume W.L.O.G. that $\{\tilde X^n\}_n$ are defined on the same probability space as $(\tilde X, X,Y)$, that each $\tilde X^n$ is independent of $(X,Y)$, and that $X^n\rightarrow X$ a.s. In fact, because we have assume that $X^n$ converges to a continuous function on the Skorohod topology, we know that the convergence is uniform,
\[\sup_{t\leq T}|\tilde X_t^n-\tilde X_t|\rightarrow 0,\qquad\hbox{a.s.}\]
as $n\rightarrow\infty$. In particular, $\sup_{t\leq T}|g(\tilde X_t^n)-g(\tilde X_t)|\tilde M_t^n\rightarrow 0$ a.s.

Taking expectations, we have
\[\mathbb E\sup_{t\leq T}\left|\mathbb E[g(\tilde X_t^n)\tilde M_t^n|\mathcal F_t^Y]-\mathbb E[g(\tilde X_t)\tilde M_t|\mathcal F_t^Y]\right|\]
\[\leq\mathbb E\sup_{t\leq T}\mathbb E\left[|g(\tilde X_t^n)-g(\tilde X_t)|\cdot\tilde M_t^n\Big|\mathcal F_t^Y\right]+\|g\|_\infty\mathbb E\sup_{t\leq T}\mathbb E\left[|\tilde M_t^n-\tilde M_t|\Big|\mathcal F_t^Y\right]\]

\[\leq\mathbb E\left[\mathbb E\left[\sup_{t\leq T}\left(|g(\tilde X_t^n)-g(\tilde X_t)|\cdot\tilde M_t^n\right)\Big|\mathcal F_t^Y\right]\right]+\|g\|_\infty\mathbb E\left[\mathbb E\left[\sup_{t\leq T}|\tilde M_t^n-\tilde M_t|\Big|\mathcal F_t^Y\right]\right]\]
\begin{equation}
\label{eq:limit}
\leq\mathbb E\left[\sup_{t\leq T}\left(|g(\tilde X_t^n)-g(\tilde X_t)|\cdot\tilde M_t^n\right)\right]+C\|g\|_\infty\mathbb E\left[\int_0^T\left|h(\tilde X_s^n)-h(X_s)\right|^2ds\right].
\end{equation}
The first term in (\ref{eq:limit}) goes to zero as $n\rightarrow\infty$ after applying the bound in (\ref{eq:Mbound}) and then by calling dominated convergence. The second term in (\ref{eq:limit}) goes to zero because of the bound in (\ref{eq:Vbound}). Therefore, the approximated filter converges in mean, and in probability as well.
\end{proof}

\begin{corollary} To generalize theorem \ref{thm:mainThm}  for any $X\in D[0,T]$ such that $X^n\Rightarrow X$, we simply need to rework the end of the proof to show that the limit holds pointwise,
\[\left|\mathcal E_t^n\left[g(X_t)\right] -\mathbb E[g(X_t)|\mathcal F_t^Y]   \right|\rightarrow 0,\qquad\hbox{for }X_t\in D[0,T]\]
in probability and in mean as $n\rightarrow\infty$, almost everywhere $t\in [0,T]$, and for any function $g(x)\in\mathcal B_b(Q)\cap C(\mathcal S)$.

\end{corollary} 

\section{Markov Chain Approximations}
Theorem \ref{thm:mainThm} applies directly when observations are available as often as needed. In addition, a continuum of observations allows us a certain amount of flexibility in our choice of approximation scheme. 
\subsection{Approximation of Filters for Contiuous-Time Markov Chains}
Let $X_t$ be a finite-state Markov chain. Consider a time step $\Delta t=\frac{1}{n}$. For $n$ finite, take $k=0,1,2,3,4,\dots,T/\Delta t-1$ and denote the Markov chain $\xi_k^n$ with transition probabilities 
\[\mathbb P(\xi_{k+1}^n=i|\xi_k^n=j) = \left[e^{Q^*\Delta t}\right]_{ij}\]
for any $i,j\in\{1,\dots,m\}$. The Markov chain $X_k^n$ is discrete but could be extended to $D[0,T]$ by taking $X_t^n = \sum_k\xi_k^n\mathbf 1_{t\in[t_k,t_{k+1})}$, \textbf{but this is not a continuous-time Markov chain.} The non-Markov structure of $X_t^n$ does not prevent us from applying theorem \ref{thm:mainThm},  but showing tightness and convergence of the martingale problem will be easier if we can find a Markovian approximation.

Let $k=1,2,3,4,\dots\dots$, take $\nu_k\sim iid\exp(1)$, and set $\tau_k = \tau_{k-1}+\nu_k\Delta t$ with $\tau_0=0$. Now let $\xi_k^n$ be the discrete Markov chain that we have already defined, but now consider all $k$ up until $\tau_{k+1}\geq T$. A continuous-time approximation of $X_t$ is then 
\[X_t^n = \sum_{k:\tau_k< T}\mathbf 1_{t\in [\tau_k,\tau_{k+1})}\xi_k^n.\]
To show that $\{ X^n\}_n$ is compact in $D[0,T]$ we proceed as follows:\\

\noindent Take any $f\in D[0,T]$ and for any $i = 01,2,3,4,\dots$ let $\mathcal J_i$ denote the time of the $ith$ jump,
\[\mathcal J_{i+1}(f) = \inf\{t>\mathcal J_i(f):f_t^n\neq f_{\mathcal J_i}^n\}\wedge T,\]
with $\mathcal J_0(f) = 0$. For general $f\in D[0,T]$ these $\mathcal J_i$'s may by infinitesimally small, but they will be informative for processes which approximated continuous-time Markov chains. 

Next, for any $\delta >0$ define the set
\[A_\delta = \left\{f\in \{1,\dots,m\}~~s.t.~~~|f_{\mathcal J_{i+1}}-f_{\mathcal J_i}|\geq \delta ~~\forall \mathcal J_{i+1}\leq T\right\}.\]
We can easily check that 
\[\mathbb P(X^n\in A_\delta^c)\leq 1- e^{\delta\max_jQ_{jj}}\leq -\delta\max_jQ_{jj}\ll1\]
for any $\delta $ small enough. Furthermore, for any sequence $\{f^n\}_n\subset A_\delta$ we can find a subsequence such that for any $i$ we have
\[\mathcal J_i(f^{n_\ell})\rightarrow \mathcal T_i\in [i\delta,T]\]
and
\[f_{\mathcal J_i(f^{n_\ell})}^{n_\ell}\rightarrow a_i\in \{1,\dots,m\}\]
as $\ell\rightarrow\infty$. Therefore, $f_s^{n_\ell}\rightarrow a_i$ for $s\in [\mathcal T_i,\mathcal T_{i+1})$, which shows that there is subsequence that converges point-wise. Therefore, since $D[0,T]$ is equipped with the point-wise metric it follows that $\bar A_\delta$ is  compact.

Now, if we look at the martingale problem, we have
\[\mathbb Eg( X_t^n) = \sum_{k:k\Delta t\leq t}\mathbb E\Delta g( X_{t_k})+\mathbb Eg(\bar X_0)\]

\[= \sum_{k:k\Delta t\leq t}\mathbb EQ g(X_{t_k})\Delta t+\mathbb Eg(X_0)+O(\Delta t)\]
which converges to the martingale problem as $\Delta t\searrow 0$. Therefore, $ X^n\Rightarrow X$ weakly in $D[0,T]$.

Then for any $s,t\in[0,T]$ with $s<t$, the change in the likelihood ratio for any path is given by
\[M_{s,t}^n = \exp\Bigg\{\frac{1}{\gamma^2}\sum_{k:s<\tau_k\leq t}h(X_{(\tau_{k-1}\vee s)}^n)(Y_{\tau_k}-Y_{(\tau_{k-1}\vee s)})\qquad\qquad\]

\[\qquad\qquad\qquad-\frac{1}{2\gamma^2}\sum_{k:s<\tau_k\leq t}h^2(X_{\tau_{k-1}\vee s})(\tau_k-(\tau_{k-1}\vee s))\Bigg\}.\]
At time $t$, let the approximating filtering mass function be denoted by $\omega_t = (\omega_t^1,\dots,\omega_t^m)$. Given $\omega_s$ we have the following recursion for the filtering mass:
\[\omega_t^i = \frac{\sum_j\mathbb E\left[\mathbf 1_{\tilde X_t=i}\tilde M_{s,t}^n\Big|\mathcal F_t^Y\vee\{\tilde X_s=j\}\right]\omega_s^j}{\sum_j\mathbb E\left[\tilde M_{s,t}^n\Big|\mathcal F_t^Y\vee\{\tilde X_s=j\}\right]\omega_s^j}.\]
\subsection{Approximation of Filter for System of SDEs}
Let the filtering problem be as follows
\begin{eqnarray}
\nonumber
dX_t&=&a(X_t)dt+\sigma(x) dB_t\\
\nonumber
dY_t&=&h(X_t)dt+\gamma dW_t
\end{eqnarray}
with $h$ bounded, $\gamma>0$, and $W_t\perp B_t$. If $\sigma^2(x)\geq \frac{1}{n}|a(x)|$ for all $x\in\mathbb R$, we can approximate $X$ with the a Markov chain $X^n$ taking paths in $D[0,T]$ and $X_t^n$ taking values in $\mathcal S^n=\{0,\pm\frac{1}{n},\pm\frac{2}{n},\dots\}$, defined as follows: $\Delta t^n(x) = \frac{1}{n^2\sigma^2(x)}$, and given $X_t^n=x$ the conditional probability distribution at time $t+\Delta  t^n(x)$ is given by
\[\mathbb P\left(X_{t+\Delta t^n(x)}^n=x\pm\frac{1}{n}\bigg|X_t^n=x\right) = \frac{\sigma^2(x)\pm \frac{1}{n}|a(x)|}{2\sigma^2(x)}.\]
We can construct a continuous-time Markov chain from this discrete-time Markov chain by using exponential arrivals as we did in with finite-state Markov chains. The subsequent process $\bar X_t^n$ has the following differential,

\[\bar X_t^n = X_0 +\int_0^t a(\bar X_s^n)ds+\int_0^t\sigma(\bar X_s^n)d\omega_s+\epsilon^n(t)\]
where $\omega_t$ is an independent Wiener process and $\epsilon^n(t)$ is a semi-martingale such that 
\[\mathbb E\sup_{t\leq T}(\epsilon^n(t))^2\rightarrow 0\]
 as $n\rightarrow\infty$, and by theorem 2.7b on page 27 of Ethier and Kurtz \cite{ethierKurtz}, $\{\bar X^n\}_n$ are tight for all $T<\infty$.
This process is a locally consistent approximation to $X$ on $[0,T]$. 

\subsection{Discrete-Time Obsevations}
For general discrete observations models, theorem \ref{thm:mainThm} applies for all $\{t_n\}_n\subset[0,T]$ for which each $t_n$ is an observation time. For simplicity, suppose that $Y_t$ is unobserved for $t\in (0,1)$ and that the only observations are available at times $t=0$ and $t=1$,
\[Y_1= Y_0+\int_0^1h(X_s)ds+W_1.\]
and assume that $X_t$ is a Markov chain on a finite state-space. We write the filtering mass recursively as
\[\pi_1(x)=\mathbb P(X_1=x|Y_0,Y_1)\]

\[= \frac{1}{c}\sum_{v\in\mathcal S}\mathbb E\left[\mathbb P(Y_1|Y_0,\{\tilde X_s\}_{s\leq 1})\Big|\mathcal F_1^Y\vee\{\tilde X_1=x,\tilde X_0=v\}\right]e^{Q^*1}(x|v)\pi_0(v)\]
where $c$ is a normalizing constant, $\mathcal F_1^Y=\sigma\{Y_0,Y_1\}$, the process $\tilde X$ is a copy of $X$ that is independent from $(X,Y)$, $e^{Q^*t}(~\cdot|~\cdot~)$ is the transition kernel of $X$, and the likelihood function is
\[\mathbb P(Y_1|Y_0,\{x_s\}_{s\leq 1}) = \exp\left\{-\frac{1}{2}\left(\frac{Y_1-Y_0-\int_0^1h(x_s)ds}{\gamma}\right)^2  \right\}\]
for any path $\{x_s\}_{s\leq 1}$. We approximate $X$ with a discrete-time Markov Chain $X^n$ such that
\[\mathbb P(X_{k+1}^n=x|X_k^n=x') = e^{Q^*/n}(x|x')\]
and in this case, theorem \ref{thm:mainThm} applies at time $t=1$ because
\[ \frac{1}{n}\sum_kh(X_k^n)\Rightarrow \int_0^1h(X_s)ds\]
as $n\rightarrow\infty$.

In theory, the approximated filter will converge, however there are still computational issues because as $n$ gets smaller we will need to devise a method to compute the expected likelihood
\[\mathbb E\left[\exp\left\{-\frac{1}{2}\left(\frac{Y_1-Y_0-\frac{1}{n}\sum_{k=0}^{n-1}h(\tilde X_k^n)}{\gamma}\right)^2  \right\}\bigg|\mathcal F_1^Y\vee\{\tilde X_n^n=x,\tilde X_0^n=v\}\right]\]
where $\tilde X^n$ is a copy of $X^n$ that is independent of $(X,Y)$. Monte Carlo methods can also be used, but the conditioning of the expectation on $X_n^n$ might slow the convergence. 

\chapter{Linear Filtering}

\noindent Filtering in general linear models is perhaps the most widely applied branch of filtering, but in the context of linearity the term `filtering' refers to something that is fundamentally different from the probabilistic models and equations that comprise what mathematicians refer to as `filtering theory.' The methods are not Bayesian, and probability's involvement can be minimal at times, but some of the most important ideas in filtering theory, such as the use of \textit{innovations}, can be traced back to their pragmatic roots in signal processing and `linear' filtering.

\section{General Linear Filters}
Let the integer $n\in\{0,2,\dots,N-1\}$ denote a time index. The simplest way to present a filtering problem is to identify a given measurement as a signal plus noise,
\[Y_n=X_n+W_n\]
where $W$ is a noise component with positive covariance $R$,
\[R_\ell\doteq \mathbb EW_{n\pm\ell}W_n\]
for some integer-valued lag $\ell$. The noise can be considered idiosyncratic, essentially meaning that it is orthogonal to $X$,

\[\mathbb EX_{n+\ell}W_n= 0\]
for any  time $n$ and any lag $\ell$. For any linear filter $H:\mathbb R^N\rightarrow \mathbb R^M$, the impulse response is its convolution with the measurement
\begin{eqnarray}
\nonumber
(H*Y)_k&=&(H*X)_k+(H*W)_k
\end{eqnarray}
where the convolution is a function of a shift $k$, 
\[(H*Y)_k = \sum_{n=0}^{M-1}X_nH_{n-k}.\]
If $k=0$, the convolution can be thought of as the inner-product. For the filter $H$, the signal-to-noise ratio (SNR) is defined as ratio of the signal response over the noise response,
\[SNR_H \doteq \frac{\|H*X\|^2}{\mathbb E\|H*W\|^2}.\]
The goal of linear filtering is to estimate $X$ with a projection of $Y$,
\[\widehat X \doteq H*Y,\]
or at least raise SNR so we are in a position better suited to make an estimate. We can consider such an estimate to be `optimal' if we have chosen $H$ for which SNR is maximized. If $\widehat X$ is an unbiased estimator, then the SNR of the optimal estimator will be greater than the SNR of the raw measurement,
\[SNR_Y = \frac{\|X\|^2}{\mathbb E\|W\|^2}<\frac{\|X\|^2}{\mathbb E\|\widehat X_n-X_n\|^2}\doteq SNR_{\widehat X}\]
and clearly, the major obstacles will be in finding the optimal linear filter. Obviously, if $\mathbb EX_nW_m=0$ for all $m,n$, then it would make sense to take $H$ to be some function that is known to have non-zero inner-product with $X$ but is also known to be orthogonal to $W$. It might be difficult find (let alone to invert) such a filter. More importantly, one should notice that maximizing the $SNR_{\widehat X}$ is the same as minimizing mean-square error (MSE),

\[MSE(\widehat X) = \mathbb E\|X-\widehat X\|^2\approx \frac{1}{N}\sum_n|X_n-\widehat X_n|^2.\]

\subsection{Reed/Matched Filters}
The Reed filter looks for the linear filter $H$ that maps $X\mapsto \mathbb R$ with optimal SNR. Since $R$ is positive-definite, there is an invertible matrix $A$ such that
\[R= AA'\]
which we can use along with the Cauchy-Schwarz inequality to get the bound on SNR for the linear filter,
\[SNR_H = \frac{|H'X|^2}{H'RH}=\frac{|H'AA^{-1}X|^2}{H'AA'H}=\frac{|(A'H)'(A^{-1}X)|^2}{(A'H)'(A'H)}\]

\[\leq \frac{|(A'H)'(A'H)|\cdot |(A^{-1}X)'(A^{-1}X)|}{(A'H)'(A'H)}\]

\[=X'(AA')^{-1}X=X'R^{-1}X\]
where $H'$, $X'$ and $A'$ are the transpose of there respective matrix/vector. The linear filter that achieves this upper bound is
\[H^{reed} = R^{-1}X,\]
yielding an optimal estimate as
\[\widehat X^{reed} = \arg\max_x\left(\frac{x'H^{reed}}{(Y-x)'H^{reed}}\right),\]
but this will require a search over the signal domain. However, if we can parameterize the domain of the signal, it will be possible to compress our search into a simpler procedure that requires us to merely test the SNR of relatively few parameters. For instance, if we know a priori that $X$ will have a significant response with a only a few of the Fourier basis functions, we can reduce an algorithm's search-time simply by searching over the domain of a few Fourier coefficients.
\subsection{Fourier Transforms and Bandwidth Filters}
\label{sec:bandwidth}
Fourier transforms and fast-Fourier transform (FFT) algorithms can easily be used as linear filters. The Fourier basis functions are can be used for the spectral decomposition of periodic functions, but we can without loss of generality extend an observed finite vector $(Y_0,\dots,Y_{N-1})$ into a periodic function simply by concatenating a backwards copy. With periodicity in hand, we can use the FFT to filter-out frequencies which we have determined apriori to not be part of the signal. In other words, we apply an FFT to the observed data and then reconstruct the signal by only considering the inverse FFT of the coefficients that are within a bandwidth known apriori to be where the signal resides.

Let $Y^*$ denote the Fourier transform of $Y$, defined as
\[ Y_k^* = \sum_{n=0}^{N-1} e^{\frac{2\pi i}{N}kn}Y_n\]
and the inverse Fourier transform
\[Y_n =\frac{1}{N} \sum_{k=0}^{N-1} e^{-\frac{2\pi i}{N}kn} Y_k^*\]
which allows us to reconstruct the measurement. The central idea in a bandwidth filter is the notion that the signal lives in specific range of frequencies. For instance, from linearity we have
\[Y^* = X^*+W^*\]
and if we know that the support of $X$'s Fourier coefficients is contained in a set $K_0\subset \{0,\dots,N-1\}$ such that
\[X_n=\frac{1}{N}\sum_{k\in K_0} e^{-\frac{2\pi i}{N}kn} X_k^*,\]
then we can construct an estimate based on the pertinent bandwidth(s).
\[\widehat X = \frac{1}{N} \sum_{k\in K_0} e^{-\frac{2\pi i}{N}kn} Y_k^*.\]
\begin{figure}[htbp] 
   \centering
   \includegraphics[width=5in]{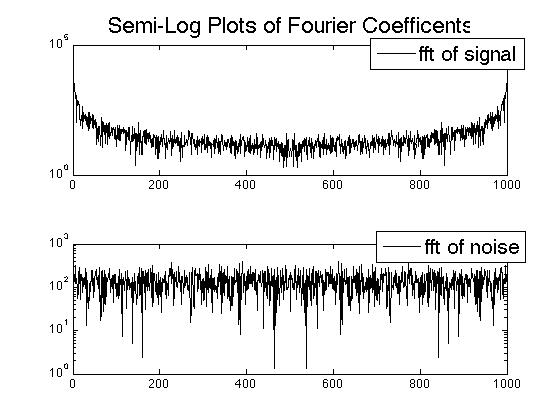} 
   \caption{\small The FFT of a random-walk considered to be the signal, and the FFT of an iid random variable that is considered to be the noise. The noise occupies the mid-level frequencies whereas the signal does not. Therefore, we construct and bandwidth filter simply by inverting the FFT without the mid-level coefficients.}
   \label{fig:bandwidths}
\end{figure}
\begin{figure}[htbp] 
   \centering
   \includegraphics[width=5in]{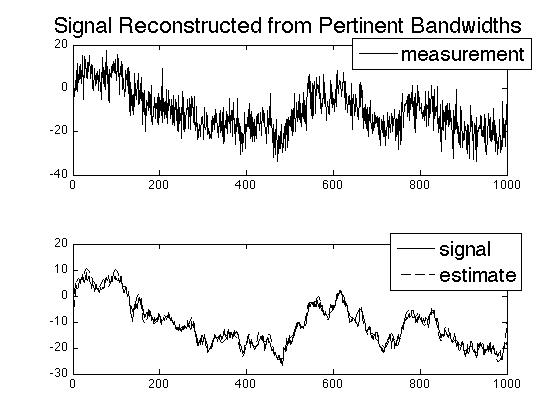} 
   \caption{\small The bandwidth filter where the mid-level frequencies are removed. The estimated signal is clearer than the raw measurement, and there is an increase is SNR, from $SNR_Y = \|X\|^2/\frac{1}{N}\sum_nW_n^2 = 7.53$, to $SNR_{\widehat X} = \|X\|^2/\frac{1}{N}\sum_n(X_n-\widehat X_n)^2=55.10$, and with $MSE = \frac{1}{N}\sum_n(X_n-\widehat X_n)^2=.0594$.}
   \label{fig:bandwidthFilter}
\end{figure}

For instance, suppose $X$ is a random-walk and $Y$ equals $X$ plus a considerable amount of noise,
\begin{eqnarray}
\nonumber
X_n&=&X_{n-1}+B_n\\
\nonumber
Y_n&=&X_n+\gamma W_n
\end{eqnarray}
where $ B_n$ and $W_n$ are independent white noises, and $\gamma>0$. We should try to identify the bandwidth(s) that contain the support of $X$'s Fourier coefficients simply by looking at the support of the FFT of a random-walk and comparing it to the FFT of noise. From figure \ref{fig:bandwidths} we see that the dominant Fourier coefficients of a random-walk are either in an extremely high or an extremely low bandwidth, whereas the Fourier coefficients of the noise are evenly distributed across all bandwidths. If we take $K_0$ to be the high and low frequencies that only contain noise, then the reconstructed signal will have a higher SNR,
\[SNR_Y=\frac{\|X^*\|^2}{\mathbb E\|W^*\|^2} =\frac{\|X^*\|^2}{\mathbb E\|Y^*-X^*\|^2} = \frac{\|X^*\|^2}{\sum_{k=0}^{n-1}\mathbb E|Y_k^*-X_k^*|^2} \]

\[= \frac{\|X^*\|^2}{\sum_{k\in K_0}\mathbb E|Y_k^*-X_k^*|^2+\sum_{k\notin K_0}\mathbb E|Y_k^*|^2}<\frac{\|X^*\|^2}{\sum_{k\in K_0}\mathbb E|Y_k^*-X_k^*|^2} \]

\[=\frac{\|X^*\|^2}{\sum_{k\in K_0}\mathbb E|\widehat X_k^*-X_k^*|^2} =\frac{\|X^*\|^2}{\mathbb E\|\widehat X^*-X^*\|^2}=SNR_{\widehat X}. \]
Indeed, as can be seen in figure \ref{fig:bandwidthFilter}, the signal becomes clearer as we eliminate the mid-level frequencies, and there is an increase in SNR from the 7.53 of the raw measurement, to 55.10 given by the bandwidth-filtered estimate, with a MSE=.0594.
\subsection{Wavelet Filters}
\label{sec:wavelets}
Wavelets are a tool that is useful in identifying the local behavior of a noise-corrupted signal. Measurements are often times contain adequate information for someone to decipher the underlying signal, usually because they can ignore noise and identify a movement in the measurement is caused by signal. This is precisely how a wavelet works: each wavelet represents a movement that the signal is capable of making, and any piece of the signal who's cross-product resonates with the wavelet is removed and placed in its respective spot as part of a noiseless reconstruction of the underlying signal.

A wavelet basis consists of a set of self-similar functions 
\[\psi_n^{k\ell} =2^{k/2}\psi_{2^{k}(n-\ell) }\]
for integers $k$ and $\ell$, where the unindexed function $\psi$ is the `mother-wavelet'. A useful wavelet has support that is small relative the length of the signal (e.g. $supp(\psi)\ll N$), and by construction should sum to zero and have norm 1,
\[\sum_{n\in supp(\psi)}\psi_n = 0,\quad\qquad\sum_{n\in supp(\psi)}|\psi_n|^2 = 1,\]
with $supp(\psi)$ denoting the support of the wavelet. The wavelets are indexed by $k$ and $\ell$ where $k$ is a dilation and $\ell$ is a translation. The indices of the wavelets are chosen to form an orthonormal basis,
\[\sum_{n=0}^{N-1}\psi_n^{k\ell}\psi_n^{k'\ell'}= \mathbf 1_{k=k'}\mathbf 1_{\ell=\ell'},\]
and like any other spectral method we can reconstruct a function from its wavelet transform,
\[Y _n= \sum_{k,\ell}\left<Y,\psi^{k\ell}\right>\psi_n^{k\ell}\]
for all $n\leq N-1$, with $\left<\cdot,\cdot\right>$ denoting inner-product. Not all wavelets can be used to form an orthonormal basis (e.g. the Mexican hat), but such wavelets should not be considered useless. Rather, a wavelet without an orthonormal basis simply requires that one use methods other than spectral decomposition.

For a signal processing problem, a particular wavelet is chosen for its generic resemblance to a local behavior of which the signal is capable. Essentially, we are taking a convolution of the function with wavelets of different thickness $k$, so that at each shift $\ell$ the local shape of the wavelet is given a chance to match itself to the input function. The wavelet-family that one uses will depend on the nature of the signal. Some possible wavelet to use in the construction of a discrete and orthogonal basis are symlets, coiflets, Daubechies, and Haar. The `mother wavelets' for symlet24, a Daubechies24, a coiflet5, and a Haar, are shown in figure \ref{fig:wavelets}. 
\begin{figure}[htbp] 
   \centering
   \includegraphics[width=5in]{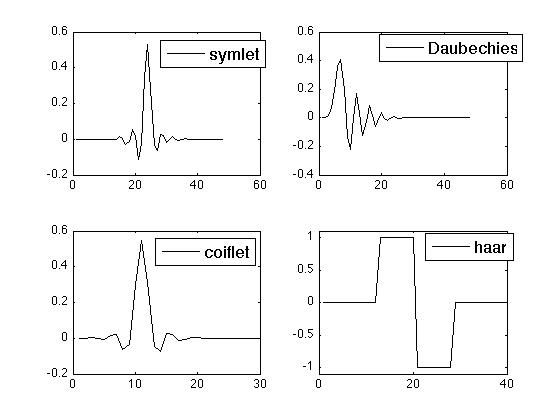} 
   \caption{\small Some examples of wavelets that can be used to construct an orthonormal basis.}
   \label{fig:wavelets}
\end{figure}
Symlet, coiflet and Daubechie wavelets can be defined with their respective degree of differentiability. For instance, a family of symlet-5 wavelets are generated by a mother-wavelet that has at least 6 derivatives and so three first 5 wavelet moments are vanishing,
\[\sum_{n\in supp(\psi)}\psi_nn^p=0\]
for $p=0,1,2,3,4,5$. In general, a wavelet that is $p+1$-times differentiability with fast enough decay in its tails will $p$-many vanishing moments.

When using wavelets to de-noise the random-walk example from section \ref{sec:bandwidth}, there are numerous choices to make such as which wavelets to use and at what parameter values will we be fitting noise and not the signal. Figure \ref{fig:denoising} shows the wavelets' ability to extract the signal from the noisy measurements, and table \ref{tab:waveletDenoising} shows how any of these wavelets does a better job de-noising the bandwidth filter from section \ref{sec:bandwidth}.
\begin{figure}[htbp] 
   \centering
   \includegraphics[width=5in]{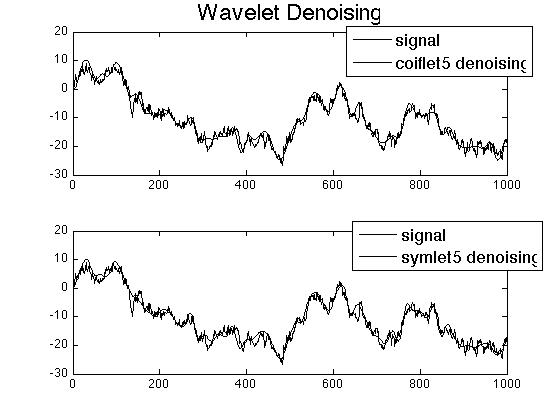} 
   \caption{\small The wavelets' de-noising of the measurement to uncover the signal is effective if we use the appropriate wavelet-family, and if we know what levels of the basis are not associated with the signal.}
   \label{fig:denoising}
\end{figure}

\begin{table}[!h!b!p]
\label{tab:waveletDenoising}
\caption{\textbf{Wavelet Denoising in Random-Walk Example.}}
\begin{center}
\begin{tabular}{c|c|c}
wavelet&SNR&MSE\\
\hline
coif2&64.90& 0.0549\\
sym2&57.89&0.0581\\
db2&57.89&0.0581\\
\hline
coif5&59.32&0.0573\\
sym5&58.38&0.0578\\
db5&57.78&0.0581\\
\hline
haar&56.09&0.0589
\end{tabular}
\end{center}
\end{table}


\section{Linear Gaussian Models}
A special case is when the impulse response is a linear model with Gaussian noise,
\[Y_n = X_n+W_n\]
where $R\doteq\mathbb EWW' $ is the variance/covariance matrix of a mean-zero Gaussian noise. We give ourselves a greater ability to infer the state of the signal simply by assuming that the noise is Gaussian. In the simplest case, just knowing the covariance properties of the system is enough to make a projection onto a basis of orthogonal basis, a projection that may even be a posterior expectation if the model can be shown to have a jointly-Gaussian structure. If we further assume that the signal evolves according to an independent Gaussian model we can apply a Kalman filter, which is extremely effective for tracking hidden Markov processes, particularly ones of multiple dimension.

\subsection{The Wiener Filter}
\label{sec:wiener}
Given the data $Y = (Y_0,\dots, Y_N)'$, the signal $X = (X_0,\dots, X_N)'$ combines with a noise $W = (W_0,\dots, W_N)'$ so that
\[Y= X+W\] 
where $\mathbb EW=0$, the covariance matrix of the noise is
\[R\doteq\mathbb EWW',\]
and the covariance matrix of $X$ is
\[Q \doteq  \mathbb E[(X-\mathbb EX)X'].\]
The information introduce by $Y$ can be encapsulated in \textbf{the innovation},
\[V \doteq Y-\mathbb EX,\]
and the optimal linear estimate of $X$ is its projection,
\[\mathcal P_YX  \doteq\mathbb EX + GV\]
where the matrix $G$ is defined apriori in such a way as to make the projection error orthogonal to the posterior information:

\[0=\mathbb E[(X-\mathcal P_YX)Y'] =\mathbb E\left[(X-\mathbb EX-GV)Y'\right] \]

\[=\mathbb E\left[(X-\mathbb EX)X'\right]+\mathbb EXW'-\mathbb E\left [G(Y-\mathbb EX)Y'\right] \]

\[=Q-G( Q+R)\qquad\qquad\qquad\qquad(*)\]
where we have assumed that $\mathbb EXW'=0$ because noise by construction should be independent of the signal. If we solve $(*)$ we get
\begin{equation}
\label{eq:wienerGain}
G = Q(Q+R)^{-1}
\end{equation}
which is the optimal projection matrix. The Wiener filter is essentially a linear projection using the matrix in (\ref{eq:wienerGain}). Notice that we have made minimal assumptions about the distributions of the random variables; all we have assumed is that we know the mean and covariance structure of $X$ and the driving noise in $Y$.

If we assume that $(X,Y)$ are jointly Gaussian, then any random variable with the same distribution as $X$ is equal in distribution to a random variable that is a linear sum of $Y$ and another Gaussian component that is independent of $Y$,
\[X =_d \mathbb EX + F_1(Y-\mathbb EX) + F_2Z \]
where $F_1$ and $F_2$ are non-random matrices of coefficients, and $Z$ is mean-zero Gaussian and independent of $Y$. With this representation we have
\[\mathbb E[X|Y] =_d \mathbb EX+F_1V.\]
Now, by independence of $Y$ and $Z$, we must have
\[0=\mathbb E\left[(X-\mathbb E[X|Y])Y'\right]\]
which leads to the solution $F_1=G$ where $G$ is the projection matrix given by (\ref{eq:wienerGain}). In general, the MSE is bounded below by that of the posterior mean,
\[\mathbb E\|X-\mathbb E[X|Y]\|^2\leq\mathbb E\|X-\mathcal P_YX\|^2.\]
But if $Y$ and $X$ are not jointly Gaussian, it can be shown that the MSE of the projection will be strictly greater than that of the posterior mean. In figure \ref{fig:wiener} the Wiener filter is used to track the random-walk example that was in section \ref{sec:bandwidth}. For the random-walk example, the matrices are

\[Q = \left[\begin{array}{cccccc}
1&1&1&1&\dots&1\\
1&2&2&2&\dots&2\\
1&2&3&3&\dots&3\\
1&2&3&4&\dots&4\\
\vdots&\vdots&\vdots&\vdots&\ddots&\vdots\\
1&2&3&4&\dots&N
\end{array}\right] ,\qquad\qquad R = I_{N\times N}\]
which are ill-conditioned, but the round-off error is not significant for $N=1000$. Indeed, the SNR and MSE of the Wiener filter is 85.79 and  .0479, both are better than the best results among the wavelet and bandwidth filters (the best was the coiflet with vanishing moments which had SNR = 64.90 and MSE =  0.0549). This example illustrates how the Wiener filter is the optimal among all posterior estimators.
\begin{figure}[htbp] 
   \centering
   \includegraphics[width=5in]{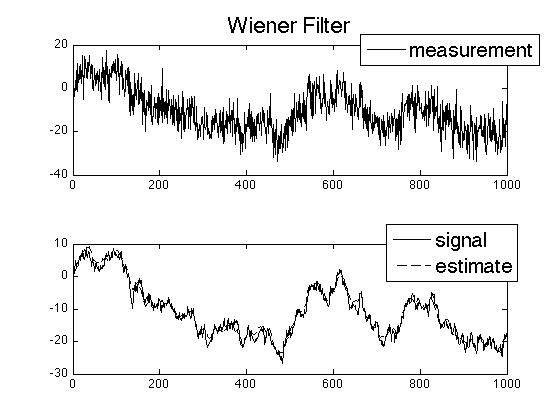} 
   \caption{\small The Wiener filter applied to a random-walk. The SNR = 85.69 and MSE = .0479. The SNR/MSE is higher/lower than it was for the bandwidth filter since the Wiener filter is the optimal posterior estimator.}
   \label{fig:wiener}
\end{figure}

\begin{remark} The resemblance of the Wiener filter to a penalized and weighted least-squares problem is clear from the first order conditions of the following minimization,
\[\min_X\left(\frac{1}{2}X'Q^{-1}X-X'(Q+R)^{-1}Y\right).\]
\end{remark}

\begin{remark} When $X = (X_0,X_2,\dots,X_{N-1})$ is a realization from a a jointly Gaussian HMM, the Wiener filter returns not only the posterior mean, but also the path of $X$ that is the maximum likelihood. In such cases, the estimator $\widehat X_n$ for $n<N-1$ can be consider a \textit{smoothing} rather than a filtering because it is an estimate of the state's past value,
\[\widehat X_n = \mathbb E[X_n|Y_{0:N-1}]\qquad\qquad\hbox{for }n<N-1.\]
\end{remark}

\begin{remark} The numerical linear algebra for computing the Wienfer filter requires no inversion of matrices, but only to solve two linear systems. Observe, $\widehat X$ is the solution to a linear system,
\[(Q+R)\underbrace{Q^{-1}\widehat X}_{=Z}=Y\]
so first we solve a linear system for $Z$
\[(Q+R)Z=Y\]
and then we solve for the filter,
\[\widehat X = QZ.\]

\end{remark}
\subsection{The Kalman Filter}
The Kalman filter can be thought of as a generalization of the Wiener filter, but for a model with a slightly more specific model for the signal. In fact, the Kalman filter is a filter for an HMM whose dynamics are Gaussian and fully linear,

\begin{eqnarray}
\nonumber
X_n &=& AX_{n-1}+B_n\\
\nonumber
Y_n&=&HX_n+W_n
\end{eqnarray}
where $W$ and $B$ are independent Gaussian random variables with covariance matrices
\begin{eqnarray}
\nonumber
Q &=&\mathbb EB_nB_n'\\
R &=&\mathbb EW_nW_n'
\end{eqnarray}
both of which are positive-definite, and the distribution of $(X_0,Y_0)$ is a joint Gaussian. We can re-write this two equations as one linear system,

\begin{equation}
\label{eq:linearSystem}
\left(\begin{array}{c}
X_n\\
Y_n
\end{array}\right)=\left(
\begin{array}{cc}
A&0\\
HA&0
\end{array}\right)\left(\begin{array}{c}
X_{n-1}\\
Y_{n-1}
\end{array}\right)+\left(
\begin{array}{cc}
\sqrt Q&0\\
H&\sqrt R
\end{array}\right)\left(\begin{array}{c}
B_n\\
W_n
\end{array}\right)
\end{equation}
which is clearly a non-degenerate Gaussian system. In fact (\ref{eq:linearSystem}) has a stationary mean if the number $1$ is not included in the spectrum of $A$.

Given the data $Y_{0:n} = (Y_0,\dots, Y_n)$, the Kalman filter will find the optimal projection of $X_n$ onto the Gaussian sub-space spanned by $Y$ by iteratively refining the optimal projection of $X_{n-1}$ onto the space spanned by $Y_{0:n-1}$. Furthermore, the optimal projection will be equivalent to the posterior mean because $(X,Y)$ are jointly Gaussian.

In this case we can identify a sequence of Gaussian random variables that are \textbf{the innovations}
\[V_n \doteq Y_n-HA\widehat X_{n-1},\]
but the idea is essentially the same as it was in the Wiener filter. Initially, letting $\widehat X_0 = \mathbb E[X_0|Y_0]$ we use a Wiener filter to get
\[\widehat X_0=\mathbb EX_0 + G_0(Y_0-H\mathbb EX_0)\]
where $G_0 = var(X_0)\left(R+var(X_0)\right)^{-1}$. Clearly, $X_0-\widehat X_0\perp Y_0$, and we can easily check that $(X_0,\widehat X_0,Y_0)$ is jointly Gaussian, and so it follows that $X_0-\widehat X_0$ is independent of $Y_0$, and so posterior covariance is not a random variable,
\[\Sigma_0\doteq \mathbb E\left[(X_0-\widehat X_0)X_0'\Big|Y_0\right]=\mathbb E(X_0-\widehat X_0)X_0'.\]
Now we proceed inductively to identify the filter of $X_n$ given $Y_{0:n}$. Suppose we have obtain the filter up to time $n-1$ with posterior mean
\[\widehat X_{n-1} \doteq \mathbb E[X_{n-1}|Y_{0:n-1}],\]
for which $X_{n-1}-\widehat X_{n-1}$ is independent of $Y_{0:n-1}$, and with covariance matrix 
\[\Sigma_{n-1}\doteq \mathbb E\left[(X_{n-1}-\widehat X_{n-1})X_{n-1}'\Big|Y_{0:n-1}\right]=\mathbb E\left[(X_{n-1}-\widehat X_{n-1})X_{n-1}'\right].\]

When the observation $Y_n$ arrives, the optimal projection will be
\[\mathcal P_nX_n \doteq A\widehat X_{n-1}+G_nV_n\]
where $G_n$ is a projection matrix that is known at time $n-1$; it's known before $Y_n$ has been observed. We can verify that
\begin{itemize}
\item the distribution of $(X_n,Y_n)$ conditioned on $Y_{0:n-1}$ is jointly Gaussian, and 
\item that $V_n$ is independent of $Y_{0:n-1}$, 
\end{itemize}
and therefore it follows that $\widehat X_n= \mathcal P_nX_n$. We can also write the following expression for the prediction covariance matrix,

\[\Sigma_{n|n-1}\doteq\mathbb E\left[(X_n - A\widehat X_{n-1})X_n'\Big|Y_{0:n-1}\right]\]

\[=\mathbb E\left[(AX_{n-1}+B_n - A\widehat X_{n-1})(AX_{n-1}+B_n)'\right]\]

\[=A\Sigma_{n-1}A'+Q\]
and from the orthogonality of the projection residual to the data, we should have a projection matrix that satisfies the following equation,
\[0=\mathbb E\left[(X_n-\widehat X_n)Y_n'\Big|Y_{0:n-1}\right]\]

\[=\mathbb E\left[(X_n - A\widehat X_{n-1})X_n'H'\Big|Y_{0:n-1}\right]-\mathbb E\left[G_n\left(Y_n-HA\widehat X_{n-1}\right)Y_n'\Big|Y_{0:n-1}\right]\]

\[=\Sigma_{n|n-1}H'-G_n\left(H\Sigma_{n|n-1}H'+R\right).\]
We solve this equation to obtain the optimal projection matrix, also known as \textbf{the Kalman filter Gain matrix}
\begin{equation}
\label{eq:kfGain}
G_n=\Sigma_{n|n-1}H'\left(H\Sigma_{n|n-1}H'+R\right)^{-1}
\end{equation}
and using the gain matrix we can write the posterior mean as a recursive function of the innovation and the previous time's posterior mean
\begin{equation}
\label{eq:kfMean}
\widehat X_n=A\widehat X_{n-1}+G_nV_n.
\end{equation}
Furthermore, we can verify that the conditional distribution of $(X_n,\widehat X_n,Y_n)$ is jointly Gaussian, and since $\mathbb E[(X_n-\widehat X_n)Y_m']=0$ for all $m\leq n$, it follows that $X_n-\widehat X_n$ is independent of $Y_{0:n}$. Therefore, the covariance matrix is not a function of the data
\[\Sigma_n=\mathbb E\left[(X_n - \widehat X_n)X_n'\right]=\mathbb E\left[(X_n - \widehat X_n)X_n'\Big|Y_{0:n-1}\right]\]

\[=\mathbb E\left[(X_n - A\widehat X_{n-1})X_n'\Big|Y_{0:n-1}\right]-G_n\mathbb E\left[(Y_n-HA\widehat X_{n-1})X_n'\Big|Y_{0:n-1}\right]\]

\[=A\Sigma_{n-1}A'+Q -G_nH\left(A\Sigma_{n-1}A'+Q\right)=(I-G_nH)\Sigma_{n|n-1}. \]

To summarize, we have shown that $X_n-\widehat X_n\sim N(0,\Sigma_n)$ and independent of $Y_{0:n}$, and from equations (\ref{eq:kfGain}), (\ref{eq:kfMean}) along with the equations for $\Sigma_{n|n-1}$ and $\Sigma_n$ we have \textbf{the Kalman filter} at time $n$
\begin{eqnarray}
\nonumber
\Sigma_{n|n-1}&=&A\Sigma_{n-1}A'+Q\\
\nonumber
G_n&=&\Sigma_{n|n-1}H'\left(H\Sigma_{n|n-1}H'+R\right)^{-1}\\
\nonumber
\widehat X_n&=&A\widehat X_{n-1}+G_nV_n\\
\nonumber
\Sigma_n&=&(I-G_nH)\Sigma_{n|n-1}
\end{eqnarray}
so that the posterior density of $X_n$ is

\[p(X_n\in dx|Y_{0:n}) = \frac{1}{(2\pi|\Sigma_n|)^{d/2}}\exp\left\{ -\frac{1}{2}(x-\widehat X_n)'\Sigma_n^{-1}(x-\widehat X_n)      \right\}\]
where $d$ is the dimension such that $X_n\in\mathbb R^d$. 

In figure \ref{fig:kalman} we see the Kalman filter's ability to track the same random-walk example on which we test the filters from sections \ref{sec:bandwidth}, \ref{sec:wavelets} and \ref{sec:wiener}. 
\begin{figure}[htbp] 
   \centering
   \includegraphics[width=5in]{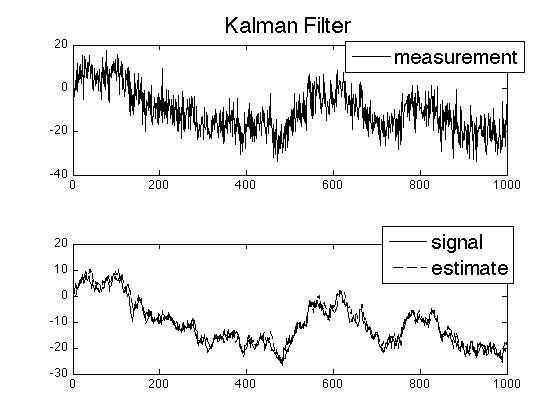} 
   \caption{\small The Kalman filter applied to the random-walk example. The path of Kalman filter estimates is not optimal as a whole, but each $\widehat X_n$ returned by the Kalman filter is optimal given the information $Y_{0:n}$.}
   \label{fig:kalman}
\end{figure}

The random-walk model is
\begin{eqnarray}
\nonumber
X_n&=&X_{n-1}+B_n\\
\nonumber
Y_n &= &X_n+\gamma W_n.
\end{eqnarray}
and Kalman filter for the random-walk is
\begin{eqnarray}
\nonumber
G_n&=& \left(\Sigma_{n-1}+1\right)\left(\Sigma_{n-1}+1+\gamma^2\right)^{-1}\\
\nonumber
\widehat X_n&=&\widehat X_{n-1}+G_n(Y_n-\widehat X_{n-1})\\
\nonumber
\Sigma_n&=&(1-G_n)\left(\Sigma_{n-1}+1\right).
\end{eqnarray}
Given $Y_{0:N}$, the Kalman filter returns an optimal estimator of $X_N$, not the entire path taken by $X$. Indeed, the Kalman filter's path has SNR =  43.88 and MSE = 0.0668, neither of which are better than the other filters. But this example should not be evidence for a dismissal of the Kalman filter, it simply shows that ex-ante estimation of the entire path of the signal is not its specialty.

The Kalman filter is far superior to the other filters we've discussed when it is applied to problems where $X_n$ is a multidimensional vector. When each observation is a vector, the curse of dimensionality makes it impossible to work with the basis' required for bandwidth and wavelets, and the size of the matrices needed for the Wiener filter also be prohibitively large. On the other hand, the Kalman filter works efficiently and in real-time.

\begin{remark} For HMMs, the Kalman filter and the Wiener filter coincide in their estimates of the latest value of the signal,
\[\widehat X_{N-1}^{wiener}=\mathbb E[X_{N-1}|Y_{0:N-1}] = \widehat X_{N-1}^{kalman}.\]
\end{remark}

\begin{remark} The Kalman filter is indeed capable of handling signals of high dimension, but there does not exist a general procedure for avoiding the explicit computation of the matrix inverse when computing the gain matrix. Sometimes this inverse may be manageable, but limitations in our ability to compute matrix inverse represent the upper-bound on the Kalman filter's capacity. 

\end{remark}

%

\chapter{The Baum-Welch \& Viterbi Algorithms}

\noindent Filtering equations for the class of fully-discrete HMMs are relatively simple to derive through Bayesian manipulation of the posteriors. These discrete algorithms are interesting because they embody the most powerful elements of HMM theory in a very simple framework. The methods are readily-implementable and have become the workhorse in applied areas where machine learning algorithms are needed. The algorithms for filtering, smoothing and parameter estimation are analogous to their counterparts in continuous models, but the theoretical background required for understanding is minimal in the discrete setting. 

\section{Equations for Filtering, Smoothing \& Prediction}
Let $n$ denote a discrete time, and suppose that $X_n$ is an unobserved Markov chain taking values in a discrete state-space denoted by $\mathcal S$. Let $\Lambda$ denote $X_n$'s kernel of transition probabilities so that
\[\mathbb P(X_{n+1}=x) = \sum_{v\in\mathcal S}\Lambda(x|v)\mathbb P(X_n=v)\]
for any $x\in\mathcal S$, and $\mathbb P(X_0=x) = p_0(x)$. 

Noisy measurements are taken in the form of a process $Y_n$ which is a nonlinear function of $X_n$, plus some noise,
\[Y_n = h(X_n)+W_n\]
where $W_n$ is an iid Gaussian random variable with mean zero and variance $\gamma^2>0$. The main feature of this discrete model is the \textbf{memoryless-channel} which allows the process to `forget the past':
\[\mathbb P(Y_n,X_n=x|X_{n-1}=v,Y_{0:n-1})=\mathbb P(Y_n|X_n=x)\Lambda(x|v)\]
for any $n\geq 0$ and for all $x,v\in\mathcal S$.

\subsection{Filtering}
The filtering mass function is 
\[\pi_n(x) \doteq \mathbb P(X_n=x|Y_{0:n})\]
for all $x\in\mathcal S$. Through an application of Bayes rule along with the properties of the HMM, we are able to break down $\pi_n$ as follows,
\[\pi_n(x) = \frac{\mathbb P(X_n=x,Y_{0:n})}{\mathbb P(Y_{0:n})}\]

\[= \frac{\sum_{v\in\mathcal S}\mathbb P(Y_n,X_n=x|X_{n-1}=v,Y_{0:n-1})\mathbb P(X_{n-1}=v,Y_{0:n-1})}{\mathbb P(Y_{0:n})}\]

\[= \frac{\mathbb P(Y_n|X_n=x)\sum_{v\in\mathcal S}\mathbb P(X_n=x|X_{n-1}=v)\mathbb P(X_{n-1}=v,Y_{0:n-1})}{\mathbb P(Y_{0:n})}\]

\[= \frac{\mathbb P(Y_n|X_n=x)\sum_{v\in\mathcal S}\Lambda(x|v)\mathbb P(X_{n-1}=v|Y_{0:n-1})}{\mathbb P(Y_n|Y_{0:n-1})}\]

\[= \frac{\mathbb P(Y_n|X_n=x)\sum_{v\in\mathcal S}\Lambda(x|v)\pi_{n-1}(v)}{\sum_{x\in\mathcal S}\hbox{numerator}}\]
where the memoryless-channel allows for the conditioning that occurs between the second and third lines. This recursive breakdown of the filtering mass is the \textbf{forward Baum-Welch Equation}, and can be written explicitly for the the system with Gaussian observation noise
\begin{equation}
\label{eq:fbw}
\pi_n(x) = \frac{1}{c_n}\psi_n(x)\sum_{v\in\mathcal S}\Lambda(x|v)\pi_{n-1}(v)
\end{equation}
where $c_n$ is a normalizing constant, and $\psi_n$ is a likelihood function
\[\psi_n(x)\doteq \mathbb P(Y_n|X_n=x) = \exp\left\{-\frac{1}{2}\left(\frac{Y_n-h(x)}{\gamma}\right)^2\right\}.\]
Equation (\ref{eq:fbw}) is convenient because it keeps the distribution updated without having to recompute old statistics as new data arrives. In `real-time' it is efficient to use this algorithm to keep track of $X$'s latest movements, but older filtering estimates will not be optimal after new data has arrived. The smoothing distribution must be used to find the optimal estimate of $X$ at some time in the past.
\subsection{Smoothing}
For some time $N>n$ up to which data has been collected, the smoothing mass function is
\[\pi_{n|N}(x) \doteq \mathbb P(X_n=x|Y_{0:N}).\]
Through an application of Bayes rule along with the properties of the model, the smoothing mass can be written as follows,

\[\pi_{n|N}(x) = \frac{\mathbb P(Y_{n+1:N}|X_n=x)\pi_n(x)}{\mathbb P(Y_{n+1:N}|Y_{0:n})}\]

\[= \frac{\sum_{v\in\mathcal S}\mathbb P(Y_{n+1:N}|X_{n+1}=v,X_n=x)\Lambda(v|x)\pi_n(x)}{\mathbb P(Y_{n+1:N}|Y_{0:n})}\]

\[= \frac{\sum_{v\in\mathcal S}\mathbb P(Y_{n+2:N}|X_{n+1}=v)\psi_{n+1}(v)\Lambda(v|x)\pi_n(x)}{\mathbb P(Y_{n+2:N}|Y_{0:n+1})\mathbb P(Y_{n+1}|Y_{0:n})}\]

\[= \frac{\sum_{v\in\mathcal S}\mathbb P(Y_{n+2:N}|X_{n+1}=v)\psi_{n+1}(v)\Lambda(v|x)\pi_n(x)}{\mathbb P(Y_{n+2:N}|Y_{0:n+1})c_{n+1}}.\qquad(*)\]
where $c_{n+1}$ is the normalizing constant from equation (\ref{eq:fbw}). Now suppose that we define a likelihood function for the events after time $n$, 
\[\alpha_n^N(x) = \frac{\mathbb P(Y_{n+1:N}|X_n=x)}{\mathbb P(Y_{n+1:N}|Y_{0:n})}\]
for $n<N$ with the convention that $\alpha_N^N\equiv 1$. Then the smoothing mass can be written as the product of the filtering mass with $\alpha$
\[\pi_{n|N}(x) = \alpha_n^N(x)\pi_n(x)\]
and from $(*)$ we can see that $\alpha_n^N$ is given recursively by a \textbf{backward Baum-Welch Equation}

\begin{equation}
\label{eq:bbw}
\alpha_n^N(x) = \frac{1}{c_{n+1}}\sum_{v\in\mathcal S}\alpha_{n+1}^N(v)\psi_{n+1}(v)\Lambda(v|x).
\end{equation}
Clearly, computation of the smoothing distribution requires a computation of all filtering distribution up to time $N$ followed by the backward recursion to compute $\alpha^N$. In exchange for doing this extra work, the sequence of $X$'s estimates will suggest a path taken by $X$ that is more plausible than the path suggested by the filtering estimates.
\subsection{Prediction}
The prediction distribution is easier to compute than smoothing. For $n<N$, the prediction distribution is
\[\pi_{N|n}(x) \doteq \mathbb P(X_N=x|Y_{0:n})\] 
and is merely computed by extrapolating the filtering distribution,

\[\pi_{N|n}(x)= \sum_{v\in\mathcal S} \Lambda(x|v)\pi_{N-1|n}(v) = \sum_{v\in\mathcal S} \Lambda^{N-n}(x|v)\pi_n(v)\]
where $\Lambda^{N-n}$ denotes the transition probability over $N-n$ time steps. 

If $X_n$ is a positive recurrent Markov chain, then there is an invariant and the prediction distribution will converge to as $N\rightarrow\infty$. In some cases, the rate at which this convergence occurs will be proportional to the spectral gap in $\Lambda$.

Suppose $X_n$ can take one of $m$-many finite-state, and is a recurrent Markov chain with only 1 communication class. Let $\Lambda\in\mathbb R^{m\times m}$ be the matrix of transition probabilities for $X$, and suppose that $\Lambda_{ji}>0$ so that
\[\mathbb P(X_{n+1}=x_i|X_n=x_j) = \Lambda_{ji}>0\]
for all $i,j\leq m$. Then the prediction distribution is
\[\pi_{N|n} = \pi_n\Lambda^{N-n}\]
and will converge exponentially fast to the invariant measure with a rate proportional to the second eigenvalue of $\Lambda$. To see why this is true, consider the basis of eigenvectors $(\mu_i)_{i\leq m}$ of $\Lambda$, some of which may be generalized,
\[\mu_{i+1}(\Lambda-\beta_iI)=\mu_i\]
for some $i\geq1$. Assuming that $\mu_1$ is the unique invariant mass function of $X_n$, we have $\mu_1\Lambda= \mu_1$. By the Perron-Frobenius Theorem we can sort the eigenvalues so that $1=\beta_1>|\beta_2|\geq|\beta_3|\geq \dots\geq|\beta_{m}|$, and we know that $\beta_1$ is a simple root of the characteristic polynomial and therefore $\mu_1$ is not a generalized eigenvector. From here we can see that

\[-\frac{1}{k}\log\|\pi_n\Lambda^k -\mu_1\|=-\frac{1}{k}\log\|(\pi_n-\mu_1) \Lambda^k \|=- \frac{1}{k}\log\|(a_1\mu_1+a_2\mu_2+\dots a_{m}\mu_{m})\Lambda^k\|\]

\[=- \frac{1}{k}\log\|a_1\mu_1+a_2\beta_2^k\mu_2+\dots a_{m}\mu_{m}\Lambda^k\|\sim \frac{1}{k}\log\left(1+a_2'|\beta_2^k|\right)\sim|\beta_2|\]
as $k\rightarrow\infty$. The spectral gap of $\Lambda$ is $1-|\beta_2|$, and from the convergence rate we see that a greater spectral gap means that the prediction distribution will take less time to converge to the invariant measure. In general, the Perron-Frobenius theorem can be applied to a recurrent finite-state Markov chain provided that there is some integer $k<\infty$ for which $\Lambda_{ji}^k>0$ for all $i,j\leq m$.

\section{Baum-Welch Algorithm for Learning Parameters}
It is not very realistic to assume that we have apriori knowledge of the HMM that is completely accurate. However, stationarity of $X$ means that we are observed repeated behavior of $X$, albeit through noisy measurements, but nevertheless we should be able to judge the frequencies with which $X$ occupies parts of the state-space and the frequencies with which it moves about.

If we have already computed the smoothing distribution based on a model that is `close' in some sense, then we should have
\begin{eqnarray}
\label{eq:muUpdate}
\frac{1}{N}\sum_{n=1}^N\mathbb P(X_n=x|Y_{0:N})&\approx &\mu(x)\\
\label{eq:lambdaUpdate}
\frac{1}{N}\sum_{n=1}^N\mathbb P(X_n=x,X_{n-1}=v|Y_{0:N})&\approx &\Lambda(x|v)\mu(v)
\end{eqnarray}
where $\mu$ is the stationary law of $X$. With the Baum-Welch algorithm, we can in fact employ some optimization techniques to find a sequence of model estimates which are of increasing likelihood, and it turns out that the (\ref{eq:muUpdate}) and (\ref{eq:lambdaUpdate}) are similar to the optimal improvement in selecting the sequence of models.

Consider two model parameters $\theta$ and $\theta'$. The Baum-Welch algorithm uses the Kullback-Leibler divergence to compare the two models,

\[0\leq D(\theta\|\theta') = \sum_{\vec x\in\mathcal S^{N+1}}\frac{\mathbb P^\theta(X_{0:N}=\vec x,Y_{0:N})}{\mathbb P^\theta(Y_{0:N})}\log\left(  \frac{\mathbb P^\theta(X_{0:N}=\vec x,Y_{0:N})\mathbb P^{\theta'}(Y_{0:N})}{\mathbb P^{\theta'}(X_{0:N}=\vec x,Y_{0:N})\mathbb P^\theta(Y_{0:N})}\right)\]

\[=\log\left(  \frac{\mathbb P^{\theta'}(Y_{0:N})}{\mathbb P^\theta(Y_{0:N})}\right)+\sum_{\vec x\in\mathcal S^{N+1}}\frac{\mathbb P^\theta(X_{0:N}=\vec x,Y_{0:N})}{\mathbb P^\theta(Y_{0:N})}\log\left(  \frac{\mathbb P^\theta(X_{0:N}=\vec x,Y_{0:N})}{\mathbb P^{\theta'}(X_{0:N}=\vec x,Y_{0:N})}\right).\]
If we set 
\[Q(\theta\|\theta') \doteq \sum_{\vec x\in\mathcal S^{N+1}}\mathbb P^\theta(X_{0:N}=\vec x,Y_{0:N})\log\left(  \mathbb P^{\theta'}(X_{0:N}=\vec x,Y_{0:N})\right),\]
we then have a simplified expression,
\[0\leq D(\theta\|\theta') = \log\left(  \frac{\mathbb P^{\theta'}(Y_{0:N})}{\mathbb P^\theta(Y_{0:N})}\right)+\frac{Q(\theta\|\theta)-Q(\theta\|\theta')}{\mathbb P^{\theta}(Y_{0:N})}\]
and rearranging the inequality we have
\[\frac{Q(\theta\|\theta')-Q(\theta\|\theta)}{\mathbb P^\theta(Y_{0:N})}\leq\log\left(  \frac{\mathbb P^{\theta'}(Y_{0:N})}{\mathbb P^\theta(Y_{0:N})}\right),\]
from which we see that $Q(\theta\|\theta')>Q(\theta\|\theta)$ implies that $\theta'$ has greater likelihood than $\theta$. The Baum-Welch algorithm uses this inequality as the basis for a criteria to iteratively refine the estimated model parameter. The algorithm obtains a sequence $\{\theta^\ell\}_\ell$ for which $Q(\theta^{\ell-1}\|\theta^\ell)\geq 0$, and so their likelihoods are increasing but bounded,
\[\mathbb P^{\theta^{\ell-1}}(Y_{0:N})\leq\mathbb P^{\theta^\ell}(Y_{0:N})\leq \mathbb P^{\hat\theta^{mle}}(Y_{0:N}),\]
where $\hat\theta^{mle}$ is the maximum likelihood estimate of $\theta$. Therefore, $\{\theta^\ell\}_\ell$ will have a limit at $\theta^*$ such that
\[\mathbb P^{\theta^*}(Y_{0:N}) = \lim_\ell\mathbb P^{\theta^\ell}(Y_{0:N}),\]
but it may be the case that $\mathbb P^{\theta^*}(Y_{0:N})<\mathbb P^{\hat\theta^{mle}}(Y_{0:N})$ (see figure \ref{fig:BWconvergence}).

\begin{figure}[htbp] 
   \centering
   \includegraphics[width=5.5in]{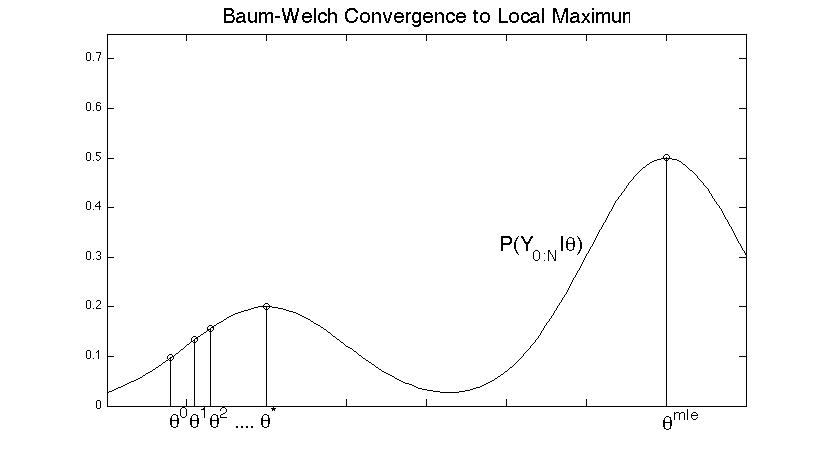} 
   \caption{\small A sequence of Baum-Welch parameter estimates with increasing likelihood, but the sequence is caught at a local maximum.}
   \label{fig:BWconvergence}
\end{figure}

In doing computations, a maximum (perhaps only a local maximum) of $Q(\theta\|~\cdot~)$ needs to be found. First-order conditions are good technique for finding one, and using the HMM we can expand $Q(\theta\|\theta')$ into an explicit form,
\begin{equation}
\label{eq:Qexpand}
Q(\theta\|\theta') = \sum_{\vec x\in\mathcal S^{N+1}}\mathbb P^\theta(X_{0:N}=\vec x,Y_{0:N})\left\{\psi_0^{\theta'}(\vec x_0)p_0^{\theta'}(\vec x_0)+\sum_{n=1}^N\log\left(  \psi_n^{\theta'}(\vec x_n)\Lambda^{\theta'}(\vec x_n|\vec x_{n-1})\right)\right\},
\end{equation}
from which we see that it is possible to differentiate with respect to $\theta'$, add the Lagrangians, and then solve for the optimal model estimate. 

The Baum-Welch algorithm is equivalent to the expectation-maximization (EM) algorithm; the EM algorithm maximizes the expectation of the log-likelihood function which is equivalent to maximizing $Q$,

\[\theta^\ell=\arg\max_\theta \mathbb E^{\theta^{\ell-1}}\left[\log\left(\mathbb P^\theta(Y_{0:N},X_{0:N})\right)\Big|Y_{0:N}\right]=\arg\max_\theta Q(\theta^{\ell-1}\|\theta).\]
\subsection{Model Re-Estimation for Parametric Transition Probabilities}
Suppose that $X_n\in\mathbb Z$, with transition probabilities parameterized by $\theta\in(0,\infty)$ so that
\[\mathbb P(X_{n+1}=i|X_n=j) =\frac{1}{c(\theta)} \exp\{-\theta|i-j|^2\},\qquad\forall i,j\in\mathbb Z,\]
where $c(\theta)=\sum_{i=-\infty}^\infty\exp\{-\theta|i-j|^2\}$. Ignoring the parts that do not depend on $\theta'$, the log-likelihood is
\[Q(\theta\|\theta') =-\sum_{n=1}^N\mathbb E^\theta\left[ \theta'|X_n-X_{n-1}|^2+\log c(\theta')\Big|\mathcal F_N^Y\right],\]
and if we differentiate with respect to $\theta'$ we have the following first-order conditions,
\[\frac{\partial}{\partial \theta'} Q(\theta\|\theta') = -\sum_{n=1}^N\mathbb E^\theta\left[ |X_n-X_{n-1}|^2-\frac{\sum_i|i-j|^2\exp\{-\theta|i-j|^2\}}{ c(\theta')}\Bigg|\mathcal F_N^Y\right]=0\]
for any $j\in\mathbb Z$. The solution to the first-order conditions is $\theta'$ that satisfies

\[\mathbb E^{\theta'}\left[|X_1-X_0|^2\Big|X_0=j\right]=\frac{1}{N}\sum_{n=1}^N\mathbb E^\theta\left[ |X_n-X_{n-1}|^2\Big|\mathcal F_N^Y\right]\]
for any $j$.
\subsection{Model Re-Estimation for Finite-State Markov Chains}
Suppose $X_n\in\mathcal S=\{1,\dots,m\}$, so that
\[\mathbb P(X_{n+1}=i|X_n=j) = \Lambda_{ji}\]
for all $i,j\in\mathcal S$. We will look for a sequence $\Lambda^{(\ell)}$ which maximizes $Q(\Lambda^{(\ell-1)}\|~\cdot~)$ subject to the constraints $\sum_i\Lambda_{ji}=1$ for all $j\leq m$. Letting $\delta_j$ be the Lagrange multiplier for the $j$th constraint, the first order conditions are then,

\[\frac{\partial}{\partial \Lambda_{ji}}\left(Q(\Lambda^{(\ell-1)}\|\Lambda)-\delta_j\sum_r\Lambda_{jr}\right)=\frac{\partial}{\partial \Lambda_{ji}}Q(\Lambda^{(\ell-1)}\|\Lambda)-\delta_j=0.\qquad\qquad(**)\]
Multiplying by $\Lambda_{ji}$ and summing over $i$ the expression in $(**)$ becomes

\[0=\sum_i\Lambda_{ji}\left( \frac{\partial}{\partial \Lambda_{ji}}Q(\Lambda^{(\ell-1)}\|\Lambda)-\delta_j\right)=\sum_i\Lambda_{ji}\frac{\partial}{\partial \Lambda_{ji}}Q(\Lambda^{(\ell-1)}\|\Lambda)-\delta_j\]
which means $\delta_j = \sum_i\Lambda_{ji}\frac{\partial}{\partial \Lambda_{ji}}Q(\Lambda^{(\ell-1)}\|\Lambda)$. By multiplying $(**)$ by $\Lambda_{ji}$ and then rearranging terms it is found that the optimal $\Lambda_{ji}^{(\ell)}$ must be chosen among the set of $\Lambda$'s such that

\begin{equation}
\label{eq:foc}
\Lambda_{ji} =\frac{\Lambda_{ji} \frac{\partial}{\partial \Lambda_{ji}}Q(\Lambda^{(\ell-1)}\|\Lambda)}{\sum_r\Lambda_{jr}\frac{\partial}{\partial \Lambda_{jr}}Q(\Lambda^{(\ell-1)}\|\Lambda)}.
\end{equation}
Now, using the expansion in (\ref{eq:Qexpand}), the derivative of $Q(\Lambda^{(\ell-1)}\|\Lambda)$ with respect to $\Lambda_{ji}$ can be computed as follows:
 \[\frac{\partial}{\partial\Lambda_{ji}}Q(\Lambda^{(\ell-1)}\|\Lambda)=\mathbb E\left[\frac{\partial}{\partial\Lambda_{ji}}\log \mathbb P(Y_{0:N},X_{0:N}|\Lambda)\Big|Y_{0:N},\Lambda^{(\ell-1)}\right]\]
 \[=\mathbb E\left[\sum_{n=1}^N\frac{1}{\Lambda_{ji}}\mathbf 1_{\{X_n=i,X_{n-1}=j\}}\Bigg|Y_{0:N},\Lambda^{(\ell-1)}\right]=\frac{1}{\Lambda_{ji}}\sum_{k=1}^N\mathbb P(X_n=i,X_{n-1}=j|Y_{0:N},\Lambda^{(\ell-1)})\]
and by plugging this into equation (\ref{eq:foc}) it is easily seen that the solution is
\begin{equation}\label{eq:emSolution}
\Lambda_{ji}^{(\ell)} = \frac{\sum_{n=1}^N\mathbb P(X_n=i,X_{n-1}=j|Y_{0:N},\Lambda^{(\ell-1)})}{\sum_i\hbox{numerator}}\end{equation}
where $\mathbb P(X_n=i,x_{n-1}=j|Y_{0:N},\Lambda^{(\ell-1)})=\alpha_n^N(i)\psi_n(i)\Lambda_{ji}^{(\ell-1)}\pi_{n-1}(j)$. It also happens that equation (\ref{eq:emSolution}) enforces non-negativity of $\Lambda_{ji}$, which is required for well-posedness of the algorithm. Equation (\ref{eq:emSolution}) is equivalent to the estimates that were conjectured in (\ref{eq:muUpdate}) and (\ref{eq:lambdaUpdate}).

\section{The Viterbi Algorithm}
Sometimes it may be more important to estimate the entire path of $X$. The Viterbi algorithm applies the properties of the HMM along with dynamic programming to find an optimal sequence $\widehat V_{0:N}\in\mathcal S^{N+1}$ that maximizes the joint-posterior probability
\[\widehat V_{0:N} =(\widehat V_0,\dots,\widehat V_N)\doteq \arg\max_{\vec x\in\mathcal S^{N+1}}\mathbb P(X_{0:N} = \vec x, Y_{0:N}).\]
Given the data $Y_{0:N}$, smoothing can be used to `look-back' and make estimates of $X_n$ for some $n<N$, but neither equations (\ref{eq:fbw}) or (\ref{eq:bbw}) is a joint posterior, meaning that they will not be able to tells us the posterior probability of a path $\vec x\in\mathcal S^{N+1}$. The size of our problem would grow exponentially with $N$ if we needed to compute the posterior distribution of $X's$ paths, but the Viterbi algorithm allows us to obtain the MAP estimator of $X$'s path with without actually calculating the posterior probabilities of all paths.

The memoryless channel of the HMM allows us to write the maximization over paths as a nested maximization,
\[\max_{\vec x\in\mathcal S^{N+1}}\mathbb P(X_{0:N} = \vec x,Y_{0:N}) =\max_{v\in\mathcal S}\psi_N(v)\max_{\vec x\in\mathcal S^N}\Lambda(v|\vec x_{N-1})\mathbb P(X_{0:N-1} = \vec x, Y_{0:N-1})\]

\[=\psi_N(\widehat V_N)\max_{\vec x\in\mathcal S^N}\Lambda(\widehat V_N|\vec x_{N-1})\mathbb P(X_{0:N-1} = \vec x,Y_{0:N-1}),\qquad(\dagger)\]
where $\psi$ is the likelihood and $c_N$ is the normalizing constant, both from the forward Baum-Welch equation in (\ref{eq:fbw}). To take advantage of this nested structure, it helps to define the following recursive function,
\begin{eqnarray}
\nonumber
\phi_0(v)&\doteq&\psi_0(v)\mathbb P(X_0=v)\\
\nonumber
\phi_n(v)& \doteq& \psi_n(v)\max_x\Lambda(v|x)\phi_{n-1}(x),\qquad\hbox{for }n=1,2,3,\dots,N.
\end{eqnarray}
We then place $\phi$ is the nested structure of $(\dagger)$ and work backwards to obtain the optimal path,
\begin{eqnarray}
\nonumber
\widehat V_N&=&\arg\max_v\phi_N(v)\\
\nonumber
\widehat V_n&=&\arg\max_v\Lambda(\widehat V_{n+1}|v)\phi_n(v),\qquad\hbox{for }n=N-1,N-2,\dots,2,1,0
\end{eqnarray}
thus obtaining the optimal path in $O(N)$-many computations. It would have taken $O\left(|\mathcal S|^N\right)$-many computations to obtain the posterior distribution of the paths.

We are interested in the Viterbi algorithm mainly because the path of estimates returned by the filtering and smoothing may 

\begin{remark} The unnormalized probabilities in $\phi$ quickly fall below machine precision levels, so it is better to consider a logarithmic version of Viterbi,
\begin{eqnarray}
\nonumber
\log\phi_0(v)&=&\log\psi_0(v)+\log\mathbb P(X_0=v)\\
\nonumber
\log\phi_n(v)&=&\log\psi_n(v)+\max_x\left\{\log\Lambda(v|x)+\log\phi_{n-1}(x)\right\}
\end{eqnarray}
and the use the $\log\phi_n$'s in the dynamic programming step,
\begin{eqnarray}
\nonumber
\widehat V_N&=&\arg\max_v\log\phi_N(v)\\
\nonumber
\widehat V_n&=&\arg\max_v\left\{\log\Lambda(\widehat V_{n+1}|v)+\log\phi_n(v)\right\},\qquad\hbox{for }n=N-1,N-2,\dots,2,1,0.
\end{eqnarray}

\end{remark}

\chapter{The Particle Filter}

\noindent Monte Carlo methods have become the most common way to compute quantities from HMMs --and with good reason; they are in fact a fast and effective way to obtain consistent estimates. In particular, the particle filter is used to approximate filtering expectations. There are similar methods that exploit Bayes formula in obtaining samples from an HMM, but `particle filtering' implies that sequential Monte Carlo (SIS) and Sampling-Importance-Resampling (SIR) are applied to the specified HMM.

\section{The Particle Filter}
Suppose that $X_n$ is an unobserved Markov chain taking values in a state-space denoted by $\mathcal S$. Let $\Lambda$ denote $X_n$'s kernel of transition densities so that
\[\frac{d}{dx}\mathbb P(X_{n+1}\leq x) = \int\Lambda(x|v)\mathbb P(X_n\in dv)\]
for any $x\in\mathcal S$, and $\frac{d}{dx}\mathbb P(X_0\leq x) = p_0(x)$. Let the observed process $Y_n$ be a nonlinear function of $X_n$,
\[Y_n = h(X_n)+W_n\]
where $W_n$ is an iid Gaussian random variable with mean zero and variance $\gamma^2>0$. In this case, the filter is easily shown to be a density function, given recursively as,
\[\pi_{n+1}(x) = \frac{1}{c_{n+1}}\psi_{n+1}(x)\int\Lambda(x|v)\pi_n(v)\]
where $c_{n+1}$ is a normalizing constant, and $\psi_{n+1}$ is a likelihood function
\[\psi_{n+1}(x) = \exp\left\{-\frac{1}{2}\left(\frac{Y_{n+1}-h(x)}{\gamma}\right)^2\right\},\]
but some kind quadrature grid would need to be established over $\mathcal S$ if we were to use this recursive expression. An alternative is to use particles.

\subsection{Sequential Importance Sampling (SIS)}
Ideally, we would be able to sample directly from the filtering distribution to obtain a Monte Carlo estimate,
\[\frac{1}{P}\sum_{\ell=1}^Pg(x_n^\ell)\approx \mathbb E[g(X_n)|\mathcal F_n^Y],\qquad\hbox{for $P$ large,}\]
where $x_n^\ell\sim iid~ \pi_n(x)$. However, difficulties in computing $\pi_n$ also make it difficult to obtain samples. However, with relative ease we can sequentially obtain samples from $X's$ unconditional distribution and then assign them weights in such a way that approximates the filter.

For $\ell=1,2,3,4,\dots $ , each particle is a path $x_{0:N}^\ell$ that is generated according to the unconditional distribution,
\begin{eqnarray}
\nonumber
x_0^\ell&\sim& p_0(~\cdot~)\\
\nonumber
x_n^\ell&\sim&\Lambda(~\cdot~|x_{n-1}^\ell)\qquad\hbox{for }n=1,2,3,\dots,N.
\end{eqnarray}
Then for $P$-many particles and any integrable function $g$, the strong law of large numbers tells us that
\[\frac{1}{P}\sum_{\ell=1}^Pg(x_{0:N}^\ell)\rightarrow \mathbb E[g(X_{0:N})]\]
almost surely as $P\rightarrow\infty$.

Given $Y_{0:n}$, let $\omega_n^\ell$ denote the \textit{importance weight} of a particle. We define $\omega_n^\ell$ to proportional to the likelihood of the $\ell$th particle's path, which we can write recursively as the product of its old weight and a likelihood function:
\[\omega_n^\ell=\frac{1}{c_n}\mathbb P(Y_{0:n}|X_{0:n}=x_{0:n}^\ell) =\frac{1}{c_n}\psi_n(x_n^\ell)\omega_{n-1}^\ell\qquad\hbox{for }n=0,1,2,3,.....,N\]
with the convention that $\omega_{-1}^\ell\equiv 1$, and $c_n$ is a normalizing constant 
\[c_n = \sum_{\ell=1}^P\psi_n(x_n^\ell)\omega_{n-1}^\ell.\]
Then the filtering expectation of an integrable function $g(X_N)$ can be consistently approximated with the weighted particles,
\[\sum_{\ell=1}^Pg(x_N^\ell)\omega_N^\ell=\frac{\frac{1}{P}\sum_{\ell=1}^Pg(x_N^\ell)\mathbb P(Y_{0:N}|X_{0:N}=x_{0:N}^\ell)}{\frac{1}{P}\sum_{\ell=1}^P\mathbb P(Y_{0:N}|X_{0:N}=x_{0:N}^\ell)}\]

\[\rightarrow \frac{\mathbb E\left[g(\tilde X_N)\mathbb P(Y_{0:N}|\tilde X_{0:N})\Bigg|\mathcal F_N^Y\right]}{\mathbb E\left[\mathbb P(Y_{0:N}|\tilde X_{0:N})\Bigg|\mathcal F_N^Y\right]}=\mathbb E[g(X_N)|\mathcal F_N^Y]\]
almost surely as $P\rightarrow\infty$ by SLLN, where $\tilde X_{0:N}$ is a random variable with distribution $(p_0,\Lambda)$ and independent from $(X_{0:N},Y_{0:N})$.
\subsection{Sampling Importance Resampling (SIR)}
Our estimation of $\mathbb E[g(X_N)|\mathcal F_N^Y]$ becomes a particle filter when SIR is used along with SIS. SIR essentially invokes a bootstrap on the samples $\{x_{0:n}^\ell\}_\ell$ at time $n$. This procedure will reallocate our sampling resources onto particles that are more likely to be close to the true signal. When invoked, SIR does the following:
\begin{algorithm}
\textbf{SIR Bootstrap Procedure.}
\begin{algorithmic}
\FOR{$\ell =1,\dots P$}
\STATE sample a random variable $x_n^{\ell,sir}$ from $\{x_n^1,\dots,x_n^P\}$ according to $\{\omega_n^1,\dots,\omega_n^P\}$.
\ENDFOR
\STATE $\{x_n^1,\dots,x_n^P\}\gets\{ x_n^{1,sir},\dots,x_n^{P,sir}\}$.
\STATE $\{\omega_n^1,\dots,\omega_n^P\}\gets\{1/P,\dots,1/P\}$.
\end{algorithmic}
\end{algorithm}
The common criterion for invoking SIR can be related to an entropy approximation of the particle distribution. At any time $n$ prior to when SIR has been performed, the entropy is defined as
\[\mathcal E_n =-\sum_\ell\omega_n^\ell\log\omega_n^\ell\geq-\log\left(\sum_\ell (\omega_n^\ell)^2\right)>0\]
and so maximizing the entropy of the particle distribution is approximately the same as minimizing the sum of squared posterior weights. Therefore, the criterion is to invoke SIR whenever the number of important particles is less than some threshold $\delta\in[1,P]$:
\[\hbox{if }\qquad\frac{1}{\sum_\ell (\omega_n^\ell)^2}\leq\delta,\qquad\hbox{then invoke SIR}.\]

Even after SIR has been incorporated, our approximation is still consistent with the nonlinear filter:
\begin{theorem}
\label{thm:sirConvergence} For any bounded function $g(x)$,
\[\frac{1}{P}\sum_{\ell=1}^Pg(x_n^{\ell,sir})\rightarrow \mathbb E[g(X_n)|\mathcal F_n^Y]\]
in $L^2$ as $P\rightarrow\infty$, and $x_n^{\ell,sir}$ and $x_n^{\ell',sir}$ are asymptotically independent for any $\ell\neq \ell'$.
\end{theorem}
\begin{proof} (taken from section 9.2 of \cite{cappe2005}) Let $\{x_n^{\ell}\}_{\ell\leq P}$ be the set SIS samples that were in use prior to SIR. The post-SIR estimator can be written as a sum of the old samples:
$$ \frac{1}{P}\sum_\ell g\left(x_n^{\ell,sir}\right)=\frac{1}{P}\sum_\ell g\left(x_n^{\ell}\right)\cdot\tau_\ell$$
where $\tau_\ell$ is the number of times the $x_n^\ell$ was resampled, $\tau_\ell=\sum_{r=1}^P \mathbf{1}_{\{x_n^{r,sir} =x_n^{\ell}\}}$. Taking conditional expectations, we have $E[\tau_{\ell}|\mathcal F_n^Y\vee\{x_n^r\}_{r\leq P}]=P\cdot\omega_n^{\ell}$ and the conditional expectation of the estimator is
$$E\left[ \frac{1}{P}\sum_\ell g\left(x_n^{\ell,sir}\right)\bigg |\mathcal F_n^Y\vee\{x_n^{r}\}_{r\leq P}\right]=\frac{1}{P}\sum_{\ell}g\left(x_n^{\ell}\right)E[\tau_\ell|\mathcal F_n^Y\vee\{x_n^{r}\}_{r\leq P}]$$

$$=\sum_{\ell}g\left(x_n^{\ell}\right)\omega_n^\ell \stackrel{a.s}{\longrightarrow}\mathbb E\left[g(X_n)\bigg|\mathcal F_n^Y\right],~~~~~~~~~~~~~~~~\hbox{as }P\rightarrow\infty.$$
From here we take expectations of both sides and use dominated convergence to equate the limit to show $L^2$ convergence,
$$E \left|\frac{1}{P}\sum_{\ell=1}^Pg(x_n^{\ell,sir})-\mathbb E\left[g(X_n)\bigg|\mathcal F_n^Y\right]\right|^2\rightarrow 0$$
as $P\rightarrow\infty$. \\

Now consider another bounded function $f(x)$,
$$\mathbb E\left[g\left(x_n^{\ell,sir}\right)f\left( x_n^{\ell',sir}\right)\Big|\mathcal F_n^Y\right]=\mathbb E\left[\mathbb E\left[g\left(x_n^{\ell,sir}\right)f\left( x_n^{\ell',sir}\right)\Big|\mathcal F_n^Y\vee\{x_n^r\}_{r\leq P}\right]\Big|\mathcal F_n^Y\right]$$

$$=\mathbb E\left[\mathbb E\left[g\left( x_n^{\ell,sir}\right)\Big|\mathcal F_n^Y\vee\{x_n^r\}_{r\leq P}\right]\mathbb E\left[f\left( x_n^{\ell',sir}\right)\Big|\mathcal F_n^Y\vee\{x_n^r\}_{r\leq P}\right]\Big|\mathcal F_n^Y\right]$$

$$=\mathbb E\left[\left(\sum_{\ell}g\left(x_n^{\ell}\right)\omega_n^\ell\right)\left(\sum_{\ell}f\left(x_n^{\ell}\right)\omega_n^\ell\right)\Bigg| \mathcal F_n^Y\right]\rightarrow \mathbb E[g(X_n)|\mathcal F_n^Y]\cdot\mathbb E[f(X_n)|\mathcal F_n^Y],$$
as $P\rightarrow\infty$. So we've found that 
\[\mathbb E\left[g\left(x_n^{\ell,sir}\right)f\left( x_n^{\ell',sir}\right)\Big|\mathcal F_n^Y\right]\sim \mathbb E\left[g\left(x_n^{\ell,sir}\right)\Big|\mathcal F_n^Y\right] \mathbb E\left[ f\left( x_n^{\ell',sir}\right)\Big|\mathcal F_n^Y\right]\]
for $\ell\neq \ell'$ and $P$ large. Therefore, $x_n^{\ell,sir}$ and $x_n^{\ell',sir}$ are asymptotically independent.
\end{proof}
\subsubsection{Variance Reduction}
For any bounded function $g(x)$, the principle of conditional Monte Carlo tells us that SIR estimator will have greater variance than the SIS estimator,
\[var\left(\frac{1}{P}\sum_{\ell'} g(x_n^{\ell',sir})\right)\]

\[ = var\left(\frac{1}{P}\sum_{\ell'} g(x_n^{\ell',sir})\Bigg|\{x_n^\ell,\omega_n^\ell\}_\ell\right)+var\left(\mathbb E\left[\frac{1}{P}\sum_{\ell'} g(x_n^{\ell',sir})\Bigg|\{x_n^\ell,\omega_n^\ell\}_\ell\right]\right)\]

\[\geq var\left(\mathbb E\left[\frac{1}{P}\sum_{\ell'} g(x_n^{\ell',sir})\Bigg|\{x_n^\ell,\omega_n^\ell\}_\ell\right]\right)=var\left(\sum_\ell g(x_n^\ell)\omega_n^\ell\right)\]
and so the estimator $\frac{1}{P}\sum_\ell g(x_n^{\ell,sir})$ may not be preferable to $\sum_\ell g(x_n^\ell)\omega_n^\ell$. However, it follows from the proof of theorem \ref{thm:sirConvergence} that 
\[var(g(X_n)|\mathcal F_n^Y)\sim var(g(x_n^{\ell,sir})|\mathcal F_n^Y),\qquad\hbox{for $P$ large,}\]
and we see a reduction in the overall variance of the particles if we write down the law of total variance,
\[var(g(x_n^\ell)) = var(g(X_n)) = var\left(g(X_n)\Big|\mathcal F_n^Y\right)+\underbrace{var\left(\mathbb E[g(X_n)|\mathcal F_n^Y]\right)}_{>0}\]

\[>var\left(g(X_n)\Big|\mathcal F_n^Y\right)\sim var(g(x_n^{\ell,sir})).\]
This reduction can be quite significant if $g(X_n)$ has a broad range. The rates are difficult to show, but by invoking SIR we obtain estimates of $\mathbb E[g(X_n)|\mathcal F_n^Y]$ that will converge faster as $P\rightarrow\infty$. This brief subsection has not attempted any proof; we have not computed any comparison of convergence rates.

\section{Examples}
In this section we present some examples to demonstrate the particle filter's uses.
\subsection{Particle Fiter for Heston Model}
Consider a Heston model with time-dependent coefficients 
\begin{eqnarray}
\nonumber
dY_t&=&\left(\mu-\frac{1}{2}X_t\right)dt+\sqrt{X_t}\left(\rho dB_t+\sqrt{1-\rho^2}dW_t\right)\\
\nonumber
dX_t &=&\nu(\bar X-X_t)dt+\gamma\sqrt{X_t}dB_t
\end{eqnarray}
where $Y_t$ is the log-price of an equity, $\sqrt{X_t}$ is the volatility, $\rho\in[-1,1]$ is the correlation parameter , and $(W_t,B_t)$ are a pair of independent Wiener processes. The observed log-prices on equities and indices is not available in continuum. Instead, there is a discrete set sequence $(t_n)_{n=0,1,2,...}$ consisting of times at which quotes on the equity or index are given,
\[Y_n\doteq Y_{t_n},\qquad\qquad\hbox{for }n=0,1,2,3,4,.......\]
We denote the time step between the $nth$ and $(n+1)th$ observations $\Delta t_n=t_{n+1}-t_n$. By considering the Stratonovich/It\^o integral  transform 

\[\int_{t_n}^{t_{n+1}}\gamma_s\sqrt{X_s}\circ dB_s=\frac{1}{2}\int_{t_n}^{t_{n+1}}\gamma_s^2ds+\int_{t_n}^{t_{n+1}}\gamma_s\sqrt{X_s}dB_s\]
and letting $X_n\doteq X_{t_n}$, we will find it useful to work with the following implicit discretization of the Stratonovich form of the Heston model,
\begin{eqnarray}
\nonumber
Y_{n+1}&=&Y_n+\left(\mu-\frac{1}{2}X_n\right)\Delta t_n+\sqrt{X_n}\left(\rho \Delta B_n+\sqrt{1-\rho^2}\Delta W_n\right)\\
\nonumber
X_{n+1}&=&X_n\left(1-\nu\Delta t_n\right)+\left(\nu \bar X-\frac{\gamma^2}{2}\right)\Delta t_n+\gamma\sqrt{X_{n+1}}\Delta B_n\qquad(*)
\end{eqnarray}
where $\Delta B_n$ and $\Delta W_n$ are increments of independent Wiener processes (i.e. $\Delta B_n\doteq B_{t_{n+1}}-B_{t_n}\sim N(0,\Delta t_n)$ and $\Delta W_n\doteq W_{t_{n+1}}-W_{t_n}\sim N(0,\Delta t_n)$). We take $\sqrt{X_{n+1}}$ to be the root of equation $(*)$ which can be obtained through the quadratic equation (see Alfonsi \cite{alfonsi}),

\[\sqrt{X_{n+1}}=\frac{1}{2}\left\{\gamma\Delta B_n\pm\sqrt{\gamma^2\Delta B_n^2+4D}\right\}\]

\[\hbox{where}\qquad D=(1-\nu\Delta t_n)X_n+\left(\nu \bar X-\frac{\gamma^2}{2}\right)\Delta t_n.\]
Provided that $\frac{1}{n}=\Delta t\leq\frac{1}{\nu}$ and $\gamma^2\leq 2\nu \bar X$, this implicit scheme is effective because it is mean-reverting and preserves positivity in $X^n$. With this scheme we can generate particles $\{x_n^\ell\}_{n,\ell}$ and approximate the nonlinear filter.

\subsection{Rao-Blackwellization}
Let $\theta_n$ be a hidden Markov chain with transition probabilities $\Lambda$, and let $X_n$ be another hidden Markov process given by the following recursion,
\[X_n = a(\theta_n)X_{n-1}+\sigma(\theta_n)B_n\]
with $B_n\sim iid N(0,1$ (to be clear, $B_n\perp \theta_n$), and Gaussian initial distribution $p_0(x)$. Let the observations process be defined discretely as
\[Y_n= h(\theta_n)X_n+\gamma(\theta_n)W_n\]
where $W_k$ are $iid N(0,1)$ (to be clear $B_n\perp W_n$ and $\theta_n$ independent of $W_n$), and $\gamma(\cdot)>0$. In this case we can use particles to marginalize $\theta_n$,
\begin{eqnarray}
\nonumber
\theta_0^\ell&\sim&p_0\\
\nonumber
\theta_n^\ell&\sim&\Lambda(~\cdot~|\theta_{n-1}^\ell)\qquad\hbox{for }n>0,
\end{eqnarray} 
and then for each particle we can compute the marginal Kalman filter,
\begin{eqnarray}
\nonumber
G_n^\ell&=&\frac{h(\theta_n^\ell)\Sigma_{n-1}^\ell}{ h^2(\theta_n^\ell)\left(a^2(\theta_n^\ell)\Sigma_{n-1}^\ell+\sigma^2(\theta_n^\ell)\right)+\gamma^2(\theta_n^\ell)}\\
\nonumber
\widehat X_n^\ell&=&a(\theta_n^\ell)\widehat X_{n-1}^\ell+G_n^\ell\left(Y_n-h(\theta_n^\ell)a(\theta_n^\ell)\widehat X_{n-1}^\ell\right)\\
\nonumber
\Sigma_n^\ell&=&\left(1-G_n^\ell h(\theta_n^\ell)\right)\left(a^2(\theta_n^\ell)\Sigma_{n-1}^\ell+\sigma^2(\theta_n^\ell)\right)
\end{eqnarray}
where we have defined $\widehat X_n^\ell \doteq \mathbb E[X_n|\mathcal F_n^Y\vee\{\theta_{0:n}^\ell\}]$ and $\Sigma_n^\ell \doteq \mathbb E[(X_n-\widehat X_n^\ell)^2|\theta_{0:n}^\ell]$ (see \cite{jazwinski,krylov} for more on the Kalman Filter). Conditioned on $\mathcal F_{n-1}^Y\vee\{\theta_{0:n}^\ell\}$, $Y_n$ is normal with mean and variance
\begin{eqnarray}
\nonumber
\mu_{n|n-1}^\ell& \doteq&\mathbb E[Y_n|\mathcal F_{n-1}^Y\vee\{\theta_{0:n}^\ell\}] = h(\theta_n^\ell)a(\theta_n^\ell)\widehat X_{n-1}^\ell\\
\nonumber
v_{n|n-1}^\ell&\doteq& var\left(Y_n\Big|\mathcal F_{n-1}^Y\vee\{\theta_{0:n}^\ell\}\right)=h^2(\theta_n^\ell)\left(a^2(\theta_n^\ell)\Sigma_{n-1}^\ell+\sigma^2(\theta_n^\ell)\right)+\gamma^2(\theta_n^\ell)
\end{eqnarray}
with the convention $\mu_{0:-1}^\ell  =h(\theta_0^\ell)a(\theta_0^\ell)\mathbb EX_0$ and $v_{0:-1} =  h^2(\theta_0^\ell)a^2(\theta_0^\ell)var(X_0)+\gamma^2(\theta_0^\ell)$. For any particle $\theta_{0:n}^\ell$, the unnormalized importance weights are updated as follows:
\[\tilde\omega_n^\ell \doteq\mathbb P(Y_{0:n}|\theta_{0:n}=\theta_{0:n}^\ell)=\frac{\exp\left\{-\frac{1}{2}\left(\frac{Y_n - \mu_{n|n-1}^\ell}{\sqrt{v_{n|n-1}^\ell}}\right)^2\right\}}{\sqrt{v_{n|n-1}^\ell}}\times\tilde \omega_{n-1}^\ell\]
with the convention that $\tilde\omega_{-1}^\ell \equiv 1$. For $P$-many particles, the unnormalized Rao-Blackwellized filter is an approximation to the unnormalized filter,
\[\phi_n^*[\theta]\doteq\frac{1}{P}\sum_\ell\theta_n^\ell\tilde\omega_n^\ell=\frac{1}{P}\sum_\ell\theta_n^\ell\mathbb P(Y_{0:n}|\theta_{0:n}=\theta_{0:n}^\ell)\approx \mathbb E[\tilde\theta_n\mathbb P(Y_{0:n}|\tilde\theta_{0:n})|\mathcal F_n^Y]\doteq\tilde{\mathbb E}[\theta_n|\mathcal F_n^Y]\]
where $\tilde\theta$ is a copy of $\theta$ that is independent of $(Y,X,\theta)$.\\

\noindent The Rao-Blackwell theorem says the following:
\begin{theorem}\textbf{Rao-Blackwell.} Given $Z_{0:n}$, let $\hat\beta $ be an estimator of a parameter $\beta$, and $T$ a sufficient statistic for $\beta$. Then the estimator $\hat\beta^*=\mathbb E[\hat \beta|T(Z_{0:n})]$ is at least as good in terms of MSE
\[\mathbb E(\hat\beta^*-\beta)^2\leq\mathbb E(\hat\beta-\beta)^2\]
for all $\beta$ in the parameter space.

\end{theorem}
\noindent Now, suppose that for each $\ell$ we generate particles $\{x_{0:n}^{\ell',\ell}\}_{\ell'\leq P'}$ instead of computing Kalman filters, and then using these particles we compute another estimator of the unnormalized filtering expectation,
\[\phi_n[\theta] \doteq\frac{1}{P}\sum_\ell \theta_n^\ell\underbrace{\frac{1}{P'} \sum_{\ell'}\mathbb P\left(Y_{0:n}\Big|X_{0:n}=x_{0:n}^{\ell',\ell},\theta_{0:n}=\theta_{0:n}^\ell\right)}_{\approx\mathbb P\left(Y_{0:n}\big|\theta_{0:n}=\theta_{0:n}^\ell\right)}.\]
But given $\theta_{0:n}^\ell$, the marginal Kalman filter has allowed us to compute the likelihood without approximation. Therefore, a sufficient statistic for $\tilde{\mathbb E}[\theta_n|\mathcal F_n^Y]$ is $\mathcal T_n\doteq(\theta_n^\ell,\tilde\omega_n^\ell)_{\ell\leq P}$, and we have

\[\tilde\omega_n^\ell=\mathbb P(Y_{0:n}|\theta_{0:n}=\theta_{0:n}^\ell)=\int \mathbb P(Y_{0:n}|X_{0:n}=x_{0:n},\theta_{0:n}=\theta_{0:n}^\ell)\mathbb P(X_{0:n}\in dx_{0:n}|\theta_{0:n}=\theta_{0:n}^\ell)\]

\[=\mathbb E\left[\mathbb P(Y_{0:n}|X_{0:n}=x_{0:n}^{\ell',\ell},\theta_{0:n}=\theta_{0:n}^\ell)\Bigg|\mathcal F_n^Y\vee\{\theta_{0:n}^\ell\}\right]\]
for all $\ell'\leq P'$, and from here it is easy to see that the expectation of $\phi_n[\theta]$ given $\mathcal F_n^Y\vee\mathcal T_n$ is the Rao-Blackwellized estimator,
\[\mathbb E\left[\phi_n[\theta]\Big|\mathcal F_n^Y\vee\mathcal T_n \right]=\mathbb E\left[\mathbb E\left[\phi_n[\theta]\Big|\mathcal F_n^Y\vee\{\theta_{0:n}^\ell\}_\ell\right]\Big|\mathcal F_n^Y\vee\mathcal T_n \right]=\frac{1}{P}\sum_\ell \theta_n^\ell\tilde\omega_n^\ell=\phi_n^*[\theta].\]
Therefore, by the Rao-Blackwell theorem we know that $\phi_n^*[\theta]$ has less or equal MSE to a particle filter computed without marginal Kalman filters,
\[\mathbb E\left[\left(\phi_n^*[\theta]-\tilde{\mathbb E}[\theta_n|\mathcal F_n^Y] \right)^2\Big|\mathcal F_n^Y\right]\leq\mathbb E\left[\left( \phi_n[\theta]-\tilde{\mathbb E}[\theta_n|\mathcal F_n^Y]\right)^2\Big|\mathcal F_n^Y\right],\]
with the advantage that the Rao-Blackwellized filter requires particles to be simulated across a domain of fewer dimensions.

\chapter{Stability, Lyapunov Exponents, and Ergodic Theory for Finite-State Filters}

\noindent A filter is said to be `stable' if it has the ability to asymptotically recover from an erroneous initial distribution. In other words, assuming that all parts of the HMM are estimated correctly, but with the exception of the initial distribution which is incorrect, a stable filter will `forget' the false assumptions as the initial data falls farther and farther into the past. It can be advantageous to work with a stable filter for a number of reasons, one being that parameter estimation algorithms for models with stable filters do not need to put as much emphasis on estimating the initial condition. Rates at which stability takes effect can also be estimated, as these rates are shown to be given by Lyapunov exponents or a spectral gap in the filter's generator, and while the exponents are generally not explicitly computable, there are ways to make estimates. Finally, it is possible that a stable filter for an ergodic state variable may also have an ergodic theorem, a property that may be also be useful for parameter estimation.
\section{Main Ideas and Their History}
Consider an HMM $(X_t,Y_t)$ where $X_t$ is a hidden Markov process taking values in a state-space $ S $ with initial distribution $\nu:\mathcal S\rightarrow [0,1]$, and where observations are made on the process $Y_t$ which we assume to be given by a function of $X_t$ plus a noise. A filtering measure can be computed using the initial condition $\nu$, and we denote it as
\[\pi_t^\nu(A) = \mathbb P^\nu(X_t\in A|\sigma\{Y_s:s\leq t\})\qquad\hbox{for all Borel sets }A\subset \mathcal S.\]
\begin{definition} Let $\tilde\nu$ be another probability measure on $\mathcal S$, and let $\pi_t^{\tilde\nu}$ denote the filter computed with $\tilde\nu$ as the initial distribution of $X$. The filter is said to be \textit{asymptotically stable} if
\[\lim_{t\rightarrow\infty}\mathbb E\left|\pi_t^\nu-\pi_t^{\tilde\nu}\right|_{TV}=0\]
where $|\cdot|_{TV}$ denotes the total-variation norm.\footnote{The total variation norm of the difference between two probability measure $p$ and $q$ is $|p-q|_{TV} = \sup\{|p(A)-q(A)|:A\in \mathcal B(\mathbb R)\}$ where $\mathcal B(\mathbb R)$ denotes the space of Borel-measurable subset of $\mathbb R$.}
\end{definition}
Intuitively, the filter should be stable if the model fits one of the following descriptions:
\begin{itemize}
\item the signal is ergodic
\item or the observations are sufficiently informative, making old information obsolete (e.g. very low noise and $h(\cdot)$ is one-to-one),

\end{itemize}
but it is difficult to prove stability results even for these basic cases, and a general theory has yet to be developed. 

Stability results for finite-state Markov chain signals are well-known along their rates, and their ergodic theory, and it has been known since the 1960's that Kalman and Kalman-Bucy filters are stable when the signal is ergodic. The fundamental way of showing filter stability for these filters is to identify the equilibrium of the posterior covariance from its Ricatti equations, use it to show that the gain matrix also approaches an equilibrium, and then verify that the dependence on the initial condition fades with time. Example \ref{ex:kalmanBucyStability} shows how this is done for general Kalman-Bucy filters with constant coefficients, thus showing that observations in linear Gaussian models are indeed sufficiently informative since the signal is not assumed to be ergodic.

\begin{example}\label{ex:kalmanBucyStability} \textbf{(Stability of Kalman-Bucy Filter with Constant Coefficients).} Consider the following linear system,
\begin{eqnarray}
\nonumber
dX_t&=&aX_tdt+\sigma dB_t\\
\nonumber
dY_t&=&hX_tdt+\gamma dW_t
\end{eqnarray}
where $B_t\perp W_t$. Applying the Kalman-Bucy filter, we have
\begin{eqnarray}
\nonumber
d\widehat X_t&=&\left(a-\frac{h^2}{\gamma^2}\Sigma_t\right)\widehat X_tdt+\frac{h}{\gamma^2}\Sigma_tdY_t\\
\nonumber
\frac{d}{dt}\Sigma_t&=&2a\Sigma_t-\frac{h^2}{\gamma^2}\Sigma_t^2+\sigma^2.
\end{eqnarray}
The solution to the Ricatti equation can be written explicitly as follows,
\[\Sigma_t=\frac{\alpha_1-K\alpha_2\exp\left\{\frac{h^2}{\gamma^2}(\alpha_2-\alpha_1)t\right\}}{1-K\exp\left\{\frac{h^2}{\gamma^2}(\alpha_2-\alpha_1)t\right\}}\]
where
\[\alpha_1 = h^{-2}\left(a\gamma^2-\gamma\sqrt{a^2\gamma^2+h^2\sigma^2}\right),\qquad\qquad\alpha_2 = h^{-2}\left(a\gamma^2+\gamma\sqrt{a^2\gamma^2+h^2\sigma^2}\right)\]
with $K = \frac{\Sigma_0-\alpha_1}{\Sigma_0-\alpha_2}$. Asymptotically, we have $\Sigma_t\sim\alpha_2$ and the filtering expectation is approximately,
\[\widehat X_t\sim \widehat X_0e^{-\beta t}+\frac{h\alpha_2}{\gamma^2}\int_0^te^{-\beta( t-s)}dY_s\]
where $\beta = \frac{1}{\gamma}\sqrt{a^2\gamma^2+h^2\sigma^2}$. This shows that the Kalman-Bucy filter will forget any initial condition on $X_0$ as $t\rightarrow\infty$, thus showing that the filter is stable.
\end{example}

General results for nonlinear filters with ergodic states were identified by Kunita in 1971 \cite{kunita1971}, but a key step in his proof is wrong. His proof essentially said the following:\\

\noindent Assume that $\nu\ll\tilde\nu$. Then for any test function $g(x)$,

\[\int_\mathcal S g(x)d\pi_t^\nu(x)=\mathbb E[g(X_t)|\mathcal F_t^Y]=\frac{\mathbb E^{\tilde\nu}\left[g(X_t)\frac{d\nu}{d\tilde \nu}(X_0)\Big|\mathcal F_t^Y\right]}{\mathbb E^{\tilde\nu}\left[\frac{d\nu}{d\tilde \nu}(X_0)\Big|\mathcal F_t^Y\right]}\]

\[=\mathbb E^{\tilde\nu}\left[g(X_t)\frac{\mathbb E^{\tilde\nu}\left[\frac{d\nu}{d\tilde \nu}(X_0)\Big|\mathcal F_t^Y\vee\{X_t\}\right]}{\mathbb E^{\tilde\nu}\left[\frac{d\nu}{d\tilde \nu}(X_0)\Big|\mathcal F_t^Y\right]}\Bigg|\mathcal F_t^Y\right]=\int_\mathcal Sg(x)\frac{\mathbb E^{\tilde\nu}\left[\frac{d\nu}{d\tilde \nu}(X_0)\Big|\mathcal F_t^Y\vee\{X_t\}\right]}{\mathbb E^{\tilde\nu}\left[\frac{d\nu}{d\tilde \nu}(X_0)\Big|\mathcal F_t^Y\right]}d\pi_t^{\tilde\nu}(x)\]
where the denominator is strictly positive a.s. because $\nu\ll\tilde\nu$. From this we see that the $\pi^\nu\ll\pi^{\tilde\nu}$ with Radon-Nykodym derivative
\[\frac{d\pi_t^\nu}{d\pi_t^{\tilde\nu}}(x) =\frac{\mathbb E^{\tilde\nu}\left[\frac{d\nu}{d\tilde \nu}(X_0)\Big|\mathcal F_t^Y\vee\{X_t=x\}\right]}{\mathbb E^{\tilde\nu}\left[\frac{d\nu}{d\tilde \nu}(X_0)\Big|\mathcal F_t^Y\right]} \]
with $\mathbb P$-a.s. From the existence of the Radon-Nykodym derivative, the TV-norm is equivalent $\mathbb P$-a.s. to the following

\[\|\pi_t^\nu-\pi_t^{\tilde\nu}\|_{TV} =\int_\mathcal S\left|\frac{d\pi_t^\nu}{d\pi_t^{\tilde\nu}}(x)-1\right|d\pi_t^{\tilde\nu}(x) \]

\begin{equation}
\label{eq:ratio}
= \frac{\mathbb E^{\tilde\nu}\left[\left|\mathbb E^{\tilde\nu}\left[\frac{d\nu}{d\tilde \nu}(X_0)\Big|\mathcal F_t^Y\vee\{X_t\}\right]-\mathbb E^{\tilde\nu}\left[\frac{d\nu}{d\tilde \nu}(X_0)\Big|\mathcal F_t^Y\right] \right|\Bigg|\mathcal F_t^Y\right]}{\mathbb E^{\tilde\nu}\left[\frac{d\nu}{d\tilde \nu}(X_0)\Big|\mathcal F_t^Y\right]} 
\end{equation}
but because of the Markov property, we realize that the distribution of $X_0$ given $\mathcal F_t^Y\vee\{X_t\}$ is the same regardless of whether information is added regarding the future. Therefore, we have
\[\mathbb E^{\tilde\nu}\left[\frac{d\nu}{d\tilde \nu}(X_0)\Big|\mathcal F_t^Y\vee\{X_t\}\right]=\mathbb E^{\tilde\nu}\left[\frac{d\nu}{d\tilde \nu}(X_0)\Big|\mathcal F_\infty^Y\vee\mathcal F_{[t,\infty)}^X\right]\]
where $\mathcal F_{[t,\infty)}^X$ denotes the tail-$\sigma$-field generated by $\{X_s:s\geq t\}$. Combining the tail-$\sigma$-field measure with the numerator in (\ref{eq:ratio}) we have

\[\mathbb E\|\pi_t^\nu-\pi_t^{\tilde\nu}\|_{TV} = \mathbb E^{\tilde\nu}\left[ \mathbb E^{\tilde\nu}\left[\frac{d\nu}{d\tilde \nu}(X_0)\Big|\mathcal F_t^Y\right]\|\pi_t^\nu-\pi_t^{\tilde\nu}\|_{TV} \right]\]

\[ =\mathbb E^{\tilde\nu}\left|\mathbb E^{\tilde\nu}\left[\frac{d\nu}{d\tilde \nu}(X_0)\Big|\mathcal F_\infty^Y\vee\mathcal F_{[t,\infty)}^X\right]-\mathbb E^{\tilde\nu}\left[\frac{d\nu}{d\tilde \nu}(X_0)\Big|\mathcal F_t^Y\right]  \right|\]
and taking limits we have
\[\lim_{t\rightarrow\infty}\mathbb E\|\pi_t^\nu-\pi_t^{\tilde\nu}\|_{TV} =\mathbb E^{\tilde\nu}\left| \mathbb E^{\tilde\nu}\left[\frac{d\nu}{d\tilde \nu}(X_0)\Big|\bigcap_{t\geq 0}\mathcal F_\infty^Y\vee\mathcal F_{[t,\infty)}^X\right]-\mathbb E^{\tilde\nu}\left[\frac{d\nu}{d\tilde \nu}(X_0)\Big|\mathcal F_t^Y\right]\right|\]
which suggests that the filters are stable if and only if
\begin{equation}
\label{eq:kunitaEnd}
\mathbb E^{\tilde\nu}\left[\frac{d\nu}{d\tilde \nu}(X_0)\Big|\bigcap_{t\geq 0}\mathcal F_\infty^Y\vee\mathcal F_{[t,\infty)}^X\right]=\mathbb E^{\tilde\nu}\left[\frac{d\nu}{d\tilde \nu}(X_0)\Big|\mathcal F_\infty^Y\right].
\end{equation}
To this point, every step is correct, but the error made by Kunita was in assuming that the limits of these filtrations were equal,
\begin{equation}
\label{eq:equalFiltrations}
\bigcap_{t\geq 0}\mathcal F_\infty^Y\vee\mathcal F_{[t,\infty)}^X\stackrel{?}{=}\mathcal F_\infty^Y,
\end{equation}
but there has since come a counter-example to equation (\ref{eq:equalFiltrations}).

\subsection{The Counter Example} Baxendale, Chigansky and Lipster \cite{baxChigLip2006} presented an example to demonstrate when (\ref{eq:equalFiltrations}) fails. Let $X_t$ be a Markov chain taking values in $\mathcal S=\{1,2,3,4\}$ with transition intensities
\[\Lambda = \left[
\begin{array}{cccc}
-1&1&0&0\\
0&-1&1&0\\
0&0&-1&1\\
1&0&0&-1
\end{array}\right].\]
Clearly, all state communicate, and $X$ is an ergodic Markov process with invariant measure $\mu = (1,1,1,1)/4$. Let $h(x) = \mathbf 1_{x=1}+\mathbf 1_{x=3}$, and consider the observations model
\[Y_t= h(X_t), \]
which is a degenerate noise model. The following lemma was proven in \cite{baxChigLip2006}, 
\begin{lemma} For this example, the limit of the filtrations in equation (\ref{eq:equalFiltrations}) is false,
\[\bigcap_{t\geq 0}\mathcal F_\infty^Y\vee\mathcal F_{[t,\infty)}^X\supsetneq \mathcal F_\infty^Y.\]
\end{lemma}
\begin{proof}
It suffices to show that $X_0$ is $\mathcal F_\infty^Y\vee\mathcal F_{[t,\infty)}^X$-measurable, but not measurable with respect to $\mathcal F_\infty^Y$. The Markov chain $X$ only admits cycles in the following order,
\[\dots\dots\{3\}\to\{4\}\to\{1\}\to\{2\}\to\dots,\]
and therefore we can recover $X_0$ given $\mathcal F_t^Y$ and $X_t$ for any $t>0$ (i.e. because we know $X_t$ we can look backwards and deduce the path of $X$ by looking at what times $Y$ has jumped). Now, because $\mathcal F_t^Y\vee\{X_t\}\subset\mathcal F_\infty^Y\vee\mathcal F_{[t,\infty)}^X$, we have
\[X_0\in \bigcap_{t\geq 0}\mathcal F_\infty^Y\vee\mathcal F_{[t,\infty)}^X.\]

Next, denote the times at which $Y$ jumps with the sequence $\{\tau_i\}_{i\geq 1}$ ($\tau_i$ is time of $Y$'s $ith$ jump). It is not hard to verify that $\tau_i$ is independent of $(X_0,Y_0)$ and the the following filtrations are equal,
\[\mathcal F_t^Y = \bigvee_{i\geq 1}\{\tau_i\leq t\}\vee\{Y_0\}\]
and so for any $t>0$ we have
\[\mathbb P(X_0=1|\mathcal F_t^Y) = \mathbb P\left(X_t=1\Bigg|\bigvee_{i\geq 1}\{\tau_i\leq t\}\vee\{Y_0\}\right)=\mathbb P(X_0=1|Y_0)\]

\[=\frac{\mathbb P(X_0=1)}{\mathbb P(X_0=1)+\mathbb P(X_0=3)}Y_0\neq\mathbf 1_{X_t=1}.\]
Since this posterior holds for any $t>0$, we must have
\[\mathbb P(X_0=1|\mathcal F_\infty^Y) \neq\mathbf 1_{X_t=1}\]
which means that $X_0\notin\mathcal F_\infty^Y$.
\end{proof}
\section{Stability for Markov Chain Models}
In this section we present some of the results in the paper by Atar and Zeitouni \cite{atarZeitouni1997}. In particular, we present their proof of the stability-rate for discrete-time filtering problems where the state variable is an ergodic finite-state Markov chain. Their paper also presents the analogous results for continuous time, as well as some other results regarding the low-noise case.

Consider a probability space $(\Omega,\mathcal F,\mathbb P)$ and let $n=0,1,2,3\dots$ denote time. Suppose that $X_n$ is an unobserved Markov chain taking values in a finite state-space $\mathcal S=\{x_1,\dots,x_d\}$. Let the matrix $\Lambda$ contain $X_n$'s transition probabilities so that
\[\mathbb P(X_{n+1}=x_i) = \sum_{j}\Lambda_{ji}\mathbb P(X_n=x_j)\]
for any $i,j\leq d$, and $\mathbb P(X_0=x_i) = \nu_i$. Suppose further that $X_n$ is recurrent with invariant law $\mu$ such that
\[(\Lambda^*)^n\nu\rightarrow \mu\]
as $n\rightarrow\infty$. We assume that $X_n$ is ergodic, which can be the case if and only if $\Lambda$ is of primitive order $k$ (i.e. there exists $k<\infty $ such that $\Lambda_{ji}^n>0$ for all $i,j\leq d$ and for all $n\geq k$). Let the observed process $Y_n$ be a nonlinear function of $X_n$,
\[Y_n = h(X_n)+W_n\]
where $W_n$ is an iid Gaussian random variable with mean zero and variance $\gamma^2>0$. The filtering mass computed with $\mu$ as its initial condition is denoted with $\pi_n^\nu$ and is given recursively by
\[\pi_{n+1}^\nu = \frac{1}{c_{n+1}}\psi_{n+1}\Lambda^*\pi_n^\nu\]
where $c_{n+1}$ is a normalizing constant (dependent on $\mu$), and $\psi_{n+1}$ is a diagonal matrix of likelihood functions
\[\psi_{n+1} =\left[
\begin{array}{cccc}
 e^{-\frac{1}{2}\left(\frac{Y_{n+1}-h(x_1)}{\gamma}\right)^2}&0&\dots&0\\
  0&e^{-\frac{1}{2}\left(\frac{Y_{n+1}-h(x_2)}{\gamma}\right)^2}&\dots&0\\
\vdots&\vdots&\ddots&\vdots\\
       0&0&\dots&e^{-\frac{1}{2}\left(\frac{Y_{n+1}-h(x_d)}{\gamma}\right)^2}
 \end{array}\right].\]
 
Stability in this case means that for any other measure $\tilde\nu:\mathcal S\rightarrow [0,1]$, we have
 \[\|\pi_n^\nu-\pi_n^{\tilde\nu}\|\rightarrow 0\]
 as $n\rightarrow \infty$, where $\|\cdot\|$ denotes the Euclidean norm on $\mathbb R^d$. The rate of convergence of this difference is described in terms of \textbf{Lyapunov exponents}
 \[\mathcal E_\gamma(\nu,\tilde\nu,\omega) = \lim\sup_{n\rightarrow\infty}\frac{1}{n}\log\|\pi_n^\nu-\pi_n^{\tilde\nu}\|,\]
for all $\omega\in\Omega$, where the limit holds in some strong sense (such as probability, mean-square, or almost-surely). However, it turns out that the Lyapunov exponent is almost-surely bounded by a deterministic constant, dependent only on model parameters such as $\gamma$.

Before moving on, we define the sequence of stochastic operators $T_n\doteq \psi_n\Lambda^*$, and define the sequence unnormalized posterior distributions,
 \begin{equation}
 \label{eq:unnormalizedFilter}
 p_n^\nu \doteq  T_np_{n-1}^\nu=T_nT_{n-1}\dots T_1\nu.
 \end{equation}
Clearly, $\pi_n^\nu = p_n^\nu/\left<p_n^\nu,\mathbf 1\right>$, where $\left<\cdot,\cdot\right>$ denotes the Euclidean inner-product on $\mathbb R^d$, and $\mathbf 1 = (1,1,1,\dots,1)^*$. 

\subsection{Perron-Frobenius \& Oseledec's Theorems}
In this section we introduce some general results from matrix theory. The theorems are their usage have heavy dependence on the algebraic concept of an`exterior product.' For now, we start with a basic theory that we can apply to the eigenvalues and eigenvectors of a Markov chain transition matrix:
\begin{theorem}\textbf{(Perron-Frobenius).} Assuming that the Markov chain transition matrix $\Lambda$ is irreducible and of primitive order. Then 1 is a simple eigenvalue of $\Lambda$ with all other eigenvalues have real-part with absolute value strictly less than one. Moreover, the unique right-eigenvector corresponding to $1$ can be multiplied by a constant to equal $\mu$, so that $\lim_n\Lambda^n = (\mu,\mu\dots,\mu)^*$. Furthermore, for any probability vector $\nu$ we have $(\Lambda^*)^n\nu\rightarrow \mu$ as $n\rightarrow\infty$.
\end{theorem} 
\begin{proof} (see Ethier and Kurtz \cite{ethierKurtz}).\end{proof}

In filtering, we apply a sequence of matrices $T_n$ which are not time-homogenous. Therefore, we need to consider a more general framework when considering eigenvalues and the space of eigenvectors for $T_n$. The multiplicative ergodic theorem, also known as Oseledec's theorem, will be useful, but before we present the theorem we need to define the following,

\begin{definition} An operator $C(x,n)$ where $x=(x_0,\dots,x_n)\in\mathcal S^{n+1}$, is a \textbf{cocycle} if 
\begin{itemize}
\item $C_0(x) = I_{d\times d}$ for all $x$,
\item $C_n(x) = C_{n-m}(x_m)C_m(x)$.
\end{itemize}
where $x_m = (x_m,\dots,x_n)\in\mathcal S^{n-m+1}$.

\end{definition}

Letting $M_n = T_nT_{n-1}\dots,T_1$ with the convention that $M_{-1}=I$, we see that $M_n$ is a cocycle. Now, we are ready for Oseledec's theorem:

\begin{theorem} (\textbf{Oseledec's Multiplicative Ergodic Theorem).} Suppose both $M_n$ and $M_n^{-1}$ are integrable for all $n$, $\mathbb EM_n+\mathbb EM_n^{-1}<\infty$ for all $n<\infty$. Then for each $\nu\in\mathbb R^d\setminus\{0\}$ we have
\[\mathcal V = \lim_n\frac{1}{n}\log\left( \frac{\|M_n\nu\|}{\|\nu\|}\right)\qquad\qquad\hbox{a.s}\]
exists and can take up to $d$-many values. For some $\ell\leq d$, if $\mathcal V_1>\mathcal V_2>\dots>\mathcal V_\ell$ are the $\ell$-many different (random) limits, then there exist (random) subspaces $\mathbb R^d\supsetneq S_\omega^1\supset S_\omega^2\supset\dots\supset S_\omega^\ell\supset\{0\}$ such that the limit is $\mathcal V_i$ if $\nu\notin \mathcal S_\omega^i$ for $i\leq \ell$.
\end{theorem}
\begin{proof} (see page 181 of \cite{carmona}).\end{proof}
In Oseledec's theorem, the exponentials $\exp(\mathcal V_1)>\exp(\mathcal V_2)>\dots>\exp(\mathcal V_\ell)$ are the eigenvalues of the matrix $\lim_n(M_n^*M_n)^{1/2n}$. If should be noted that when $\ell<d$, the space spanned by the non-generalized eigenvectors of $\lim_n(M_n^*M_n)^{1/2n}$ will be $\mathbb R^\ell\subsetneq\mathbb R^d$. 

A piece of matrix theory that will be useful is the exterior product. The exterior product or wedge product maps an two vectors to the parallelogram that they form. For more than two vectors, the exterior product corresponds the shape analogous to the parallelogram in the dimension equal to the number of vectors (for instance, the exterior product of three vectors is a parallelepiped). For any two vectors $a,b\in\mathbb R^d$, their \textbf{exterior product} is 
\[a\wedge b = \sum_{i,j}a_ib_j\left(e_i\wedge e_j\right)\]
where $(e_i)_i$ represents the canonical basis of $\mathbb R^d$. Some basic properties of the exterior product are

\begin{itemize}
\item $e_i\wedge e_j = -e_j\wedge e_i$
\item $a\wedge a=0$
\item $\|a\wedge b\|^2 = \|a\|^2\|b\|^2-\left<a,b\right>^2$
\end{itemize}
and with regard to the Lyapunov exponents from Oseledec's theorem, we have
\[\lim\sup_n\frac{1}{n}\log\left(\frac{\|M_n(a\wedge b)\|}{\|a\wedge b\|}\right)\leq \mathcal V_1+\mathcal V_2\qquad\hbox{a.s.}\]

\subsection{Lyapunov Exponents of the Filter}
Continuing to let $M_n = T_nT_{n-1}\dots,T_1$, we can now show that the rate of convergence is given by the spectral gap in $M_n$:
\begin{theorem} 
\label{thm:lyapExp}Assuming that $X_n$ is an ergodic Markov chain, there exists a deterministic function of $\gamma$, namely $\mathcal E_\gamma$, such that for any $\nu\neq\tilde\nu$, we have
\[\lim\sup_{n\rightarrow\infty}\frac{1}{n}\log\|\pi_n^\nu-\pi_n^{\tilde\nu}\|=\mathcal E_\gamma(\nu,\tilde\nu) \leq \mathcal E_\gamma,\]
$\mathbb P$-a.s. In particular, $\mathcal E_\gamma = \mathcal V_2-\mathcal V_1<0$, which is the spectral gap in the matrix $\lim_n((M_n)^*M_n)^{1/2n}$. 
\end{theorem}
\begin{proof}
From the triangle inequality, we can assume W.L.O.G. that the initial distribution of $X$ is its invariant distribution, and take $\tilde\nu=\mu$. In this case, the matrices $T_n$ possess a stationary law which is also ergodic. Moreover,
\[\mathbb E\log^+\|T_n\|\leq c\mathbb E\max_i\gamma^{-2}\left(Y_nh(x_i)-.5h^2(x_i)\right)^+<\infty.\]
Hence, we can apply Oseledec's theorem to conclude there exists a random subspace $S_\omega^1$ such that if $\nu\notin S_\omega^1$ then
\begin{equation}
\label{eq:singleLyap}
\frac{1}{n}\log\|p_n^\nu\|\rightarrow\mathcal V_1
\end{equation}
$\mathbb P$-a.s. In this setting, $\mathcal V_1>\mathcal V_2>\dots>\mathcal V_d$ are the Lyapunov exponents associated with the matrix $M_n$. It is well-known that $((M_n)^*M_n)^{1/2n}$ has a (random) limit a.s., the eigenvalues of which are $e^{\mathcal V_i}$. Note that $(M_n)^*M_n$ is a non-negative matrix, thus by Perron-Frobenius theorem the eigenvector associated with the highest eigenvalue of $(M_n)^*M_n$ has all entries real and non-negative. The last property thus holds for $(M_n^*M_n)^{1/2n}$ too, and hence for $\lim_n(M_n^*M_n)^{1/2n}$. Since $S_\omega^1$ must be orthogonal to the eigenvector associated with the highest eigenvalue of $\lim_n(M_n^*M_n)^{1/2n}$, if follows that $S_\omega^1$ cannot include any probability vector with all entries strictly positive. As cases where $\nu$ does not have all positive entries, notice that $p^\nu$ does for $n\geq k$ where $k$ was the constant such that $\Lambda_{ij}^n>0$ for $i,j\leq d$ when $n\geq k$. Thus, (\ref{eq:singleLyap}) holds for any probability measure $\nu:\mathcal S\rightarrow[0,1]$.

Using Oseledec's theorem again, this time for $\mathbb R^d\wedge\mathbb R^d$-valued process $p_n^\mu\wedge p_n^\nu$, there exists a (random) strict subspace $S_\omega^2\subset\mathbb R^d\wedge\mathbb R^d$ such that is $\mu\wedge\nu\notin  S_\omega^2$ then
\begin{equation}
\label{eq:doubleLyap}
\frac{1}{n}\log\|p_n^\mu\wedge p_n^\nu\|\rightarrow \mathcal V_1+\mathcal V_2
\end{equation}
$\mathbb P$-a.s., and for $p_n^\mu\wedge p_n^\nu\in S_\omega^2$ we have
\begin{equation}
\label{eq:doubleLyapEstimate}
\lim\sup_n\frac{1}{n}\log\|p_n^\mu\wedge p_n^\nu\|\leq\mathcal V_1+\mathcal V_2
\end{equation}
$\mathbb P$-a.s. 

Then using the inequality $\frac{1}{\sqrt d}|\sin(a,b)|\leq \|a-b\|\leq \sqrt d|\sin(a,b)|$ where $\sin(a,b)$ is the sine of the angle between vectors $a$ and $b$ (see lemma \ref{lem:sineInequality}), and the fact that 
\[\sin^2(a,b) = 1-\cos^2(a,b) = \frac{\|a\|^2\|b\|^2-\left<a,b\right>^2}{\|a\|^2\|b\|^2} = \frac{\|a\wedge b\|^2}{\|a\|^2\|b\|^2},\]
we can conclude that
\[\lim\sup_n\frac{1}{n}\log\|\pi_n^\mu-\pi_n^\nu\| = \lim\sup_n\frac{1}{n}\left(\log\|p_n^\mu\wedge p_n^\nu\|-\log\|p_n^\mu\|-\log\|p_n^\nu\|\right)\leq\mathcal V_1+\mathcal V_2-2\mathcal V_1\]

\[=\mathcal V_2-\mathcal V_1<0,\]
which completes the proof.
\end{proof}
The difference $\mathcal V_2-\mathcal V^1<0$ is a spectral gap and its negativity is sufficient for the stability of the filters. In their paper \cite{atarZeitouni1997}, Atar and Zeitouni proceed to prove that when $\Lambda_{ij}>0$ for all $i,j\leq d$, then there exists a constant $c$ such that
\[\mathcal E_\gamma\leq c<0\]
where $c$ does not depend on $h$ or $\gamma$. They go on to prove that
\[c\leq -2\min_{i\neq j}\sqrt{\Lambda_{ij}\Lambda_{ji}}.\]
They also prove the following bounds for low-noise models,
\begin{eqnarray}
\label{eq:lowNoiseUpper}
\lim\sup_{\gamma\searrow 0}\gamma^2\mathcal E_\gamma&\leq&-\frac{1}{2}\sum_{i=1}^d\mu_i\min_{i\neq j}(h(x_i)-h(x_j))^2\\
\label{eq:lowNoiseLower}
\lim\inf_{\gamma\searrow 0}\gamma^2\mathcal E_\gamma&\geq&-\frac{1}{2}\sum_{i=1}^d\mu_i\sum_{j=1}^d(h(x_i)-h(x_j))^2.
\end{eqnarray}

Finally, with regard to the ergodic theory, it was shown by Chigansky in 2006 \cite{chigansky2006} that the Markov-Feller process $(X_n,\pi_n)$ has a unique invariant measure $\mathcal M$, such that for any continuous $g$,
\[\lim_n\frac{1}{n}\sum_ng(X_n,\pi_n) = \sum_i\int g(x_i,u)\mathcal M(x_i,du) \]

\begin{equation}
\label{eq;ergodicTheory}
= \sum_i\int \mu_ig(x_i,u)\mathcal M_{\mu_i}(du) =\lim_n\mathbb Eg(X_n,\pi_n)
\end{equation}
where $\mathcal M_{\mu_i}$ is the $\mu_i$-marginal of $\mathcal M$. 

\subsection{Proof of $\frac{1}{\sqrt d}|\sin(a,b)|\leq \|a-b\|\leq \sqrt d|\sin(a,b)|$}
Let $\mathcal D\subset\mathcal R^d$ denote the set of $d$-dimensional distribution vectors. If $a\in\mathcal D$ then $a_i\geq 0$ for all $i\leq d$, and $\sum_ia_i=1$. Furthermore, we can easily verify with a Jensen inequality that $\frac{1}{d}\leq \|a\|^2\leq 1$ where $\|\cdot\|$ is the Euclidean norm on $\mathbb R^d$. 
An inequality that will be useful is presented in the following lemma:
\begin{lemma}
\label{lem:sineInequality}
For any $a,b\in\mathcal D$, we have
\begin{equation}
\label{eq:sineInequality}
\frac{1}{d}\sin^2(a,b)\leq \|a-b\|^2\leq d\sin^2(a,b)
\end{equation}
where $\sin(a,b)$ is the sine of the angle between vectors $a$ and $b$.
\end{lemma}

\begin{proof} For $a=b$, the lemma is trivial, so the proof will focus on the case when $a\neq b$. 

The set $\mathcal D$ can be defined by a hyperplane $\mathcal H$ which is a $d-1$ dimensional surface inscribed in the non-negative region of $\mathbb R^d$ whose distance from the origin is exactly unity under the $\ell^1$-norm. For any vector $x\in\mathbb R^d$, its distance to $\mathcal H$ is defined as
\[\|x-\mathcal H\|\doteq\inf_{a\in\mathcal D}\|x-a\|.\]
From Jensen's inequality, we know that $\|a\|^2\geq \frac{1}{d}$ for all $a\in\mathcal D$ with equality iff $a_i\equiv\frac{1}{d}$ for all $i\leq d$, so for $x=0$ we have
\[\|0-\mathcal H\|=\inf_{a\in\mathcal D}\|a\| =\|a_0\|= \frac{1}{\sqrt{d}}\]
where $a_0\doteq\frac{1}{d}(1,1,\dots,1)$.  From the law of sines, for any vectors $a,b\in\mathcal D$ with $a\neq b$, the vectors and the hyperplane's surface form a triangle, and so we have a law of sines
\begin{equation}
\label{eq:lawOfSines}
\frac{\sin(a,b)}{\|a-b\|} = \frac{\sin(a,b-a)}{\|b\|}=\frac{\sin(b,b-a)}{\|a\|}.
\end{equation}
From (\ref{eq:lawOfSines}) we easily obtain  

\[2\frac{\sin^2(a,b)}{\|a-b\|^2} =  \frac{\sin^2(a,b-a)}{\|b\|^2}+\frac{\sin^2(b,b-a)}{\|a\|^2}\leq 2d,\]
which shows that $\frac{1}{d}\sin^2(a,b)\leq \|a-b\|^2$ and proves the lower-bound in (\ref{eq:sineInequality}).

To get the upper-bound requires significantly more preparation. For any $a\in\mathcal D$, let $\sin(a,\mathcal H)$ denote the sine of $a$'s angle of incidence with $\mathcal H$. Obviously, $\sin(a_0,\mathcal H) = \sin(\pi/2) = 1$, and for any $a\neq a_0$ the most acute angle that $a$ can make with a vector parallel to the hyperplane is its angle of incidence, which is its angle with the vector $a-a_0$,

\[\sin(a,\mathcal H) = \sin(a,a-a_0)\qquad\hbox{for }a\neq a_0.\]
With the angle of incidence in mind, we observe the following inequality,
\begin{equation}
\label{eq:incInequal}
\sin(a,b-a) \geq \sin(a,\mathcal H)\geq \sin(e_i,\mathcal H)
\end{equation}
for any $e_i$. The first inequality in (\ref{eq:incInequal}) follows from the angle of incidence, and the second inequality is seen to be true if one notices that $\forall a\in\mathcal D$, the sine of its angle of incidence must be greater than or equal to that of $e_i$, because the hyperplane's surface is flat and the most acute angle of incidence is formed by a vector that stretches the farthest, which happens to be any one that touches a corner of $\mathcal H$. If we apply (\ref{eq:lawOfSines}) to the triangle formed from $a_0$ and any $e_i$, we obtain
\[\frac{\sin^2(e_i,\mathcal H)}{\|a_0\|^2}= \frac{\sin^2(a_0,\mathcal H)}{\|e_i\|^2} = 1\qquad\hbox{for all }i\leq d,\]
giving us $\sin^2(e_i,\mathcal H) = \|a_0\|^2 = \frac{1}{d}$. Using this along with (\ref{eq:incInequal}), we have
\[\frac{\sin^2(a,b-a)}{\|b\|^2}\geq\sin^2(a,b-a)\geq \sin^2(e_i,\mathcal H) = \frac{1}{d},\]
and using (\ref{eq:lawOfSines}) we have
\[\frac{\sin^2(a,b)}{\|a-b\|^2} = \frac{\sin^2(a,b-a)}{\|b\|^2}\geq \frac{1}{d}\]
which proves the upper-bound.

\end{proof}


\end{document}